\documentclass[11pt]{article}
\usepackage{fullpage}
\usepackage{fixltx2e}
\usepackage{amsmath,amsthm}
\usepackage{amssymb}
\usepackage{times}
\usepackage{mathrsfs}
\usepackage[latin1]{inputenc}
\usepackage{graphicx}
\usepackage{caption}
\usepackage[usenames,dvipsnames,svgnames,table]{xcolor}
\usepackage{subcaption}
\usepackage{framed}
\usepackage[ruled,vlined]{algorithm2e}
\usepackage{multirow}
\usepackage{tikz}
\usepackage{paralist}
\usetikzlibrary{arrows,automata,shapes,decorations,calc,matrix,decorations.pathmorphing}
\definecolor{dark-gray}{gray}{0.45}
\newcommand{\val}{\mbox{\rm val}}
\newcommand{\cost}{\mbox{\rm r}}

\newcommand{\SP}{\mbox{\rm SP}}

\newcommand{\RS}{R_S}

\DeclareRobustCommand*\circled[1]{\tikz[baseline=(char.base)]{ 
            \node[shape=circle,draw,inner sep=2pt] (char) {\text{#1}};}}
\makeatletter
\newcommand*{\textoverline}[1]{$\overline{\hbox{#1}}\m@th$}
\makeatother

\newcommand{\LimitReach}{\mathsf{LPre}}

\newcommand{\N}{\ensuremath{{\rm \mathbb N}}}
\newcommand{\R}{\ensuremath{{\rm \mathbb R}}}
\newcommand{\E}{\ensuremath{{\rm \mathbb E}}}

\newcommand{\EXP}{\mathsf{LimAvgPre}}
\newcommand{\ExpRew}{\mathsf{ExpRew}}

\newcommand{\act}{A}
\newcommand{\mov}{\Gamma}
\newcommand{\trans}{\delta}
\newcommand{\stra}{\sigma}
\newcommand{\bigstra}{\Sigma}
\newcommand\distr{{\mathcal D}}
\newcommand\pat{\omega}
\newcommand\pats{\Omega}
\newcommand{\game}{G}
\newcommand{\supp}{\mathrm{Supp}}
\newcommand{\dest}{\mathrm{Succ}}
\newcommand{\cala}{{\mathcal A}}
\newcommand{\LimInfAvg}{\mathsf{LimInfAvg}}
\newcommand{\LimSupAvg}{\mathsf{LimSupAvg}}
\def\set#1{\{ #1 \}}
\newcommand{\coBuchi}{\mathsf{coBuchi}}
\newcommand{\Safe}{\mathsf{Safe}}
\newcommand{\Reach}{\mathsf{Reach}}

\newcommand{\ov}{\overline}

\newcommand{\Mem}{\mathsf{Mem}}

\newcommand{\stay}{\mbox{\rm Stay}}
\newcommand{\cover}{\mbox{\rm Cover}}

\makeatletter

\begingroup \catcode `|=0 \catcode `[= 1
\catcode`]=2 \catcode `\{=12 \catcode `\}=12
\catcode`\\=12 |gdef|@xcomment#1\end{comment}[|end[comment]]
|endgroup

\def\@comment{\let\do\@makeother \dospecials\catcode`\^^M=10\def\par{}}

\def\begincomment{\@comment\@xcomment}

\makeatother

\newenvironment{comment}{\begincomment}{}

 \newtheorem{theorem}{Theorem}

 \newtheorem{lemma}[theorem]{Lemma}
 \newtheorem{remark}[theorem]{Remark}

\begin{document}

\title{The Value~1 Problem Under Finite-memory Strategies for Concurrent Mean-payoff Games}

\author{
Krishnendu Chatterjee (IST Austria) 
%\thanks{IST Austria. Email: {\tt krish.chat@ist.ac.at}}
\and Rasmus Ibsen-Jensen (IST Austria)
%\thanks{Department of Computer Science, Aarhus University, Denmark. E-mail:{\tt rij@cs.au.dk}
}
%%\institute{IST Austria}
%}

\date{}
\maketitle

\begin{abstract}
We consider concurrent mean-payoff games, a very well-studied class of two-player (player~1 vs player~2) 
zero-sum games on finite-state graphs where every transition is assigned a reward between~0 and~1, 
and the payoff function is the long-run average of the rewards. 
%%%with limit-average (or mean-payoff) function.   
%%where  
%%played on a graph 
%by two-players, where in each round the players simultaneously choose a move,
%and the current state along with the joint moves determine the successor state.
%We study the most fundamental objective for concurrent games, namely, 
%mean-payoff or limit-average objective, where a reward between~0 and~1 
%is associated to every transition, and the goal of player~1 is to maximize the 
%long-run average of the rewards, and the objective of player~2 is strictly 
%the opposite (i.e., the games are zero-sum).
The value is the maximal expected payoff that player~1 can guarantee against 
all strategies of player~2.
We consider the computation of the set of states with value~1 under 
finite-memory strategies for player~1, and our main results for the 
problem are as follows: 
(1)~we present a polynomial-time algorithm;
(2)~we show that whenever there is a finite-memory strategy,
there is a stationary strategy that does not need memory at all; and 
(3)~we present an optimal bound (which is double exponential) 
on the patience of stationary strategies 
(where patience of a distribution is the inverse of the smallest
positive probability and represents a complexity measure of a stationary 
strategy).
\end{abstract}

\clearpage
\setcounter{page}{1}

\section{Introduction}

\noindent{\bf Concurrent mean-payoff games.}
Concurrent mean-payoff games are played on finite-state graphs 
by two players (player~1 and player~2) for infinitely many rounds.
In each round, the players simultaneously choose moves (or actions), 
and the current state along with the two chosen moves determine a
probability distribution over the successor states. 
The outcome of the game (or a \emph{play}) is an infinite sequence of states
and action pairs.
Every transition is associated with a reward between~$0$ and~$1$, 
and the mean-payoff (or limit-average payoff) of a play is the limit-inferior 
(or limit-superior) average of the rewards of the play. 
Concurrent games were introduced in a seminal work of Shapley~\cite{Sha53},
where \emph{discounted} sum objectives (or games that halt with probability~1) 
were considered.
The generalization to concurrent games with mean-payoff objectives 
(or games that have zero stop probabilities) was presented by Gillette in~\cite{Gil57}.
The player-1 \emph{value} $\val(s)$ of the game at a state $s$ is the 
supremum value of the expectation that player~1 can guarantee for the 
mean-payoff objective against all strategies of player~2.
The games are zero-sum where the objective of player~2 is the opposite.

\smallskip\noindent{\bf Important previous results.}
Many celebrated results have been established for concurrent mean-payoff 
games and its sub-classes:
(1)~the existence of values (or determinacy or equivalence of switching of 
strategy quantifiers for the players as in von-Neumann's min-max theorem) for 
concurrent discounted games was established in~\cite{Sha53};
(2)~the result of Blackwell and Ferguson established existence of values for
the celebrated game of Big-Match~\cite{BF68} (the celebrated Big-Match example is 
from~\cite{Gil57})\footnote{note that even showing existence of a value for the 
specific Big-Match game was open for years, which shows the hardness of analysis of such 
games}; and
(3)~developing on the results of~\cite{BF68} and of Bewley and Kohlberg on 
Puisuex series~\cite{BK76} the existence of values for concurrent mean-payoff 
games was established in~\cite{MN81}.
The decision problem of whether the value $\val(s)$ is at least a rational 
constant $\lambda$ can be decided in PSPACE~\cite{CMH08,HKLMT11}; and 
the results of~\cite{HKLMT11} present an algorithm for approximation 
which is polynomial in the number of actions and double exponential in the size 
of the state space (hence if the number of states is constant then the value can be 
approximated in polynomial time).
Several special cases of concurrent mean-payoff games have been widely studied, for 
example,
(a)~concurrent reachability games~\cite{dAHK98} where reachability objectives are the 
very special case of mean-payoff objectives where reward zero is assigned to all transitions
other than a set of sink terminal states which are assigned reward~1;
(b)~turn-based deterministic mean-payoff games~\cite{EM79,ZP96}, where in each state at most 
one of the players have the choice of more than one action and the transition function is
deterministic; and
(c)~turn-based (stochastic) reachability games~\cite{Con92}. 
The decision problem of whether the value $\val(s)$ is at least a rational constant $\lambda$
is \emph{square-root sum} hard even for concurrent reachability 
games~\cite{EY06}, and 
even for the special case of turn-based stochastic reachability games~\cite{Con92}
or turn-based deterministic mean-payoff games~\cite{ZP96} the existence of a 
polynomial-time algorithm is a major and long-standing open problem.

\smallskip\noindent{\bf Value~1 problem and its potential significance.}
While the decision problem for value computation is notoriously hard 
for concurrent mean-payoff games, 
an important special case of the problem is to compute the set of 
states with value~1.
We refer to this problem as the value-1 set computation problem.
We discuss the potential significance of the value~1 problem for mean-payoff objectives.
It was shown in~\cite{CGHIKPS08} that reliability requirements can be specified as a mean-payoff
condition, where in every step a computation is done, and if the computation succeeds
a reward~1 is assigned, and if the computation might fail, then reward~0 is assigned.
The reliability is the long-run average reward. 
The value~1 problem asks whether there exists a strategy to ensure that reliability 
arbitrarily close to~1 can be achieved.
Note that this problem cannot naturally be modeled as a reachability objective.

\smallskip\noindent{\bf Strategies.}
A strategy in a concurrent game, considers the past 
history of the game (the finite sequence of states and actions played so 
far), and specifies a probability distribution over the next moves.
Thus a strategy requires memory to remember the past history of the 
game. 
A strategy is \emph{stationary} if it is independent of the past history
and only depends on the current state.  
The complexity of a stationary strategy is described by its \emph{patience}
which is the inverse of the minimum non-zero probability assigned to a move.
The notion of patience was introduced in~\cite{Eve57} and also studied in the
context of concurrent reachability games~\cite{HKM09,HIM11}.
A strategy is \emph{finite-memory} if the memory set used by the strategy is 
finite. 
Note that for implementability of a strategy (such as by an automata), we need 
a finite-memory strategy.

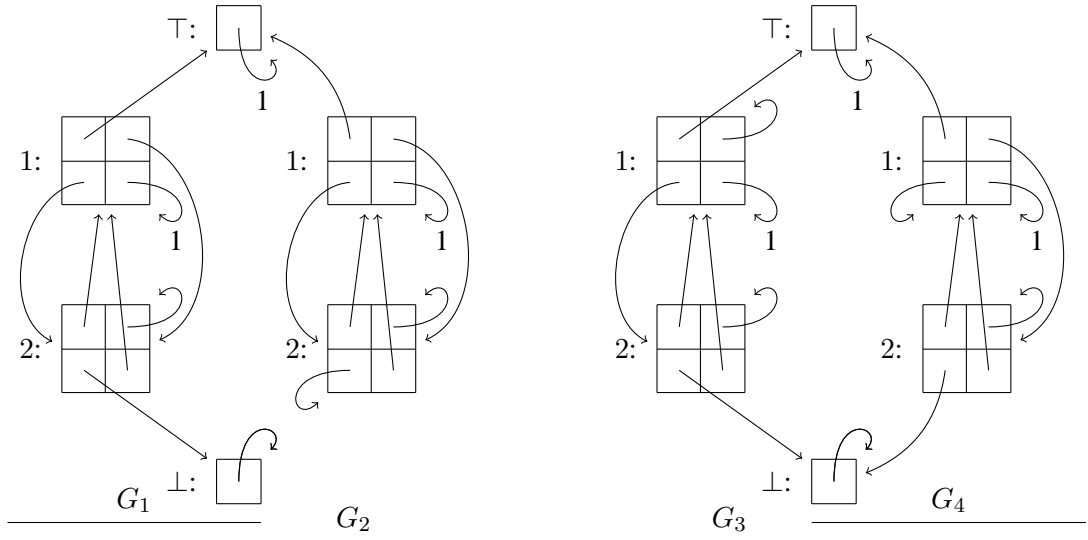
\begin{figure}
\centering
\begin{subfigure}[r]{0.45\textwidth}
\begin{tikzpicture}[node distance=2.5cm]

\matrix (s1) [label=left:$1$:,minimum height=1.5em,minimum width=1.5em,matrix of math nodes,nodes in empty cells, left delimiter={.},right delimiter={.}]
{
&\\
&\\
};
\draw[black] (s1-1-1.north west) -- (s1-1-2.north east);
\draw[black] (s1-1-1.south west) -- (s1-1-2.south east);
\draw[black] (s1-2-1.south west) -- (s1-2-2.south east);
\draw[black] (s1-1-1.north west) -- (s1-2-1.south west);
\draw[black] (s1-1-2.north west) -- (s1-2-2.south west);
\draw[black] (s1-1-2.north east) -- (s1-2-2.south east);
\matrix (s2) [label=left:$2$:,below of=s1,minimum height=1.5em,minimum width=1.5em,matrix of math nodes,nodes in empty cells, left delimiter={.},right delimiter={.}]
{
&\\
&\\
};
\draw[black] (s2-1-1.north west) -- (s2-1-2.north east);
\draw[black] (s2-1-1.south west) -- (s2-1-2.south east);
\draw[black] (s2-2-1.south west) -- (s2-2-2.south east);
\draw[black] (s2-1-1.north west) -- (s2-2-1.south west);
\draw[black] (s2-1-2.north west) -- (s2-2-2.south west);
\draw[black] (s2-1-2.north east) -- (s2-2-2.south east);
\matrix (v0) [label=left:$\bot$:,below right of=s2,minimum height=1.5em,minimum width=1.5em,matrix of math nodes,nodes in empty cells, left delimiter={.},right delimiter={.}]
{
\\
};
\draw[black] (v0-1-1.north west) -- (v0-1-1.north east);
\draw[black] (v0-1-1.south west) -- (v0-1-1.south east);
\draw[black] (v0-1-1.north west) -- (v0-1-1.south west);
\draw[black] (v0-1-1.north east) -- (v0-1-1.south east);

\matrix (v1) [label=left:$\top$:,above right of=s1,minimum height=1.5em,minimum width=1.5em,matrix of math nodes,nodes in empty cells, left delimiter={.},right delimiter={.}]
{
\\
};
\draw[black] (v1-1-1.north west) -- (v1-1-1.north east);
\draw[black] (v1-1-1.south west) -- (v1-1-1.south east);
\draw[black] (v1-1-1.north west) -- (v1-1-1.south west);
\draw[black] (v1-1-1.north east) -- (v1-1-1.south east);

\draw[->](s1-1-1.center) to (v1);
\draw[->](s1-1-2.center) to[bend left=75] (s2);
\draw[->](s1-2-1.center) to[bend right=75]  (s2);
\draw[->](s1-2-2.center) .. controls +(0:1) and +(-45:1.5) ..  node[near end,below] (x) {1} (s1);
\draw[->](s2-1-1.center) to  (s1);
\draw[->](s2-1-2.center) .. controls +(0:1) and +(45:1.5) ..   (s2);
\draw[->](s2-2-1.center) to  (v0);
\draw[->](s2-2-2.center) to  (s1);
\draw[->](v0-1-1.center) .. controls +(90:1) and +(45:1) ..   (v0);

\matrix (s3) [label=left:$1$:,below right of=v1,minimum height=1.5em,minimum width=1.5em,matrix of math nodes,nodes in empty cells, left delimiter={.},right delimiter={.}]
{
&\\
&\\
};
\draw[black] (s3-1-1.north west) -- (s3-1-2.north east);
\draw[black] (s3-1-1.south west) -- (s3-1-2.south east);
\draw[black] (s3-2-1.south west) -- (s3-2-2.south east);
\draw[black] (s3-1-1.north west) -- (s3-2-1.south west);
\draw[black] (s3-1-2.north west) -- (s3-2-2.south west);
\draw[black] (s3-1-2.north east) -- (s3-2-2.south east);
\matrix (s4) [label=left:$2$:,above right of=v0,minimum height=1.5em,minimum width=1.5em,matrix of math nodes,nodes in empty cells, left delimiter={.},right delimiter={.}]
{
&\\
&\\
};
\draw[black] (s4-1-1.north west) -- (s4-1-2.north east);
\draw[black] (s4-1-1.south west) -- (s4-1-2.south east);
\draw[black] (s4-2-1.south west) -- (s4-2-2.south east);
\draw[black] (s4-1-1.north west) -- (s4-2-1.south west);
\draw[black] (s4-1-2.north west) -- (s4-2-2.south west);
\draw[black] (s4-1-2.north east) -- (s4-2-2.south east);

\draw[->](v1-1-1.center) .. controls +(-90:1) and +(-45:1) .. node[below,midway] {1}   (v1);

\draw[->](s3-1-1.center) to[bend right] (v1);
\draw[->](s3-1-2.center) to[bend left=75] (s4);
\draw[->](s3-2-1.center) to[bend right=75] (s4);
\draw[->](s3-2-2.center) .. controls +(0:1) and +(-45:1.5) ..  node[near end,below] (x) {1} (s3);
\draw[->](s4-1-1.center) to  (s3);
\draw[->](s4-1-2.center) .. controls +(0:1) and +(45:1.5) ..   (s4);
\draw[->](s4-2-1.center) .. controls +(180:1) and +(225:1.5) ..   (s4);
\draw[->](s4-2-2.center) to  (s3);
\draw[->](v0-1-1.center) .. controls +(90:1) and +(45:1) ..   (v0);
\path let \p1 = (v0-1-1.south west) , \p2 =(s4-2-2.south east) in  (\x1,\y1-.5cm) edge node[above] {$G_2$} (\x2+1 cm,\y1-.5cm);
\path let \p1 = (v0-1-1.south east) , \p2 =(s2-2-2.south west) in  (\x1,\y1-.25cm) edge node[above] {$G_1$} (\x2-1.3 cm,\y1-.25cm);

\end{tikzpicture}
\end{subfigure}
\quad
\begin{subfigure}[r]{0.45\textwidth}
\begin{tikzpicture}[node distance=2.5cm]

\matrix (s1) [label=left:$1$:,minimum height=1.5em,minimum width=1.5em,matrix of math nodes,nodes in empty cells, left delimiter={.},right delimiter={.}]
{
&\\
&\\
};
\draw[black] (s1-1-1.north west) -- (s1-1-2.north east);
\draw[black] (s1-1-1.south west) -- (s1-1-2.south east);
\draw[black] (s1-2-1.south west) -- (s1-2-2.south east);
\draw[black] (s1-1-1.north west) -- (s1-2-1.south west);
\draw[black] (s1-1-2.north west) -- (s1-2-2.south west);
\draw[black] (s1-1-2.north east) -- (s1-2-2.south east);
\matrix (s2) [label=left:$2$:,below of=s1,minimum height=1.5em,minimum width=1.5em,matrix of math nodes,nodes in empty cells, left delimiter={.},right delimiter={.}]
{
&\\
&\\
};
\draw[black] (s2-1-1.north west) -- (s2-1-2.north east);
\draw[black] (s2-1-1.south west) -- (s2-1-2.south east);
\draw[black] (s2-2-1.south west) -- (s2-2-2.south east);
\draw[black] (s2-1-1.north west) -- (s2-2-1.south west);
\draw[black] (s2-1-2.north west) -- (s2-2-2.south west);
\draw[black] (s2-1-2.north east) -- (s2-2-2.south east);
\matrix (v0) [label=left:$\bot$:,below right of=s2,minimum height=1.5em,minimum width=1.5em,matrix of math nodes,nodes in empty cells, left delimiter={.},right delimiter={.}]
{
\\
};
\draw[black] (v0-1-1.north west) -- (v0-1-1.north east);
\draw[black] (v0-1-1.south west) -- (v0-1-1.south east);
\draw[black] (v0-1-1.north west) -- (v0-1-1.south west);
\draw[black] (v0-1-1.north east) -- (v0-1-1.south east);
\matrix (v1) [label=left:$\top$:,above right of=s1,minimum height=1.5em,minimum width=1.5em,matrix of math nodes,nodes in empty cells, left delimiter={.},right delimiter={.}]
{
\\
};
\draw[black] (v1-1-1.north west) -- (v1-1-1.north east);
\draw[black] (v1-1-1.south west) -- (v1-1-1.south east);
\draw[black] (v1-1-1.north west) -- (v1-1-1.south west);
\draw[black] (v1-1-1.north east) -- (v1-1-1.south east);

\draw[->](s1-1-1.center) to (v1);
\draw[->](s1-1-2.center) .. controls +(0:1) and +(45:1.5) ..  (s1);

\draw[->](s1-2-1.center) to[bend right=75]  (s2);
\draw[->](s1-2-2.center) .. controls +(0:1) and +(-45:1.5) ..  node[near end,below] (x) {1} (s1);
\draw[->](s2-1-1.center) to  (s1);
\draw[->](s2-1-2.center) .. controls +(0:1) and +(45:1.5) ..   (s2);
\draw[->](s2-2-1.center) to  (v0);
\draw[->](s2-2-2.center) to  (s1);
\draw[->](v0-1-1.center) .. controls +(90:1) and +(45:1) ..   (v0);

\matrix (s3) [label=left:$1$:,below right of=v1,minimum height=1.5em,minimum width=1.5em,matrix of math nodes,nodes in empty cells, left delimiter={.},right delimiter={.}]
{
&\\
&\\
};
\draw[black] (s3-1-1.north west) -- (s3-1-2.north east);
\draw[black] (s3-1-1.south west) -- (s3-1-2.south east);
\draw[black] (s3-2-1.south west) -- (s3-2-2.south east);
\draw[black] (s3-1-1.north west) -- (s3-2-1.south west);
\draw[black] (s3-1-2.north west) -- (s3-2-2.south west);
\draw[black] (s3-1-2.north east) -- (s3-2-2.south east);
\matrix (s4) [label=left:$2$:,below of=s3,minimum height=1.5em,minimum width=1.5em,matrix of math nodes,nodes in empty cells, left delimiter={.},right delimiter={.}]
{
&\\
&\\
};
\draw[black] (s4-1-1.north west) -- (s4-1-2.north east);
\draw[black] (s4-1-1.south west) -- (s4-1-2.south east);
\draw[black] (s4-2-1.south west) -- (s4-2-2.south east);
\draw[black] (s4-1-1.north west) -- (s4-2-1.south west);
\draw[black] (s4-1-2.north west) -- (s4-2-2.south west);
\draw[black] (s4-1-2.north east) -- (s4-2-2.south east);

\draw[->](v1-1-1.center) .. controls +(-90:1) and +(-45:1) .. node[below,midway] {1}   (v1);

\draw[->](s3-1-1.center) to[bend right] (v1);
\draw[->](s3-1-2.center) to[bend left=75]  (s4);
\draw[->](s3-2-1.center) .. controls +(180:1) and +(225:1.5) ..   (s3);
\draw[->](s3-2-2.center) .. controls +(0:1) and +(-45:1.5) ..  node[near end,below] (x) {1} (s3);
\draw[->](s4-1-1.center) to  (s3);
\draw[->](s4-1-2.center) .. controls +(0:1) and +(45:1.5) ..   (s4);
\draw[->](s4-2-1.center) to[bend left]  (v0);
\draw[->](s4-2-2.center) to  (s3);
\draw[->](v0-1-1.center) .. controls +(90:1) and +(45:1) ..   (v0);

\path let \p1 = (v0-1-1.south west) , \p2 =(s4-2-2.south east) in  (\x1,\y1-.25cm) edge node[above] {$G_4$} (\x2+1 cm,\y1-.25cm);
\path let \p1 = (v0-1-1.south east) , \p2 =(s2-2-2.south west) in  (\x1,\y1-.5cm) edge node[above] {$G_3$} (\x2-1.3 cm,\y1-.5cm);

\end{tikzpicture}
\end{subfigure}

\caption{The games $G_1$ to $G_4$}
\label{fig:g1-g4}
%%\vspace{-1em}
\end{figure}

\smallskip\noindent{\bf Examples.}
We now illustrate concurrent mean-payoff games with a few examples. 
Consider the four games ($G_1, G_2, G_3,$ and $G_4$) shown in Figure~\ref{fig:g1-g4}: the transition 
functions are deterministic and shown as arrows; and transition with rewards~1
are annotated, and all other rewards are~0.
Each game has four states, namely,~1,~2, $\top$ and $\bot$; and since 
$\top$ and $\bot$ remain the same, in the figures $G_1$ and $G_2$ 
(also $G_3$ and $G_4$) are drawn such that they share $\top$ and $\bot$.
The state $\top$ has value~1 and state $\bot$ has value~0.
In the first game $G_1$, both state~1 and state~2 have value $1/2$ (because of symmetry).
The other three example games, $G_2$, $G_3$ and $G_4$, are minor variants of $G_1$ 
(only one successor is changed). 
\begin{compactenum}
\item In $G_2$, the edge from state $2$ to $\bot$ is changed to a self-loop.
In $G_2$, there exists an infinite-memory strategy to ensure that the mean-payoff
is~1, and for every $\epsilon>0$ there is a stationary strategy to ensure
mean-payoff $1-\epsilon$.
The witness stationary strategy is as follows: in state~1 
play the action pairs with probability $(\epsilon/4, 1-\epsilon/4)$ 
and in state~2 with probability $(1/2,1/2)$.

\item In $G_3$, the top edge from state~1 to state~2 is changed to a self-loop.
In $G_3$, there is  no strategy to ensure that the mean-payoff
is~1, but for every $\epsilon>0$ there is a stationary strategy to ensure
mean-payoff $1-\epsilon$.
The witness stationary strategy is as follows: in state~1 
play the action pairs with probability $(\epsilon/2, 1-\epsilon/2)$ 
and in state~2 with probability $(1-\epsilon^2/2,\epsilon^2/2)$.

\item In $G_4$,  the bottom edge from state~1 to state~2 is changed to a self-loop.
In $G_4$, there exists no stationary %%(or even finite-memory) 
strategy that can ensure positive mean-payoff value; however, for every $\epsilon>0$ 
there exists an infinite-memory strategy to ensure mean-payoff $1-\epsilon$.
\end{compactenum}
Details regarding the analysis of the values of the above games and in depth
discussion on the strategy constructions for them are available in~\cite[Section~1.6.2]{RasmusThesis}.

\smallskip\noindent{\bf Our contributions.} Our main contributions are 
related to the computation of the value~1 problem for concurrent mean-payoff
games where player~1 is restricted to finite-memory strategies\footnote{note that once 
a finite-memory strategy for player~1 is fixed, then there always exists a finite-memory 
optimal counter-strategy for player~2, and thus the strategies for player~2 are not restricted}. 
Our main results are as follows:
(1)~We present a polynomial-time algorithm to compute the value~1 set.
(2)~We show that stationary strategies are sufficient, i.e., whenever 
finite-memory strategies exist, then there is a stationary strategy.
(3)~We establish an optimal double exponential patience bound for the 
witness stationary strategies (our contribution for patience is the upper
bound, and the matching lower bound follows from~\cite{HIM11,HKM09} for the 
special case of reachability objectives).
A key and novel insight of our polynomial-time algorithm is that we 
establish that we can use local operators and iterate them to compute the value~1 set; 
this is perhaps counter-intuitive for concurrent mean-payoff games as no strategy-iteration 
algorithm is known to exist. 
In addition we also establish a robustness result, which shows that for concurrent
games, if the support of the transition probabilities match (but the 
precise transition probabilities may differ), then the value~1 set 
remains unchanged.

\smallskip\noindent{\bf Related works.}
The problem of value-1 set computation has been extensively studied in many different 
contexts; such as, concurrent games with reachability objectives~\cite{dAHK98} as well as
with $\omega$-regular and prefix independent objectives~\cite{CdAH11,Cha07a,Cha14}, 
probabilistic automata~\cite{CT12,FGO12},
and probabilistic systems with counters~\cite{BBEKW10}.
However, the value-1 set computation was not considered for concurrent mean-payoff games
which we consider in this work.
A related problem of computing the set of states where there exists an optimal strategy that ensures 
mean-payoff~1 (almost-sure winning) has been considered 
in~\cite{CI14a}.

\newcommand{\dist}{\xi}
\newcommand{\C}{\mathcal{C}}
\newcommand{\Exp}{\mathbb{E}}
\newcommand{\vare}{\varepsilon}
\newcommand{\patc}{\mathsf{pat}}

\section{Definitions}

In this section we present the definitions of game structures, strategies, 
mean-payoff objectives, the value and value~1 problem, and other basic notions.

\smallskip\noindent{\bf Probability distributions.}
For a finite set~$A$, a {\em probability distribution\/} on $A$ is a
function $\trans\!:A\to[0,1]$ such that $\sum_{a \in A} \trans(a) = 1$.
We denote the set of probability distributions on $A$ by $\distr(A)$. 
Given a distribution $\trans \in \distr(A)$, we denote by $\supp(\trans) = 
\{x\in A \mid \trans(x) > 0\}$ the {\em support\/} of the distribution 
$\trans$.
For a distribution, the \emph{patience} of the distribution is the inverse of 
the minimum non-zero probability assigned to an element: formally,
the patience $\patc(\trans)$ is  
$\max_{a \in A} \set{ \frac{1}{\trans(a)} \mid \trans(a)>0}$.

\smallskip\noindent{\bf Concurrent game structures.} 
A (two-player) {\em concurrent stochastic game structure\/} 
$\game =  (S, \act,\mov_1, \mov_2, \trans)$ consists of the 
following components.

\begin{itemize}

\item A finite state space $S$ and a finite set $\act$ of actions (or moves).

\item Two move assignments $\mov_1, \mov_2 \!: S\to 2^{\act}
	\setminus \emptyset$.  For $i \in \{1,2\}$, assignment
	$\mov_i$ associates with each state $s \in S$ the non-empty
	set $\mov_i(s) \subseteq \act$ of moves available to player $i$
	at state $s$.  
	For technical convenience, we assume that $\mov_i(s) \cap \mov_j(t) = \emptyset$ 
	unless $i=j$ and $s=t$, for all $i,j \in \{1,2\}$ and $s,t \in S$. 
	If this assumption is not met, then the moves can be trivially renamed to satisfy the assumption.

\item A probabilistic transition function
	$\trans\!:S\times\act\times\act \to \distr(S)$, which
	associates with every state $s \in S$ and moves $a_1 \in
	\mov_1(s)$ and $a_2 \in \mov_2(s)$ a probability
	distribution $\trans(s,a_1,a_2) \in \distr(S)$ for the
	successor state.
\end{itemize}
For a set $Q \subseteq S$ of states we will denote by $\ov{Q} =S\setminus Q$ 
the complement of $Q$.
We will denote by $\trans_{\min}$ the minimum non-zero transition 
probability, i.e., $\trans_{\min}=\min_{s,t \in S} \min_{a_1 \in \mov_1(s),a_2\in \mov_2(s)}
\set{\trans(s,a_1,a_2)(t) \mid \trans(s,a_1,a_2)(t)>0}$.
We will denote by $n$ the number of states (i.e., $n=|S|$), and by 
$m$ the maximal number of actions available for a player at a state 
(i.e., $m=\max_{s\in S} \max\set{|\mov_1(s)|,|\mov_2(s)|}$). We will later define Markov chains as games where $m=1$. Since finding the mean-payoff of Markov chains can be done in polynomial time, we will only consider the case where $m\geq 2$.
For all states $s \in S$, moves $a_1 \in
\mov_1(s)$ and $a_2 \in \mov_2(s)$, let 
$\dest(s,a_1,a_2) = \supp(\trans(s,a_1,a_2))$ denote
the set of possible successors of $s$ when moves $a_1$ and $a_2$
are selected. 
The size of the transition relation of a game structure is defined as
$|\delta|=\sum_{s\in S}\sum_{a_1 \in \mov_1(s)} \sum_{a_2 \in \mov_2(s)} |\dest(s,a_1,a_2)|$.

\medskip\noindent{\bf One step probabilities.} 
Given a concurrent game structure  $\game$, a state $s$, 
two distributions $\dist_1 \in \distr(\mov_1(s))$ and $\dist_2 \in \distr(\mov_2(s))$,
the one step probability transition for a set $U$ of states, 
denoted as $\trans(s,\dist_1,\dist_2)(U)$ is 
$\sum_{a_1 \in \mov_1(s), a_2 \in \mov_2(s),t \in U} \trans(s,a_1,a_2)(t)\cdot \dist_1(a_1) \cdot \dist_2(a_2)$.
Often we will consider the distribution of player~2 to be a single action, i.e.,
$\dist_2(a_2)=1$ for an action $a_2$, and then use 
the notation $\trans(s,\dist_1,a_2)$.
We will also write $\dest(s,\dist_1,\dist_2) =
\bigcup_{a_1 \in \supp(\dist_1), a_2\in \supp(\dist_2)} \dest(s,a_1,a_2)$
for the set of possible successors under the distributions.

\smallskip\noindent{\bf Turn-based stochastic games, turn-based deterministic games and MDPs.} 
A game structure $\game$ is {\em turn-based stochastic\/} if at every
state at most one player can choose among multiple moves; that is, for
every state $s \in S$ there exists at most one $i \in \{1,2\}$ with
$|\mov_i(s)| > 1$. 
A turn-based stochastic game with a deterministic transition function is 
a turn-based deterministic game.
A game structure is a player-2 \emph{Markov decision process (MDP)} if for all 
$s \in S$ we have $|\mov_1(s)|=1$, i.e., only player~2 has choice of 
actions in the game, and player-1 MDPs are defined analogously.

\medskip\noindent{\bf Plays.}
At every state $s\in S$, player~1 chooses a move $a_1\in\mov_1(s)$,
and simultaneously and independently
player~2 chooses a move $a_2\in\mov_2(s)$.  
The game then proceeds to the successor state $t$ with probability
$\trans(s,a_1,a_2)(t)$, for all $t \in S$. 
A {\em path\/} or a {\em play\/} of $\game$ is an infinite sequence
$\pat =\big( (s_0,a^0_1, a^0_2), (s_1, a^1_1, a^1_2), (s_2,a_1^2,a_2^2)\ldots\big)$ of states and action pairs such that for all 
$k\ge 0$ we have (1)~$s_{k+1} \in \dest(s_k,a^k_1,a^k_2)$; and (2)~$a_1^k\in \mov_1(s_k)$; and (3)~$a_2^k\in \mov_2(s_k)$.
We denote by $\pats$ the set of all paths.

%%\medskip\noindent{\bf Size of a game.} The \emph{size} of a concurrent 
%game is the sum of the size of the state space and the number of the 
%entries of the transition function. Formally the size of a game is 
%$|S| + \sum_{s\in S,a \in \mov_1(s), b\in \mov_2(s)} |\dest(s,a,b)|$. 

\smallskip\noindent{\bf Strategies.}
A {\em strategy\/} for a player is a recipe that describes how to 
extend prefixes of a play.
Formally, a strategy for player $i\in\{1,2\}$ is a mapping 
$\stra_i\!:(S\times \act \times \act)^* \times S \to\distr(\act)$ that associates with every 
finite sequence $x \in (S\times \act \times \act)^*$  of state and action pairs, and the 
current state $s$ in $S$, representing the past history of the game, 
a probability distribution $\stra_i(x \cdot s)$ used to select
the next move. 
The strategy $\stra_i$ can prescribe only moves that are available to player~$i$;
that is, for all sequences $x\in (S \times \act \times \act)^*$ and states $s\in S$, we require that
$\supp(\stra_i(x\cdot s)) \subseteq \mov_i(s)$.  
We denote by $\bigstra_i$ the set of all strategies for player $i\in\{1,2\}$.
Once the starting state $s$ and the strategies $\stra_1$ and $\stra_2$
for the two players have been chosen, 
%%the game is reduced to an ordinary stochastic  process.
%Hence, 
the probabilities of events are uniquely defined~\cite{VardiP85}, where an {\em
event\/} $\cala\subseteq\pats$ is a measurable set of
paths.
For an event $\cala\subseteq\pats$, we denote by $\Pr_s^{\stra_1,\stra_2}(\cala)$ 
the probability that a path belongs to $\cala$ when the game starts from $s$ and 
the players use the strategies $\stra_1$ and~$\stra_2$.
We denote by $\Exp_s^{\stra_1,\stra_2}[\cdot]$ the associated expectation measure.
We will consider the following special classes of strategies:

\begin{enumerate}

\item \emph{Stationary (memoryless) and positional strategies.} 
A strategy $\sigma_i$ is \emph{stationary} (or memoryless) if it is independent of the history 
but only depends on the current state, i.e., for all $x,x'\in(S\times A\times A)^*$ and all $s\in S$, 
we have $\sigma_i(x\cdot s)=\sigma_i(x'\cdot s)$, and thus can be expressed as a function 
$\stra_i: S \to \distr(\act)$. 
For stationary strategies, the complexity of the strategy is described by the 
\emph{patience} of the strategy, which is the inverse of the minimum non-zero 
probability assigned to an action~\cite{Eve57}. 
Formally, for a stationary strategy  $\stra_i:S\to\distr(\act)$ for player~$i$, 
the patience is  $\max_{s \in S} \patc(\stra_i(s))$, where $\patc(\stra_i(s))$
is the patience of the distribution $\stra_i(s)$.
A strategy is \emph{pure (deterministic)} if it does not use randomization, i.e., for any history there is always
some unique action $a$ that is played with probability~1.
A pure stationary strategy $\stra_i$ is also called a {\em positional} strategy, and 
represented as a function $\stra_i: S \to \act$.
We denote by $\bigstra_i^S$ the set of stationary strategies for player~$i$.

\item \emph{Strategies with memory and finite-memory strategies.} 
A strategy $\stra_i$ can be equivalently defined as a pair of functions $(\stra_i^u,\stra_i^n)$,
along with a set $\Mem$ of memory states, such that 
(i)~the next move function $\stra_i^n: S \times \Mem \to \distr(\act)$ given the current state of 
the game and the current memory state specifies the probability distribution over the actions; and
(ii)~the memory update function $\stra_i^u: S \times \act \times \act \times \Mem \to \Mem$ given the current state of the game, the action pairs, 
and the current memory state updates the memory state. 
Any strategy can be expressed with an infinite set $\Mem$ of memory states, and a strategy is 
a {\em finite-memory} strategy if the set $\Mem$ of memory states is finite, otherwise it is an 
{\em infinite-memory} strategy.
We denote by $\bigstra_i^F$ the set of finite-memory strategies for player~$i$.

%\item \emph{Markov strategies.}
%A strategy $\stra_i$ is a \emph{Markov strategy} if it only depends on the length of
%the play and current state. 
%Formally, for all finite prefixes $x, x' \in (S\times \act \times \act)^*$ such that $|x|=|x'|$ (i.e., the 
%length of $x$ and $x'$ are the same, where the length of $x$ and $x'$ are the number of 
%states that appear in $x$ and $x'$, respectively) and all $s \in S$ 
%we have $\stra_i(x\cdot s)=\stra_i(x' \cdot s)$.

%\item \emph{Time-dependent memory.} 
%Consider a strategy $\stra_i$ with memory $\Mem$. 
%For every finite sequence $x \in (S \times \act \times \act)^*$ 
%there is an unique memory element $t(x)=m \in \Mem$ 
%such that after the finite sequence $x$ the current memory state is $m$
%(note that the memory update function is a deterministic function).
%For a time bound $T$, the time-dependent memory of the strategy $\stra_i$, 
%is the size of the set of memory elements used for histories upto length
%$T$, i.e., 
%$|\set{m \in \Mem \mid \exists x \in (S\times \act \times \act)^*, |x|\leq T, 
%t(x) =m }|$.
%Note that a Markov strategy can be played with time-dependent memory of 
%size $T$, for all $T>0$.
%A trivial upper bound on the time-dependent memory of a strategy is 
%$(|S|\cdot |A|\cdot |A|)^T$, for all $T \geq 0$.

\end{enumerate}

\smallskip\noindent{\bf Absorbing states.}
A state $s$ is {\em absorbing} if for all actions $a_1\in \Gamma_1(s)$ and all actions $a_2\in \Gamma_2(s)$ we have $\dest(s,a_1,a_2)=\{s\}$. In the present paper we will also require that $|\Gamma_1(s)|=|\Gamma_2(s)|=1$ if $s$ is absorbing.

\smallskip\noindent{\bf Objectives.} 
A quantitative objective $\Phi: \pats \to \R$ is a measurable 
function.
In this work we will consider \emph{limit-average} (or mean-payoff) objectives.
We will consider concurrent games with a reward function $\cost: S \times \act \times 
\act \to [0,1]$ that assigns a reward value $\cost(s,a_1,a_2)$ for all 
$s\in S$, $a_1 \in \mov_1(s)$ and $a_2 \in \mov_2(s)$.
%(see Remark~\ref{rem:rew} for 
%general real-valued reward functions\footnote{We consider boolean rewards for simplicity in presentation,
%and in Remark~\ref{rem:rew} we argue how the results extend to general rewards}.).
For a path $\pat= \big((s_0, a^0_1, a^0_2), (s_1, a^1_1,a^1_2), \ldots\big)$, 
the limit-inferior average (resp. limit-superior average) is defined as 
follows:
$\LimInfAvg(\pat)= \lim\inf_{n \to \infty} \frac{1}{n} \sum_{i=0}^{n-1} \cost(s_i,a^i_1,a^i_2)$
(resp. $\LimSupAvg(\pat)= \lim\sup_{n \to \infty} \frac{1}{n} \sum_{i=0}^{n-1} \cost(s_i,a^i_1,a^i_2)$).
%For a threshold $\lambda \in [0,1]$ we consider the following objectives:
%\begin{align*}
%&\LimInfAvg(\lambda) &= \set{\pat \mid \LimInfAvg(\pat) \geq \lambda}; 
%&\quad \LimSupAvg(\lambda) &= \set{\pat \mid \LimSupAvg(\pat) \geq \lambda};\\
%& \overline{\LimInfAvg}(\lambda) &= \set{\pat \mid \LimInfAvg(\pat) < \lambda}; 
%&\quad \overline{\LimSupAvg}(\lambda) &= \set{\pat \mid \LimSupAvg(\pat) < \lambda};\\
%& \overline{\LimInfAvg}_{\leq}(\lambda) &= \set{\pat \mid \LimInfAvg(\pat) \leq \lambda}; 
%&\quad \overline{\LimSupAvg}_{\leq}(\lambda) &= \set{\pat \mid \LimSupAvg(\pat) \leq \lambda}.\\
%\end{align*}
For the analysis of concurrent games with Boolean limit-average objectives 
(with rewards~0 and~1 only) we will also need \emph{reachability} and \emph{safety} objectives.
Given a target set $U\subseteq S$, the reachability objective $\Reach(U)$ requires 
some state in $U$ be visited at least once, i.e., defines the set
\[
\Reach(U)= \set{\pat=\big((s_0, a^0_1, a^0_2), (s_1, a^1_1,a^1_2), \ldots\big) 
\mid \exists i \geq 0.\ s_i \in U}
\] of paths.
The dual safety objective for a set $F\subseteq S$ of safe states requires that the 
set $F$ is never left, i.e., 
\[
\Safe(F)= \set{\pat=\big((s_0, a^0_1, a^0_2), (s_1, a^1_1,a^1_2), \ldots\big)
\mid \forall i \geq 0.\ s_i \in F}.
\]
We also consider the eventual safety objective, namely \emph{coB\"uchi} 
objective, that requires for a given set $F$ that ultimately only states 
in $F$ are visited, i.e.,
\[
\coBuchi(F)= \set{\pat=\big((s_0, a^0_1, a^0_2), (s_1, a^1_1,a^1_2), 
\ldots\big) \mid \exists j \geq 0. \forall i \geq j. \ s_i \in F}.
\]
Observe that reachability objectives are a very special case of 
Boolean reward limit-average objectives where states in $U$ 
are absorbing and are exactly the states with reward~1, and similarly
for safety objectives.

%%{\bf KRISH: Remark 0-1 rewards, to maximal rewards later.}

\smallskip\noindent{\bf Markov chains.}
A game structure $\game$ is a {\em Markov chain} if $m=1$.  We will in that case write $\trans(s)$ for the distribution $\trans(s,a_1,a_2)$, where $a_1$ is the unique action in $\Gamma_1(s)$ and $a_2$ is the unique action in $\Gamma_2(s)$. 
Markov chains defines a weighted graph $(S,E,w)$, where $(s,s')\in E$ iff $\trans(s)(s')>0$ and for all $(s,s')\in E$ we have that $w((s,s'))=\trans(s)(s')$. 
For an event $\cala\subseteq\pats$, we denote by $\Pr_s(\cala)$ the probability $\Pr_s^{\stra_1,\stra_2}(\cala)$, where $\stra_1$ and $\stra_2$ are the unique strategies for player~1  and player~2, respectively.
A state $s$ is reachable from another state $s'$ iff $s'$ is reachable from $s$ in $(S,E,w)$. A set of states $Z$ is reachable from a state $s$ iff a state in $Z$ is reachable from $s$.
For any set of states $Z$ in a Markov chain, let $\RS(Z)$, be the set of states from which $Z$ is not reachable. Clearly, $\RS(Z)\subseteq (S\setminus Z)$. A set of states $L$ is called a {\em recurrent class} if for each pair of states $s,s'\in L$ we have that $s'$ is reachable from $s$ and for each pair of states $s\in L$ and $s''\in (S\setminus L)$ we have that $s''$ is not reachable from $s$.
A \emph{recurrent class} in a Markov chain is a bottom scc (strongly connected component) in the graph of the Markov chain, where a bottom scc $L$ is an scc with no 
edges leaving the scc.

\smallskip\noindent{\bf Properties of Markov chains to be explicitly used in proofs.}
We will use several basic properties of Markov chains in our proof and we explicitly state them here.
Let us fix a Markov chain with state space $S$.
\begin{enumerate}
\item Given a set $Z \subseteq S$, for all $s \in S$, with probability~1 either $Z$ is visited infinitely often or $\RS(Z)$ is reached.
\label{markov property RSinf}

\item Given $Z \subseteq S$, for all $s\in S$, with probability~1 $\RS(Z)$ or $Z$ is reached, i.e.,
$\Pr_s(\Reach(\RS(Z)\cup Z))=1$.\label{markov property RS}

\item Given sets $Z\subseteq S$ and $Z' \subseteq S$, such that $Z$ can only be left from $(Z'\cap Z)$, 
then for all $s\in Z$ with probability~1 $(\RS(Z')\cap Z)$ or $(Z'\cap Z)$ is reached,
i.e., $\Pr_s(\Reach((\RS(Z')\cap Z)\cup (Z'\cap Z)))=1$.
Note the similarity with the previous property, only intersection with $Z$ is taken.
\label{markov property rs2new}

\item  Given sets $Z\subseteq S$ and $Z' \subseteq S$, such that $Z$ can only be left from $(Z'\cap Z)$ and 
from each state in $(Z'\cap Z)$ there is a positive probability to leave $Z$, then for all $s \in Z$ with 
probability~1  $(\RS(Z')\cap Z)$ or $(S\setminus Z)$ is reached, i.e., 
$\Pr_s(\Reach((\RS(Z')\cap Z)\cup (S\setminus Z)))=1$.\label{markov property rs2new2}

\item From every state $s\in S$, with probability~1 some recurrent class $L$ is reached; 
and given a recurrent class $L$ is reached, with probability~1 every state in $L$ is 
reached. \label{markov property recurrent class}

\item Consider $Z\subseteq S$ and $Z' \subseteq S$ such that for all $z\in Z$ the set $Z'$ is reachable. 
Then for all $s \in S$ with probability~1 either $\RS(Z)$ or $Z'$ is reached, i.e.,
$\Pr_s(\Reach(\RS(Z)\cup Z'))=1$.\label{markov property rs2}

\item Consider $Z\subseteq S$ and $Z' \subseteq S$ such that for all $s\in (S\setminus (Z\cup Z'))$, 
we have that $\trans(s)(Z)\cdot \epsilon\geq \trans(s)(Z')$, for $\epsilon>0$. 
Then, for all  $s\in (S\setminus (Z\cup Z'))$ the probability to reach $Z$ or $\RS(Z \cup Z')$ 
is at least $1-\epsilon$, i.e., $\Pr_s(\Reach(Z\cup \RS(Z\cup Z')))\geq 1-\epsilon$.
\label{markov property rs3}

\item Consider $Z\subseteq S$ and $Z' \subseteq S$ such that for all $s\in Z$ 
the set $Z'$ is reachable. 
Then for all $s \in Z$ with probability~1 $(S\setminus Z)$ or $Z'$ is reached, i.e., 
$\Pr_s(\Reach((S\setminus Z) \cup Z'))=1$. \label{markov property reachability}
\end{enumerate}
We will refer to these properties as Markov property~1 to Markov property~\ref{markov property reachability}, respectively. 
%%We argue about their correctness in the technical appendix (INSERT).

\smallskip\noindent{\bf $\mu$-calculus.} %%, complementation, and levels.} 
%\label{sec-levels}
Consider a $\mu$-calculus expression $\Psi = \mu X. \psi(X)$ over a
finite set $S$, where $\psi: 2^S \mapsto 2^S$ is monotonic.
The least fixpoint $\Psi = \mu X. \psi(X)$ is equal
to the limit $\lim_{k \to \infty} X_k$, where $X_0 = \emptyset$,
and $X_{k+1} = \psi(X_k)$.  
For every state $s \in \Psi$, we define the {\em level\/} $k \geq 0$
of $s$ to be the integer such that $s \not\in X_k$ and $s \in X_{k+1}$.  
The greatest fixpoint $\Psi = \nu X. \psi(X)$ is
equal to the limit $\lim_{k \to \infty} X_k$, where $X_0 = S$, and
$X_{k+1} = \psi(X_k)$.   
For every state $s \not\in \Psi$, we define the {\em level\/} $k \geq 0$ of
$s$ to be the integer such that $s \in X_k$ and $s \not\in X_{k+1}$.  
The {\em height\/} of a $\mu$-calculus expression 
$\gamma X. \psi(X)$, where $\gamma \in \set{\mu, \nu}$, 
is the least integer $h$ such that $X_h = \lim_{k \to \infty} X_k$.
An expression of height $h$ can be computed in $h+1$ iterations. 
%Given a $\mu$-calculus expression $\Psi=\gamma X. \psi(X)$, where 
%$\gamma \in \set{\mu, \nu}$, the complement $\neg \Psi = S \setminus \Psi$ of
%$\gamma$ is given by 
%$\overline{\gamma} X. \neg \psi (\neg X)$, 
%where $\overline{\gamma} = \mu$ if $\gamma = \nu$, and 
%$\overline{\gamma} = \nu$ if $\gamma= \mu$.
A $\mu$-calculus formula with nested $\mu$ and $\nu$ operators is a very 
succinct description of a nested iterative algorithm.

\smallskip\noindent{\bf Interpretation of $\mu$-calculus formula.}
Consider a $\mu$-calculus formula 
\[
\nu Y. \mu X. [f(Y,X)],
\]
where $f$ is pointwise monotonic. 
The intuitive way to read the formula is as $\nu Y. ( \mu X. [f(Y,X)] )$, i.e., 
given a value of $Y$ (say $Y_i$) we compute the inner least fixpoint with 
function $f(Y_i,X)$ which has only one free variable $X$. 
Thus for every $Y_i$, $\mu X. [f(Y_i,X)]$ assigns a value for $Y_i$.
In other words, the function $\mu X. [f(Y,X)]$ can be interpreted as a function 
$g(Y)$ on $Y$, and the outer fixpoint computes the greatest fixpoint of $g$.
The interpretation for computation of $\mu Y. \nu X. [f(Y,X)]$ is similar,
and is extended straightforwardly to more nested $\mu$-calculus formula.

\smallskip\noindent{\bf The value problem.}
Given an objective $\Phi$, and a class $\C$ of strategies for player~1,
the value for player~1 under the class $\C$ of strategies is the maximal
payoff that player~1 can guarantee with a strategy in class $\C$.
Formally,
$\val(\Phi,\C)(s) =\sup_{\sigma_1 \in \C} \inf_{\sigma_2 \in \bigstra_2} 
\Exp_s^{\sigma_1,\sigma_2}[\Phi]$.
In this work we will consider the computation of the \emph{value~1 set} 
under finite-memory strategies, i.e., the computation of the set 
$\set{s \in S \mid \val(\LimInfAvg(\cost),\bigstra_1^F)(s) =1}$.
Observe that to ensure value~1, player~1 must ensure that for all 
$\vare>0$, the probability to visit reward~1 is at least $1-\vare$,
and hence it follows if all rewards less than~1 are decreased to~0 
the value~1 set still remains the same, and hence for 
simplicity for the value~1 set computation we will consider
Boolean reward functions.

\newcommand{\algopred}{{\sc AlgoPred}}
\newcommand{\algo}{{\sc RankAlgo}}
\newcommand{\rank}{\mathsf{rk}}
\newcommand{\CC}{C}
\newcommand{\cale}{\mathcal{E}}

\section{The Value~1 Set Computation}
In this section we will present a polynomial-time algorithm to compute
the value~1 set, $\val_1(\Phi,\bigstra^F_1)$, for mean-payoff objectives 
$\Phi$. 
We start with a very basic and informal overview of the algorithm.

\smallskip\noindent{\bf Basic overview of the algorithm.}
The algorithm will compute the value~1 set $W$ by iteratively 
adding chunks of states that are guaranteed to be in the value~1
set, and the iteration will finally converge to $W$.
Let $U \subseteq W$ be the set of states that are already guaranteed
to be in the value~1 set (already identified as subset of $W$ in 
some previous iteration).
Then a new chunk $X$ of states are added such that $U \subseteq X \subseteq W$,
and the new chunk of states are also added iteratively (the algorithm is a nested 
iterative algorithm).
For the set $X$, let $U \subseteq Y \subseteq X$ be the subset that 
is already added, and then a new chunk $Y \subseteq Z \subseteq X$ is added 
such that player~1 can ensure that one of the following three conditions hold:
(1)~the probability to reach $U$ in one step can be made arbitrarily large as compared
to the probability to leave $W$ in one step (then $U$ can be reached with probability 
arbitrarily close to~1);
or (2)~the probability to stay in $X$ in one step is~1 and the probability to reach $Y$
in one step is positive (then $Y$ can be reached with probability~1);
or (3)~the probability to stay in $X$ in one step is~1, the one step expected reward
and the probability to stay in $Z$ in one step can be made arbitrarily close to~1.
Figure~\ref{fig:limit reach}, Figure~\ref{fig:almost reach}, and Figure~\ref{fig:get 1}
illustrate the above three conditions, respectively, pictorially.
Very informally, if always one of the the last two conditions is satisfied, then 
then the mean-payoff can be made arbitrarily close to~1; and the first condition 
ensures that the already computed value~1 set can be reached with probability arbitrarily 
close to~1. 
The initialization of the sets are as follows: $U$ and $Y$ are initialized to the empty 
set, and $W$, $X$, and $Z$ are initialized to the set of all states.
Note that the above three conditions are \emph{local} (one-step) 
conditions and we will first define an one-step predecessor operator to 
capture the above conditions.
We will then show how to compute the one-step predecessor operator 
in polynomial time, and finally show how to use the one-step predecessor
operator in a nested iterative algorithm to compute the value~1 set in
polynomial time.

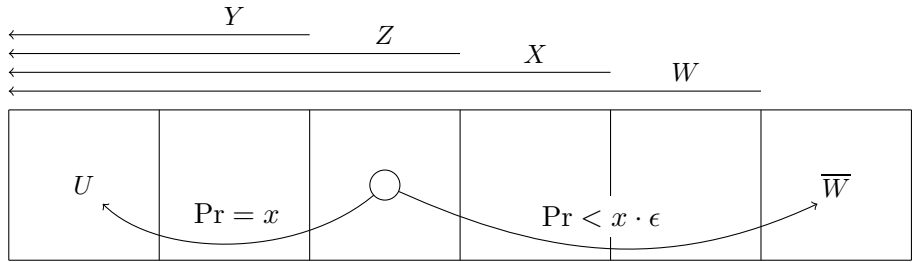
\begin{figure}
\begin{center}

\begin{tikzpicture}[node distance=3cm]
\tikzstyle{every state}=[fill=white,draw=black,text=black,font=\small , inner sep=0pt,minimum size=4mm]
\tikzstyle{txt}=[state,draw=white]
\foreach \x in {0,2,...,12} {\draw (\x,0) -- (\x,-2);  }
\draw (0,0) -- (12,0); 
\draw (12,-2) -- (0,-2);
\foreach \x / \y / \z in {$U$/1/-1,$Y$/3/1.25,$Z$/5/1,$X$/7/0.75,$W$/9/0.5,$\ov{W}$/11/-1} {\node[txt]  at (\y,\z) {\x};  }
\foreach \y / \z in {4/1,6/0.75,8/0.50,10/0.25} {\draw[<-] (0,\z) -- (\y,\z); }
\draw[<-] (1.25,-1.25).. controls (2,-2) and (4,-2) ..node[above=4pt] {$\Pr=x$} (5,-1);
\draw[<-,out=-155,in=-25] (10.75,-1.25) to node[above=4pt,fill=white] {$\Pr<x\cdot \epsilon$} (5,-1);
\node[state] (m1) at (5,-1){};

\end{tikzpicture}
\end{center}
\caption{Pictorial illustration of Equation~\ref{exp:limit reach}.}\label{fig:limit reach}
\end{figure}

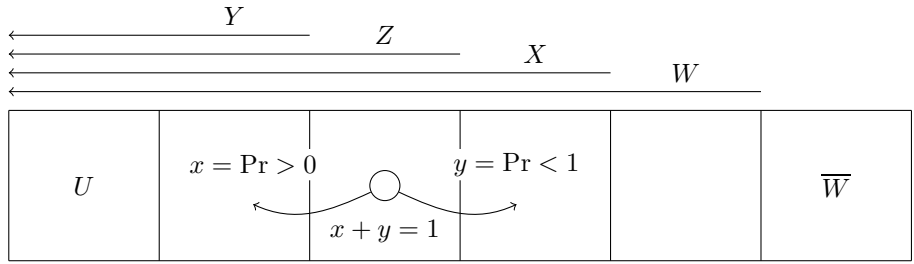
\begin{figure}
\begin{center}
\begin{tikzpicture}[node distance=3cm]
\tikzstyle{every state}=[fill=white,draw=black,text=black,font=\small , inner sep=0pt,minimum size=4mm]
\tikzstyle{txt}=[state,rectangle,draw=white]
\node[txt] (text) at (5,-1.6){$x+y=1$};
\foreach \x in {0,2,...,12} {\draw (\x,0) -- (\x,-2);  }
\draw (0,0) -- (12,0); 
\draw (12,-2) -- (0,-2);
\foreach \x / \y / \z in {$U$/1/-1,$Y$/3/1.25,$Z$/5/1,$X$/7/0.75,$W$/9/0.5,$\ov{W}$/11/-1} {\node[txt]  at (\y,\z) {\x};  }
\foreach \y / \z in {4/1,6/0.75,8/0.50,10/0.25} {\draw[<-] (0,\z) -- (\y,\z); }
\node[txt,yshift=15] (text2) at (3.25,-1.25)  {$x=\Pr>0$};
\node[txt,yshift=15] (text3) at (6.75,-1.25)  {$y=\Pr<1$};
\node[state] (m1) at (5,-1){};

\draw[<-,out=-25,in=-155] (3.25,-1.25) to  (m1);
\draw[<-,out=-155,in=-25] (6.75,-1.25) to (m1);

\end{tikzpicture}
\end{center}
\caption{Pictorial illustration of Equation~\ref{exp:almost reach}.}\label{fig:almost reach}
\end{figure}
\begin{figure}
\begin{center}
\begin{tikzpicture}[node distance=3cm]
\tikzstyle{every state}=[fill=white,draw=black,text=black,font=\small , inner sep=0pt,minimum size=4mm]
\tikzstyle{txt}=[state,draw=white]
\foreach \x in {0,2,...,12} {\draw (\x,0) -- (\x,-2);  }
\draw (0,0) -- (12,0); 
\draw (12,-2) -- (0,-2);
\foreach \x / \y / \z in {$U$/1/-1,$Y$/3/1.25,$Z$/5/1,$X$/7/0.75,$W$/9/0.5,$\ov{W}$/11/-1} {\node[txt]  at (\y,\z) {\x};  }
\foreach \y / \z in {4/1,6/0.75,8/0.50,10/0.25} {\draw[<-] (0,\z) -- (\y,\z); }
\draw[<-,out=-155,in=-25] (6.75,-1.25) to node[above right, very near start] {$\Pr\leq \epsilon$} (5,-1);
\node[state] (m1) at (5,-1){};
\path[<-] (m1) edge[loop left] node[above left,fill=white] {$\Pr\geq 1-\epsilon$} node[below left,fill=white] {$\ExpRew\geq 1-\epsilon$} (m1);

\end{tikzpicture}
\end{center}
\caption{Pictorial illustration of Equation~\ref{exp:get 1}.}\label{fig:get 1}
\end{figure}
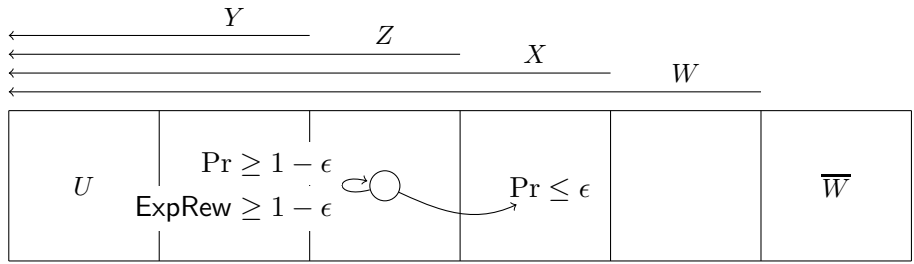

%%If we consider CMPG $G$ with only one state it is easy to see that one can simply play the stationary strategy for the matrix game with reward in entry $i,j$ being $\cost(i,j)$. This will in such games ensure the value precisely and that we can easily find the value using linear programming. We know from bla bla (insert ref to paper on patience of matrix games) that the patience is at most BLA to play optimally and from bla bla that the patience is at most BLA to play $\epsilon$-optimally. Note though that if we, like in the present paper, restrict us to games of where the value is equal to the maximal reward then any such matrix game has patience only 1 (since some action for player~1 must be played with probability greater than $\frac{1}{m}$ and therefore ensure against all actions of player~2 that maximal reward is achived).

\subsection{One-step predecessor operator\label{sec:one-step}}
We first formally define the one-step predecessor operator that was described 
informally in the basic overview of the algorithm. 
Given a state $s$ and two distributions $\dist_1 \in \distr(\mov_1(s))$ and 
$\dist_2 \in \distr(\mov_2(s))$, the expected one-step reward 
$\ExpRew(s,\dist_1,\dist_2)$ is defined as follows:
$\sum_{a_1 \in \mov_1(s), a_2 \in \mov_2(s)} \dist_1(a_1) \cdot 
\dist_2(a_2) \cdot \cost(s,a_1,a_2)$.
We often use distributions for player~2 that plays a single action $a_2$
with probability~1, and use $a_2$ to denote such a distribution.
For sets $U\subseteq Y\subseteq Z\subseteq X\subseteq W$, 
the one-step predecessor operator for limit-average (mean-payoff) objectives,
denoted as $\EXP(W,U,X,Y,Z)$, is the set of states $s$ such that for all $0<\epsilon<\frac{1}{2}$, 
there exists a distribution $\dist_1^{\epsilon}$ over $\mov_1(s)$ such that for all 
actions $a_2$ in $\mov_2(s)$, we have that
\begin{align}
&\ \big(\epsilon \cdot \trans(s,\dist_1^\epsilon,a_2)(U) >  \trans(s,\dist_1^\epsilon,a_2)(\overline{W})\big) & \label{exp:limit reach}\\[2ex] 
\vee &\ \big( \trans(s,\dist_1^\epsilon,a_2)(X)=1 \wedge  \trans(s,\dist_1^\epsilon,a_2)(Y) > 0 \big) & \label{exp:almost reach}\\[2ex]
\vee &\ \big( \trans(s,\dist_1^\epsilon,a_2)(X)=1 \wedge 
\ExpRew(s,\dist_1^{\epsilon},a_2) \geq 1-\epsilon   \wedge  \trans(s,\dist_1^\epsilon,a_2)(Z) \geq 1- \epsilon\big) \label{exp:get 1}\enspace .
\end{align}
We denote the above conditions as Equation~\ref{exp:limit reach},
Equation~\ref{exp:almost reach}, and Equation~\ref{exp:get 1}, 
respectively. 
Also our nested iterative algorithm (as informally described) 
that uses the $\EXP(W,U,X,Y,Z)$ operator will ensure the required 
inclusion $U\subseteq Y\subseteq Z\subseteq X\subseteq W$.
% will be guaranteed by by the nested iterated algorithm for the {\em RASMUS: I AM UNSURE WHAT YOU WROTE HERE}.
Before presenting the algorithm for the computation of the $\EXP$ set, we first
discuss the special case when we only have the first condition 
Equation~\ref{exp:limit reach}, then describe some key properties of witness 
distributions, and finally present an iterative algorithm to compute $\EXP$.

\smallskip\noindent{\bf The $\LimitReach$ operator and witness parametrized 
distribution.} 
An algorithm for the computation of the predecessor operator (called the $\LimitReach$
operator) for reachability games was presented in~\cite{dAHK98} where only Equation~\ref{exp:limit reach} 
is required to be satisfied.
We extend the results of~\cite{dAHK98,CdAH08} to obtain the following properties 
(details presented in technical appendix):

\begin{itemize}

\item \emph{(Input and output).} 
The algorithm takes as input a state $s$, two sets $U \subseteq W$ of states, 
two sets of action sets $A_1 \subseteq \mov_1(s)$ and $A_2 \subseteq \mov_2(s)$,
and either rejects the input or returns the largest set $A_3 \subseteq A_2$ 
such that the following conditions hold: for every $0<\epsilon<\frac{1}{2}$ there exists a 
witness distribution 
$\dist_1^\epsilon \in \distr(A_1)$, with patience at most 
$\left(\frac{\epsilon \cdot \trans_{\min}}{2}\right)^{-(|A_1|-1)}$,
such that 
(i)~for all actions $a_2 \in A_3$ Equation~\ref{exp:limit reach}
is satisfied; and (ii)~for all actions $a_2' \in (A_2 \setminus A_3)$ 
we have $\dest(s,\dist_1^\epsilon,a_2') \subseteq W$.
The set $A_3$ is largest in the sense that if $A_4 \subseteq A_2$ and 
$A_4$ satisfies the above conditions, then $A_4 \subseteq A_3$. Notice that this indicates that for all $a_2 \in (A_2 \setminus A_3)$ 
we have $\dest(s,\dist_1^\epsilon,a_2) \cap U=\emptyset$, because otherwise $a_2$ would be in $A_3$. 
Moreover, the distribution $\dist_1^\epsilon$ has the largest possible
support, i.e., for all actions $a_1 \in (A_1\setminus \supp(\dist_1^\epsilon))$,
there exists an action $a_2$ in $(A_2 \setminus A_3)$ such that
$\dest(s,a_1,a_2) \cap \ov{W} \neq \emptyset$.
An input would only be rejected if for each action $a_1\in A_1$ there exists an action $a_2\in A_2$ such that $\dest(s,a_1,a_2) \cap \ov{W} \neq \emptyset$.

\item \emph{(Parametrized distribution).}
Finally, the witness family of distributions $\dist_1^\epsilon$, for $0<\epsilon<\frac{1}{2}$, is presented in a parametrized fashion as follows: 
the support $\supp(\dist_1^\epsilon)$ for all $0<\epsilon<\frac{1}{2}$ is 
the same (denoted as $A^*$),
and the algorithm gives the support set $A^*$, and a ranking function 
that assigns a number from $0$ to at most $|A^*|$ to every action in $A^*$, 
and for any $0<\epsilon<\frac{1}{2}$, the witness distribution $\dist_1^\epsilon$ 
plays actions with rank $i$ with probability proportional to $\epsilon^i$.
In other words, the support set $A^*$ and the ranking number of the actions in 
$A^*$ is a polynomial witness for the parametrized family of witness 
distributions $\dist_1^\epsilon$, for all $0<\epsilon<\frac{1}{2}$.
\end{itemize}
We summarize the important properties which we explicitly use later:
$\LimitReach(s,W,U,A_1,A_2)$ for $U \subseteq W$ returns the following 
(see Technical Appendix for correctness proof):
\begin{enumerate}
\item \emph{(Reject property of $\LimitReach$).} 
Reject and then for all $a_1 \in A_1$ there exists $a_2 \in A_2$ such 
that $\dest(s,a_1,a_2) \cap \ov{W} \neq \emptyset$ 

\item \emph{(Accept properties of $\LimitReach$).}
Accepts and returns the set $A_3 \subseteq A_2$ and 
a parametrized distribution $\dist_1^\epsilon$, for $0<\epsilon<\frac{1}{2}$, with support $\supp(\dist_1^\epsilon)\subseteq A_1$, such that the following properties
hold:
\begin{itemize}
\item (Accept property~a). For all $a_2 \in A_3$, the distribution 
$\dist_1^\epsilon$ satisfies Equation~\ref{exp:limit reach} for $a_2$.

\item (Accept property~b). For all $a_2 \in (A_2 \setminus A_3)$, 
we have $\dest(s,\dist_1^\epsilon,a_2) \cap \ov{W}=\emptyset$ and 
$\dest(s,\dist_1^\epsilon,a_2) \cap U=\emptyset$.

\item (Accept property~c). For all 
$a_1 \in (A_1\setminus \supp(\dist_1^\epsilon))$, there exists an action $a_2$ 
in $(A_2 \setminus A_3)$ such that $\dest(s,a_1,a_2)\cap\ov{W}\neq \emptyset$.

\item (Accept property~d). The set $A_3$ is largest in the sense that for all
$a_2 \in (A_2 \setminus A_3)$ and for all parametrized distributions $\dist_1^{\epsilon}$ 
over $A_1$, the Equation~\ref{exp:limit reach} cannot be satisfied,
while satisfying actions in $A_2$ using Equation~\ref{exp:limit reach}, 
or Equation~\ref{exp:almost reach}, or Equation~\ref{exp:get 1}, for any $X,Y,Z$ 
such that $U \subseteq Y \subseteq Z \subseteq X \subseteq W$. 

\end{itemize}
 
\end{enumerate}

\smallskip\noindent{\bf One action with large probability property.}
We will now show that if a state belongs to $\EXP$, then there is a family
of witness distributions where one action $a$ is played with very large 
probability.

\begin{lemma}\label{lemm:one_action}
Given $U\subseteq Y\subseteq Z\subseteq X\subseteq W$, if 
$s \in \EXP(W,U,X,Y,Z)$, then for all $0<\epsilon \leq \frac{\delta_{\min}}{m}$ 
there is a witness distribution to satisfy at least one of the three conditions 
(Equation~\ref{exp:limit reach}, Equation~\ref{exp:almost reach}, 
or Equation~\ref{exp:get 1}) of $\EXP$ where an action $a \in \mov_1(s)$ is played with 
probability at least  $1-\epsilon\cdot \delta_{\min}$.
\end{lemma}
\begin{proof}
Given $0<\epsilon \leq \frac{\delta_{\min}}{m}$, let $\dist_1^\epsilon$ be a witness
distribution such that for all actions in $\mov_2(s)$ at least one of the three conditions 
for $\EXP$ is satisfied.
Let $C_1$ be the set of actions $a_2$ in $\mov_2(s)$ such that $\dist_1^\epsilon$ 
and $a_2$ satisfy Equation~\ref{exp:limit reach}; 
respectively, $C_2$ for Equation~\ref{exp:almost reach}, and $C_3$ for Equation~\ref{exp:get 1}.
Let $a$ be some action such that $\dist_1^\epsilon(a)\geq \frac{1}{m}$ 
(note that such an action must exist). 
If $\dist_1^\epsilon(a) \geq  1-\epsilon\cdot \delta_{\min}$, then 
we already have the desired action $a$; and we are done.
Otherwise, we consider the distribution $\dist'_1$ defined as follows:
\[
\dist_1'(a_1)=\begin{cases}
1-\epsilon\cdot \delta_{\min} & \text{if }a=a_1\\
\epsilon\cdot \delta_{\min}\cdot \frac{\dist_1^\epsilon(a_1)}{1-\dist_1^\epsilon(a)} & \text{otherwise} \enspace .
\end{cases}
\]
We now consider three cases to show $\dist_1'$ is also a witness distribution 
to satisfy at least one of the three conditions of $\EXP$ for $\epsilon$.
\begin{enumerate}
\item Consider an action $a_2$ in $C_1$. 
Since $a_2$ in $C_1$ and $\epsilon< \frac{\trans_{\min}}{m}$, we must have 
that $\dest(s,a,a_2)\cap \overline{W} \neq \emptyset$,
because otherwise given $\dist_1^\epsilon$ and $a_2$ the set $\ov{W}$ is reached
with probability at least $\frac{\trans_{\min}}{m}$ (as $a$ is played with 
probability at least $\frac{1}{m}$ by $\dist_1^\epsilon$), i.e.,
$\trans(s,\dist_1^\epsilon,a_2)(\ov{W}) \geq \frac{\trans_{\min}}{m}>\epsilon$.
This contradicts that $a_2$ satisfies Equation~\ref{exp:limit reach} 
for $\dist_1^\epsilon$ for the given $\epsilon < \frac{\trans_{\min}}{m}$.
Hence given $a$ and $a_2$, the probability to leave the set $W$ is~0;
and since all the other actions are only scaled in $\dist_1'$ as compared
to $\dist_1^\epsilon$ we have 
\[
\frac{\trans(s,\dist_1^\epsilon,a_2)(U)}{\trans(s,\dist_1^\epsilon,a_2)(\ov{W})}
\leq \frac{\trans(s,\dist_1',a_2)(U)}{\trans(s,\dist_1',a_2)(\ov{W})}
\]
Hence, given $\dist_1'$ the action $a_2$ must also satisfy Equation~\ref{exp:limit reach} for $\epsilon$. 

\item Consider an action $a_2$ in $C_2$. 
Since $a_2$ in $C_2$ (i.e., satisfies Equation~\ref{exp:almost reach}) we must have 
$\dest(s,\dist_1^\epsilon,a_2) \subseteq X$ (stay in $X$ with probability~1) and 
$\dest(s,\dist_1^\epsilon,a_2) \cap Y \neq \emptyset$ (next state in $Y$ with positive
probability).
Since $\dist'_1$ assigns positive probability to precisely the same set of actions as $\dist_1^\epsilon$, 
i.e., $\supp(\dist'_1)=\supp(\dist_1^\epsilon)$, 
we have that $\dest(s,\dist_1',a_2) = \dest(s,\dist_1^\epsilon,a_2) \subseteq X$ (stay in $X$ with probability~1) and 
$\dest(s,\dist_1',a_2) \cap Y = \dest(s,\dist_1^\epsilon,a_2) \cap Y  \neq \emptyset$ (next state in $Y$ with positive
probability).
Hence we have that $\dist_1'$ and $a_2$ must also satisfy Equation~\ref{exp:almost reach}.

\item Finally consider an action $a_2$ in $C_3$. 
We must have that (i)~$\dest(s,a,a_2) \subseteq Z$ and 
(ii)~$\cost(s,a,a_2)=1$; 
because otherwise we would either not end up in $Z$ or not get reward~1 with probability at 
least $\frac{\trans_{\min}}{m}$ when $a_2$ is played against $\dist_1^\epsilon$ 
(contradicting that $a_2$ satisfies  Equation~\ref{exp:get 1}).
Since $\dist_1'$ plays $a$ with larger probability than $\dist_1^\epsilon$, and
all other actions are scaled with probabilities of $\dist_1^\epsilon$, it
follows that for every $a_2$ in $C_3$ we must have that $\dist_1'$ and $a_2$ 
satisfy Equation~\ref{exp:get 1}. 
%%by construction.
\end{enumerate}
The desired result follows.
\end{proof}

\smallskip\noindent{\bf The action with large probability.}
In Lemma~\ref{lemm:one_action} we showed that some action is played with large
probability.
In the lemma the action was chosen depending on $\epsilon$, but since there
are only finitely many actions and if an action satisfies for some $0<\epsilon<\frac{1}{2}$,
then it also satisfies for all $\epsilon'$ such that $\epsilon\leq \epsilon' <\frac{1}{2}$, and
thus it follows that there is an action that is played with large probability.
We will call a parametrized distribution $\dist_1^\epsilon$, for $0<\epsilon<\frac{1}{2}$, an
\emph{$a$-large} distribution if the distribution plays action $a$ with probability 
at least $1-\epsilon\cdot \trans_{\min}$.
Thus the existence of witness $a$-large distributions, if such distributions 
exist, follows from Lemma~\ref{lemm:one_action}.
The main crux of the algorithm would be to find an action $a$ and 
a parametrized distribution that is $a$-large as a witness distribution for 
$\EXP$.
Our algorithm  will use the $\LimitReach$ operator iteratively.
The key information we need is encoded as a matrix as follows.

\smallskip\noindent{\bf The matrix for action sets.}
Given a state $s$, and the sets $U\subseteq Y\subseteq Z\subseteq X\subseteq W$, 
we define an $|\mov_1(s)|\times |\mov_2(s)|$-matrix $M$, such that 
$M_{a_1,a_2}\in \{\ov{W},W,U,X,Y,Z^0,Z^1\}$, that corresponds to the type of successor encountered 
if player~1 plays action $a_1$ and player~2 plays action $a_2$.
Let 
\[M_{a_1,a_2}=\begin{cases}
\overline{W} & \text{if }\dest(s,a_1,a_2)\cap \overline{W}\neq \emptyset \\
U & \text{if }\dest(s,a_1,a_2)\cap U\neq \emptyset \text{ and }\dest(s,a_1,a_2)\cap \overline{W}= \emptyset\\
W & \text{if }\dest(s,a_1,a_2)\cap (W\setminus X)\neq \emptyset \text{ and }\dest(s,a_1,a_2)\cap (\overline{W}\cup U)= \emptyset\\
Y & \text{if }\dest(s,a_1,a_2)\cap (Y\setminus U)\neq \emptyset \text{ and }\dest(s,a_1,a_2)\cap (\overline{W}\cup U\cup (W\setminus X))= \emptyset\\
X & \text{if }\dest(s,a_1,a_2)\cap (X\setminus Z)\neq \emptyset \\ &\qquad \text{and }\dest(s,a_1,a_2)\cap (\overline{W}\cup U\cup (W\setminus X)\cup (Y\setminus U))= \emptyset\\
Z^\ell & \text{if }\dest(s,a_1,a_2)\cap (Z\setminus Y)\neq \emptyset \\ &\qquad \text{and }\dest(s,a_1,a_2)\cap (\overline{W}\cup U\cup (W\setminus X) \cup (Y\setminus U)\cup (X\setminus Z))= \emptyset \\
    & \qquad \text{and } \cost(s,a_1,a_2)=\ell, \text{ for } \ell\in \set{0,1} \enspace .
\end{cases}
\]
The matrix uses that $U\subseteq Y\subseteq Z\subseteq X\subseteq W$, to ensure that the matrix is well-defined. Notice that $M$ encodes all the information needed by $\LimitReach$ (the entries equal to $W,Y,X,Z^1,Z^0$ all ensures both $\ov{W}$ and $U$ are not reached, $U$ ensures that $U$ is reached with probability at least $\delta_{\min}$ and $\ov{W}$ is not reached. The entries $\ov{W}$ ensures that $\ov{W}$ is reached with probability between $\delta_{\min}$ and $1$). Hence, we could alternatively give $M$ as input to $\LimitReach$.

\smallskip\noindent{\bf Intuitive description of the algorithm.} 
We first present an intuitive description of our algorithm and then present it formally.
The basic idea of the algorithm is to use $\LimitReach$ iteratively and the existence of $a$-large witness distributions. 
Given a candidate action $a$, we reject $a$ or accept $a$ using the following procedure.
First, given the action $a$, if there is an action $a_2$ such that $W$ is left with positive probability given $a$ and $a_2$
(i.e., $M_{a,a_2}=\ov{W}$), then we reject $a$.
Second, we check if playing $a$ with probability 1 satisfies all actions (by either of the three conditions), 
and if so we accept. 
If neither of the first two conditions hold, then we use an iterative procedure.
Let $C$ be the set of actions which are guaranteed to be satisfied (by Equation~\ref{exp:limit reach}) 
by playing an $a$-large distribution ($C$ consists of each action $a_2$ such that $M_{a,a_2}=U$). 
We run $\LimitReach$, and start with $(\Gamma_1(s)\setminus \{a\})$ as \emph{available actions} for player~1 
(we are only interested in $a$-large distributions and we do not consider $a$ for $\LimitReach$) and 
$(\Gamma_2(s)\setminus C)$ as available actions for player~2. 
If $\LimitReach$ rejects, we also reject: this is because no matter which action $a_1\neq a$ is played with the largest 
probability (and we could not play $a$ alone) there is an action $a_2$, such that $M_{a_1,a_2}=\ov{W}$ and $M_{a,a_2}\neq U$, 
which ensures that all three equations are violated.
If $\LimitReach$ accepts, then we obtain a %parametrized 
witness distribution $\dist_1$ and a set $A_3$ of actions of player~2
such that $\dist_1$ satisfies Equation~\ref{exp:limit reach} for all actions in $A_3$. 
We then create $\dist_1'$, which is $\dist_1$ scaled so that it plays an $a$-large distribution 
(note that $\dist_1$ plays $a$ with probability 0).
Afterwards we check if all actions for player~2 are satisfied by $\dist_1'$. 
If so, we accept. Otherwise, we check that whether for each action $a_2$ outside $(A_3\cup C)$ 
we can satisfy either Equation~\ref{exp:almost reach} or Equation~\ref{exp:get 1}: 
%(such actions must be satisfied using either of those two equations, by Accept property~d of $\LimitReach$): that is, we use that 
for $a_2$ to be satisfied using Equation~\ref{exp:get 1}, we must have that $M_{a,a_2}=Z^1$; and 
for $a_2$ to be satisfied using Equation~\ref{exp:almost reach}, the distribution $\dist_1'$ must play some 
action $a_1$ with positive probability such that $M_{a_1,a_2}=Y$. 
If for some $a_2$ outside $(A_3\cup C)$, neither 
$M_{a,a_2}=Z^1$, nor $M_{a_1,a_2}=Y$, for some $a_1$ played with positive probability, we reject. 
Otherwise, if we did not reject, we remove each action $a_1$ for player~1 from available actions, 
for which there exists an $a_2\in (A_3\cup C)$, such that $M_{a_1,a_2}=W$. 
Note that if $M_{a_1,a_2}=W$, then we cannot satisfy $a_2$ using either Equation~\ref{exp:almost reach} 
or Equation~\ref{exp:get 1}, if we play $a_1$ with positive probability. 
If the set of available actions does not contain $a$, then we cannot play $a$ with 
positive probability in an $a$-large distribution, which clearly means that no $a$-large distribution exists and thus we reject. 
If this new, smaller set of actions for player~1 contains $a$, we iterate on with the new set as 
the set of available actions for player~1, and the available set for player~2 always remains 
as $(\Gamma_2(s)\setminus C)$.
Since, in every iteration, we get a smaller set of actions for player~1, we terminate at some point.

\smallskip\noindent{\bf The algorithm \algopred.}
We now describe the steps of the algorithm which we refer as \algopred\ 
(algorithm for predecessor computation). 
For a state $s$, we consider every action $a \in \mov_1(s)$ as a candidate for the existence of an $a$-large witness distribution.
For each action $a$ we execute the following steps:
\begin{enumerate} 

\item \emph{(Reject~1).} Reject the choice of $a$ if there exists $a_2 \in \mov_2(s)$ such that $M_{a,a_2}=\ov{W}$.

\item \emph{(Accept~1).} Accept $a$ if for all $a_2 \in \mov_2(s)$ we have $M_{a,a_2}\in \{U,Y,Z^1\}$, and then return the distribution that plays 
$a$ with probability 1, and return ``Accept" for state $s$.

\item Let $C$ be the set of actions $a_2$ in $\mov_2(s)$ such that $M_{a,a_2}\neq U$. 
Initialize $B_1^0$ and $A_1^0$ as $(\mov_1(s)\setminus \{a\})$. 
The remainder of the algorithm will be done in iterations. 

\item \emph{(Iteration).}\label{algo:iteration}
In iteration $i\geq 1$, run $\LimitReach(s,W,U,((A_1^{i-1}\cap B_1^{i-1})\setminus \{a\}),C)$. \\
(Reject~2): if $\LimitReach(s,W,U,((A_1^{i-1}\cap B_1^{i-1})\setminus \{a\}),C)$ rejects the input, then reject this choice of $a$.
Otherwise let $A_2^i$ be the returned set; and let $\dist_{1}^{\epsilon,i}$ be a witness parametrized 
distribution (parametrized by $0<\epsilon<\frac{1}{2}$ which is obtained by the support of $\dist_1^{\epsilon,i}$ and
the ranking of the actions in the support). 
We will now define some sets of actions.
\begin{enumerate}
\item \label{algo:label_1} Let $A_1^i=\supp(\dist_1^{\epsilon,i}) \cup \set{a}$.
\item Let $B_1^i$ be all actions $a_1$ in $\mov_1(s)$ such that for all $a_2\in (C\setminus A_2^i)$ we have $M_{a_1,a_2}\neq W$.
\item Let $B_2^i$ be all actions $a_2$ in $(C\setminus A_2^i)$ such that either (i)~$M_{a,a_2}=Z^1$;  or 
(ii)~there exists an action $a_1\in A_1^i$ with $M_{a_1,a_2}=Y$. 
\end{enumerate}
\item We reject in the following cases: 
\begin{itemize}
\item {\em (Reject~3)}. If $((A_1^{i}\cap B_1^{i})\setminus \set{a})=\emptyset$,
then reject this choice of $a$. 
\item {\em (Reject~4)}. If $(C\setminus A^i_2)\neq B^i_2$, then reject this choice of $a$. 
\item {\em (Reject~5)}. If $a\not\in B^i_1$, then reject this choice of $a$. 
\end{itemize}
\item {\em (Accept~2)}. Otherwise if $A_1^i\subseteq B_1^i$, then return accept $a$, and return the parametrized 
distribution $\dist_1^\epsilon$, for $0<\epsilon<\frac{1}{2}$, that plays $a$ with probability 
$1-\epsilon\cdot \delta_{\min}$ and with probability $\epsilon\cdot \delta_{\min}$ follows 
$\dist_{1}^{\epsilon,i}$, and also ``Accept" state $s$.
\item If the action is neither accepted nor rejected, then go to iteration $i+1$ in step~\ref{algo:iteration}.
\end{enumerate}
If all choices of action $a \in \mov_1(s)$ get rejected, then ``Reject" state $s$.

The parametrized distribution for Accept~2 is returned as the special action $a$ (to be 
played with probability $1-\epsilon\cdot \delta_{\min}$, for $0<\epsilon<\frac{1}{2}$), 
the support set of $\dist_1^{\epsilon,i}$ and the ranking function of the support as 
given by the $\LimitReach$ operator
(which gives the parametrized distribution for $\dist_1^{\epsilon,i}$ which is 
multiplied by $\epsilon\cdot \delta_{\min}$ to get the parametrized $a$-large witness distribution $\dist_1^{\epsilon}$ and $a$ is played with the remaining probability).

\smallskip\noindent{\em Illustrations with examples.}
We illustrate our algorithm on four $M$-matrices shown in Figure~\ref{fig: M-matrices}. 
First observe that the only feasible candidate for an $a$-large distribution is the first row, 
because each other row contains an $\ov{W}$ entry, and thus will be rejected at the start.
The first matrix shown in Figure~\ref{fig:m-accept} will be accepted by the algorithm and 
the other three will be rejected by the algorithm. 
\begin{enumerate}
\item Consider first the matrix in Figure~\ref{fig:m-accept}. 
Then the algorithm is run with the first row as $a$, it will call $\LimitReach$ with the 
all rows but the first row for player~1 and all columns but the first column for player~2
(since given the first row, the first column satisfies Equation~\ref{exp:limit reach}).
The $\LimitReach$ algorithm will then return the distribution $d$ of playing the second row 
with probability $1-\frac{\epsilon}{2}$ and the third row with probability $\frac{\epsilon}{2}$. 
It also returns the set $A_3$ containing the second and third column (they satisfy Equation~\ref{exp:limit reach}).  
We then get accept in that iteration, because column~4 and column~5 can be satisfied by Equation~\ref{exp:almost reach} and 
column~6 can be satisfied by Equation~\ref{exp:get 1}.

\item Consider now the second matrix, the one in Figure~\ref{fig:m-reject1}. 
It will get rejected at start, because in this case each row contains an $\ov{W}$ entry. 

\item The third matrix, the one in Figure~\ref{fig:m-reject2}, will get rejected in the second iteration. 
In the first iteration, $\LimitReach$ will return the same distribution $d$ as for the first matrix along with the same $A_3$. 
This time, we cannot accept directly, because $d$ no longer satisfies any of the three equations, for column~5. 
At that point, the algorithm considers that each column $a_2 \in \set{4,5,6}$ such that $M_{a_1,a_2}=Y$ for some $a_1 \in\set{1,2,3}$ or 
$M_{a,a_2}=Z^1$ (where $a=1$). 
Thus, the algorithm removes row~2, from the set of possible rows, because column~5 is such that $M_{2,5}=W$, and $5\not\in A_3$ and iterate. 
Then the algorithm calls $\LimitReach$ and gets back reject, because each of the rows left contains at least one instance of $\ov{W}$.
Hence the algorithm rejects.
 
\item For the last matrix, the one in Figure~\ref{fig:m-reject3}, the algorithm calls $\LimitReach$ and gets $d$ and $A_3$, but this time the algorithm rejects at that point, 
because row 6 (which is not in $A_3$) does not contain an action $a_1$ played with positive probability such that $M_{a_1,6}=Y$ or is such that $M_{a,6}=Z^1$. %%which causes us to reject it.

\end{enumerate}

\begin{figure}
        \centering
        \begin{subfigure}[b]{0.45\textwidth}
\[ 
M=\left( \begin{array}{cccccc}
U         & W         & W         & X         &  Y        & Z^1\\
\ov{W} & U         & W          & Y         &  X        & X \\
\ov{W} & \ov{W} & U          & X         &  X        & X \\
\ov{W} & \ov{W} & \ov{W}  & \ov{W} & \ov{W} & \ov{W} \\
\end{array} \right)\] 
\caption{This illustrates a $M$-matrix, which has an $a$-large distribution, where $a$ corresponds to the first row.\label{fig:m-accept}}        

		\end{subfigure}
		\qquad
		\begin{subfigure}[b]{0.45\textwidth}
\[ 
M=\left( \begin{array}{cccccc}
U         & \circled{\textoverline{$W$}}        & W         & X         &  Y        & Z^1\\
\ov{W} & U         & W          & Y         &  W        & X \\
\ov{W} & \ov{W} & U          & X         &  X        & X \\
\ov{W} & \ov{W} & \ov{W}  & \ov{W} & \ov{W} & \ov{W} \\
\end{array} \right)\] 
\caption{This illustrates a $M$-matrix, which has no $a$-large distribution. The cicled entry is the only entry changed as compared to Figure~\ref{fig:m-accept}.\label{fig:m-reject1}}

		\end{subfigure}
		\\
        \begin{subfigure}[b]{0.45\textwidth}
\[ 
M=\left( \begin{array}{cccccc}
U         & W         & W         & X         &  Y        & Z^1\\
\ov{W} & U         & W          & Y         & {\circled{$W$}}        & X \\
\ov{W} & \ov{W} & U          & X         &  X        & X \\
\ov{W} & \ov{W} & \ov{W}  & \ov{W} & \ov{W} & \ov{W} \\
\end{array} \right)\] 
\caption{This illustrates a $M$-matrix, which has no $a$-large distribution. The cicled entry is the only entry changed as compared to Figure~\ref{fig:m-accept}.\label{fig:m-reject2}}

		\end{subfigure}
		\qquad
        \begin{subfigure}[b]{0.45\textwidth}
\[ 
M=\left( \begin{array}{cccccc}
U         & W         &  W         & X         &  Y        & {\circled{$X$}}\\
\ov{W} & U         & W          & Y         &  X        & \circled{$Z^1$} \\
\ov{W} & \ov{W} & U          & X         &  X        & X \\
\ov{W} & \ov{W} & \ov{W}  & \ov{W} & \ov{W} & \ov{W} \\
\end{array} \right)\] 
\caption{This illustrates a $M$-matrix, which has no $a$-large distribution. The circled entries are the only entries changed as compared to Figure~\ref{fig:m-accept}.\label{fig:m-reject3}}        

		\end{subfigure}
\caption{\label{fig: M-matrices}}
\end{figure}

\begin{lemma}\label{lemm: distribution algorithm accept}
Given $U\subseteq Y\subseteq Z\subseteq X\subseteq W$ and a state $s$, if algorithm \algopred\ accepts $s$, 
then $s\in \EXP(W,U,X,Y,Z)$. Furthermore, for every $0<\epsilon<\frac{1}{2}$ there exists a witness distribution $\dist_1^\epsilon$
with patience at most $\left(\frac{\epsilon \cdot \trans_{\min}}{2}\right)^{-(|\mov_1(s)|-1)}$ to satisfy at least one of the three 
required conditions (Equation~\ref{exp:limit reach}, Equation~\ref{exp:almost reach}, 
or Equation~\ref{exp:get 1})  for $\EXP$ for every action $a_2 \in \mov_2(s)$. 
\end{lemma}
\begin{proof}
We will next show that if \algopred\ returns a parametrized distribution $\dist_1^\epsilon$,
then for all $0<\epsilon<\frac{1}{2}$ and for all actions $a_2\in \mov_2(s)$,
at least one of the three conditions of $\EXP$ is satisfied. 
This will show that $s\in \EXP(W,U,X,Y,Z)$.
The algorithm accepts state $s$ and returns a distribution at two places, namely, 
(Accept~1) and (Accept~2). 
For the case of Accept~1: the algorithms returns a distribution that plays 
some action $a$ with probability 1; and 
for the case of Accept~2 it returns a distribution that plays some subset of actions 
(at least 2) with positive probability. 
We analyze both the cases below.

\begin{enumerate}

\item \emph{Case Accept~1.} 
In the first case for all actions $a_2$ we have that $M_{a,a_2}\in \{U,Y,Z^1\}$. 
We analyze the three sub-cases.
\begin{enumerate}
\item If $M_{a,a_2}=U$, 
then $\dest(s,a,s_2) \cap U \neq \emptyset$ (i.e., the next
state is in $U$ with positive probability) and $\dest(s,a,a_2) \cap \ov{W}
=\emptyset$ (i.e., the next state is in $\ov{W}$ with probability~0) and
hence Equation~\ref{exp:limit reach} is satisfied.

\item If $M_{a,a_2}=Y$, then 
(i)~$\dest(s,a,a_2)\cap (Y\setminus U)\neq \emptyset$ which implies that 
$\dest(s,a,a_2) \cap Y \neq \emptyset$, since $(Y\setminus U) \subseteq Y$; and  
(ii)~$\dest(s,a,a_2)\cap (\overline{W}\cup U\cup (W\setminus X))= \emptyset$
which implies that $\dest(s,a,a_2)\cap (\overline{X}\cup U)= \emptyset$
because as $X \subseteq W$ we have $(\overline{W}\cup U\cup (W\setminus X))= \overline{X} \cup U$;
and hence $\dest(s,a,a_2) \subseteq X$.
The first condition ensures that the next state is in $Y$ with positive 
probability and the second condition ensures the next state is in $X$ with 
probability~1, 
and thus  Equation~\ref{exp:almost reach}  is satisfied. 

\item If $M_{a,a_2}=Z^1$, then
(i)~$\dest(s,a,a_2)\cap (Z\setminus Y)\neq \emptyset$ which implies that 
$\dest(s,a,a_2) \cap Z \neq \emptyset$; and
(ii)~$\dest(s,a,a_2)\cap (\overline{W}\cup U\cup (W\setminus X) \cup (Y\setminus U) \cup (X\setminus Z))= \emptyset$
which implies that $\dest(s,a,a_2)\cap (\overline{Z}\cup U\cup Y)= \emptyset$,
because as $Z \subseteq X \subseteq W$ we have 
$(\overline{W}\cup U\cup (W\setminus X) \cup Y\cup (X\setminus Z))
= (\overline{Z}\cup U\cup Y)$, and hence $\dest(s,a,a_2) \subseteq Z$ 
(i.e., next state in $Z$ with probability~1);
and (iii)~$\cost(s,a,a_2)=1$ (i.e., expected reward is~1).
It follows that Equation~\ref{exp:get 1} is satisfied.
\end{enumerate}

\item \emph{Case Accept~2.}
In the second case, we consider the case when the algorithm returns a 
parameterized distribution $\dist_1^\epsilon$, for $0<\epsilon<\frac{1}{2}$, in iteration $i$. 
Let the action played with probability $1-\epsilon\cdot \delta_{\min}$ be $a$. 
Such an action clearly exists, by construction. 
For any $a_2\in \mov_2(s)$ such that $M_{a,a_2}=U$, 
then the next state is in $U$ with probability at least  $(1-\epsilon\cdot \trans_{\min})\cdot \trans_{\min}$
and the next state is in $\ov{W}$ with probability at most $\epsilon\cdot \trans_{\min}$ 
and the ratio is at least $2\cdot \epsilon$;
thus the distribution $\dist_1^\epsilon$ and $a_2$ satisfy Equation~\ref{exp:limit reach} for 
$2\cdot \epsilon$. As $0<\epsilon<\frac{1}{2}$ is arbitrary the result follows for all $a_2$ such
that $M_{a,a_2}=U$.
We consider the set $C$ of remaining actions in $\mov_2(s)$, i.e., for all $a_2\in C$ 
we have $M_{a,a_2}\neq U$.
%We can partition the actions in $C$ into two sets, $A$ and $B$. 
%The set $A$ is such that Equation~\ref{exp:limit reach} is satisfied and 
%$B$ is such that Equation~\ref{exp:limit reach} is not satisfied. 

\emph{Satisfying Equation~\ref{exp:limit reach} in $A_2^i$.}
We have that $M_{a,a_2}\neq \overline{W}$, for all $a_2\in \mov_2(s)$, because otherwise the guess of action $a$ would have been rejected, in (Reject~1). 
We also have that $\LimitReach(s,W,U,B',C)$, for $B'\subseteq (\mov_1(s)\setminus \{a\})$ must return an distribution $\dist_1'$ over $B'$ and a set $A' 
\subseteq C$, such that for all $a_2\in A'$, the action $a_2$ and the distribution $\dist_1'$ satisfies Equation~\ref{exp:limit reach}
(by Accept property~a of $\LimitReach$). 
%Thus we see that our our construction of $\dist^{\epsilon}$ ensures that $A^i_2\subseteq A$, where $A^i_2$ is as defined in the description of the algorithm.
In the last iteration the set $A_2^i$ is the set returned by $\LimitReach(s,W,U, ((A_1^{i-1}\cap B_1^{i-1})\setminus \{a\}),C)$,
and the distribution $\dist_1^{\epsilon,i}$ satisfies Equation~\ref{exp:limit reach} for all actions in $A_2^i$
(again by Accept property~a of $\LimitReach$ since $A_2^i$ is the returned subset of $C$). 
Since $\dist_1^\epsilon$ only plays $a$ with high probability and 
only scales the distribution $\dist_1^{\epsilon,i}$ it follows
(similarly to Case~1 of Lemma~\ref{lemm:one_action}) that 
$\dist_1^\epsilon$ satisfies Equation~\ref{exp:limit reach} for all actions in $A_2^i$.

\emph{Satisfying Equation~\ref{exp:almost reach} or  Equation~\ref{exp:get 1} in $(C \setminus A_2^i)$.}
By definition of $B^i_1$ and $A^i_1$ (Step~4~(a) and Step~4~(b) of the 
algorithm), and that $A^i_1 \subseteq B^i_1$ (from Accept~2 of the algorithm),
it follows that the distribution $\dist_1^\epsilon$ is such that 
for all $a_2\in(C\setminus A^i_2)$ and $a_1 \in \supp(\dist_1^\epsilon) \cup \set{a}=A^i_1$ 
we have $M_{a_1,a_2}\neq W$.
Also for all $a_2\in (C\setminus A^i_2)$ and all $a_1$ such that 
$\dist_1^\epsilon(a_1)>0$, we have from Accept property~b of $\LimitReach$ that $M_{a_1,a_2}\neq \overline{W}$ and $M_{a_1,a_2}\neq U$. 
Notice that therefore for all $a_1 \in \supp(\dist_1^\epsilon)$ and $a_2 \in (C \setminus A_2^i)$ 
we have $M_{a_1,a_2}\in \set{X,Y,Z^0,Z^1}$,  which implies that $\dest(s,\dist_1^\epsilon,a_2)(X)=1$.
For all $a_2\in (C\setminus A^i_2)$ we have that  either (i)~$M_{a,a_2}=Z^1$; or 
(ii)~$\dist_1^\epsilon$ assigned positive probability to some $a_1$ such that $M_{a_1,a_2}=Y$, 
because otherwise $(C\setminus A^i_2)\neq B^i_2$ and we would have rejected this choice of $a$
(by Reject~4 of the algorithm). 
Notice that $M_{a,a_2}=Z^1$ implies that $\dest(s,a,a_2)(Z)=1$ and that $\cost(s,a,a_2)=1$, thus, since the distribution the algorithm returned was $a$-large, we get that we reach $Z$  in one step with probability at least $1-\epsilon\cdot \delta_{\min}$ and get reward~1 with probability at least $1-\epsilon\cdot \delta_{\min}$, hence Equation~\ref{exp:get 1} is satisfied.
If the second case holds (i.e.,  $M_{a_1,a_2}=Y$), 
 we have $\dest(s,\dist_1^\epsilon,a_2) \cap (Y\setminus U) \neq \emptyset$
(i.e., $Y$ is reached with positive probability in one step), thus implying that Equation~\ref{exp:almost reach} is satisfied.  
\end{enumerate}

Therefore the distribution $\dist_1^\epsilon$ is a witness distribution to satisfy the required
conditions for $0<\epsilon<\frac{1}{2}$ for $\EXP$. It follows that $s \in \EXP(W,U,X,Y,Z)$.

\smallskip\noindent{\bf Patience.}
The distribution returned by $\LimitReach$ over $|\mov_1(s)|-1$ actions has patience at 
most $\left(\frac{\epsilon \cdot \trans_{\min}}{2}\right)^{-(|\mov_1(s)|-2)}$. Hence it is clear from the algorithm that 
the distribution returned by the algorithm has patience at most 
$\left(\frac{\epsilon \cdot \trans_{\min}}{2}\right)^{-(|\mov_1(s)|-1)}$.
\end{proof}

%\begin{lemma}For a given $U\subseteq Y\subseteq Z\subseteq X\subseteq W$ and some state $s$, let $a$ be some action in $\mov_1(s)$. Also, let $i\geq 1$ be some number. For all $j$ let $A^j_1$, $A^j_2$, $B^j_1$ and $B^j_2$ be as defined by Algorithm \algopred(s,W,U,X,Y,Z). Assume that all witness distributions $\dist$, such that $\dist(a)\geq 1-\epsilon\cdot \delta_{\min}$, are restricted to actions in $(A_1^{i-1}\cap B_1^{i-1})\setminus \{a\}$, then all witness distributions, $\dist'$, , such that $\dist'(a)\geq 1-\epsilon\cdot \delta_{\min}$, are restricted to actions in $A^i_2$.
%Also, algorithm \algopred accepts or rejects a choice of action to play with high probability after at most $\min(|\Gamma_1(s)|,|\Gamma_2(s)|)$ iterations of the inner loop.
%\end{lemma}
Our next goal is to present a lemma that complements the previous lemma. 
In other words, we would show that if \algopred\ rejects an action $a$,
then there would be no $a$-large distributions as witnesses for $\EXP$.
The algorithm rejects an action $a$ at four places, and we will show
that all the rejections are \emph{sound} (i.e., if $a$ is rejected,
then there is no $a$-large witness distribution).
We first show that the first rejection is sound.

\smallskip\noindent{\bf Soundness of Reject~1.}
We consider the case of Reject~1. In this case, there exists
an action $a_2$ such that $M_{a,a_2}=\ov{W}$.
Given an $a$-large distribution $\dist_1^\epsilon$, 
the one step probability to reach $\ov{W}$  
(i.e., $\trans(s,\dist_1^\epsilon,a_2)(\ov{W})$) is 
at least $x=(1-\epsilon\cdot\delta_{\min})\cdot\delta_{\min}> \epsilon$, since $\epsilon<\frac{1}{2}$ and $\delta_{\min}\leq 1$,
and even if $U$ is reached with the remaining probability 
(i.e., even if $\trans(s,\dist_1^\epsilon,a_2)(U) =1-x$), 
it follows that Equation~\ref{exp:limit reach} is violated, 
for all $0<\epsilon<\frac{1}{2}$.
The remaining two expressions cannot be satisfied because $X\subseteq W$ 
and since we leave $W$ with positive probability we as well 
leave $X$ with positive probability.
It follows that the rejection of action $a$ is sound for Reject~1.

\smallskip\noindent{\bf Rejects in iteration.}
The other places the algorithm can reject action $a$, i.e., 
(Reject~2), (Reject~3), (Reject~4), and (Reject~5), 
are part of the iterative procedure.
To prove soundness of these rejects we will define a loop invariant and 
prove the loop invariant inductively. 
We will also show that with the loop invariant we can establish soundness
of the rejects in the iterative procedure as well as the termination
of the algorithm.

\smallskip\noindent{\bf The loop invariant.}
The \emph{loop invariant} is as follows:
\begin{itemize}
\item Any $a$-large witness distribution $\dist_1^\epsilon$ for $\EXP$ 
only plays actions in $(A_1^{i}\cap B_1^{i})\cup \{a\}$ with positive probabilities, 
for all $i\geq 0$, i.e., $\supp(\dist_1^\epsilon) \subseteq 
(A_1^{i}\cap B_1^{i})\cup \{a\}$.
\end{itemize}
We will also establish the {\em monotonicity} (strictly decreasing till a fixpoint is reached) 
property that $(A_1^{i}\cap B_1^{i})\cup \{a\}\subseteq(A_1^{i-1}\cap B_1^{i-1})\cup \{a\}$, for all $i>0$; 
and equality implies termination in iteration $i$. 

\smallskip\noindent{\bf Inductive proof of loop invariant.}
We present the basic inductive argument for the loop invariant:
%The loop invariant is that that any satisfying distribution $\dist^\epsilon$ only plays actions in $(A_1^{i}\cap B_1^{i})\cup \{a\}$ with positive probabilities, for all $i\geq 0$. Also $(A_1^{i}\cap B_1^{i})\cup \{a\}\subseteq(A_1^{i-1}\cap B_1^{i-1})\cup \{a\}$, for all $i>0$ and equality implies termination.
%The proof will be by induction in $i$. We will show that $(A_1^{i}\cap B_1^{i})\cup \{a\}\subset(A_1^{i-1}\cap B_1^{i-1})\cup \{a\}$, for all $i>0$ separately though.
\begin{itemize}
\item \noindent{\bf The base case, $i=0$.} The base case, for $i=0$ is trivial, since $A_1^{0}= B_1^{0}=(\mov_1(s)\setminus \{a\})$, 
thus implying that $(A_1^{i}\cap B_1^{i})\cup \{a\}=\mov_1(s)$.

\item \noindent{\bf The induction case, $i>0$.} By inductive hypothesis,
any $a$-large witness distribution $\dist_1^\epsilon$ only plays actions in $(A_1^{i-1}\cap B_1^{i-1})\cup \{a\}$ 
with positive probabilities, and we need to establish for $i$.
We will show that any $a$-large witness distribution can only play actions in 
$A_1^{i}\cup \{a\}=A_1^{i}$,
(see the following description of $A_1^i$ which uses the inductive hypothesis).
We refer to this as required property~1 for loop invariant.
Similarly, we establish the same for $B_1^i$ (see the following description of $B_1^i$ 
which uses the inductive hypothesis).
We refer to this as required property~2 for loop invariant.
%Also any satisfying distribution can only play actions in $B^i_1$, see the description of $B^i_1$ (again notice that it uses the induction hypothesis, that is any witness distribution only uses actions in $(A_1^{i-1}\cap B_1^{i-1})\setminus \{a\}$). 
Hence any witness $a$-large distribution can only play actions in $(A_1^{i}\cap B_1^{i})\cup \{a\}$.
\end{itemize}
The above proof requires to establish the key properties of $A_1^i$ and $B_1^i$. 
Before establishing them we first show the monotonicity property.

\smallskip\noindent{\bf Monotoncity property.}
We will show that we have $(A_1^{i}\cap B_1^{i})\cup \{a\}\subseteq(A_1^{i-1}\cap B_1^{i-1})\cup \{a\}$, for all $i>0$, 
and equality implies termination of the inner loop in iteration $i$. 
Notice that this implies that for any choice of $a$ the inner loop rejects $a$ or 
finds a distribution after at most $|\Gamma_1(s)|$ iterations. 
We have that $A_1^i = \supp(\dist_1^\epsilon) \cup \set{a}$ (by Step~4~(a) of \algopred), 
where $\dist_1^\epsilon$ is a witness distribution returned by $\LimitReach(s,W,U, ((A_1^{i-1}\cap B_1^{i-1})\setminus \{a\}),C)$.
Since  $\supp(\dist_1^\epsilon) \subseteq ((A_1^{i-1}\cap B_1^{i-1})\setminus \{a\})$, if $\LimitReach$ accepts, we have that  $A_1^i \subseteq (A_1^{i-1}\cap B_1^{i-1})\cup \{a\}$. 
Thus we get that $(A_1^{i}\cap B_1^{i})\cup \{a\}\subseteq A_1^{i} \cup \set{a} 
\subseteq (A_1^{i-1}\cap B_1^{i-1})\cup \{a\} $.
This establish monotonicity and now we show the termination.
Assume that $(A_1^{i}\cap B_1^{i})\cup \{a\}=(A_1^{i-1}\cap B_1^{i-1})\cup \{a\}$. 
Therefore we have that $\dist_1^\epsilon$ can only use actions in $((A_1^{i-1}\cap B_1^{i-1})\setminus \{a\})$, 
which is thus also $((A_1^{i}\cap B_1^{i})\setminus \{a\})$. 
But then either (i)~$a\not\in B_1^i$ or 
(ii)~$\supp(\dist_1^\epsilon) \cup\set{a}=A^i_1\subseteq (A_1^{i}\cap B_1^{i})\cup \{a\}$; which 
implies that $A^i_1\subseteq B_1^i$. But in the first case we reject (in (Reject~5)) and 
in the second case we accept (in (Accept~2)). 
This establishes the termination property.

\smallskip\noindent{\bf The properties of the sets for loop invariant.}
We now present the associated properties of the sets 
$A_1^i$, $A_2^i$, $B_1^i$, and $B_2^i$ to complete the inductive 
proof of the loop invariant.

\begin{enumerate}
\item {\em The property of the set $A^i_2$.} 
We first argue that $A^i_2$ has certain properties which will imply the key properties for $A^i_2$.
\begin{enumerate}
\item 
Since $\LimitReach(s,W,U, ((A_1^{i-1}\cap B_1^{i-1})\setminus \{a\}),C)$ accepts, we have that $A_2^i$ is a
subset of $C$.
There exists a witness parametrized distribution $\dist_1^\epsilon$, 
over $((A_1^{i-1}\cap B_1^{i-1})\setminus \{a\})$ such that for all $a_2\in A_2^i$ 
we have that $\dist_1^\epsilon$ and $a_2$ satisfies Equation~\ref{exp:limit reach}
(by Accept property~a of $\LimitReach$).

\item 
Also for all $a_2\in (C\setminus A_2^i)$ we have 
that $M_{a_1,a_2}\neq \overline{W}$ for all $a_1\in \supp(\dist_1^\epsilon)$ (Accept 
property~b of $\LimitReach$). 

\item
%%Furthermore $A^i_2$ was the maximum such set (Accept property~c of $\LimitReach$). 
Notice also that for any action $a_2\in C$, if a distribution over $A_1^{i-1}\cap B_1^{i-1}$ 
cannot satisfy $a_2$ using Equation~\ref{exp:limit reach}, then no distribution over 
$(A_1^{i-1}\cap B_1^{i-1})\cup \{a\}$ can either, 
since $M_{a,a_2}\neq U$ (from the definition of the set $C$) and hence $U$ cannot 
be reached as long as the distribution plays $a$. 
For an distribution $\dist'_1$ to be a witness distribution, all actions in $\mov_2(s)$ must 
satisfy either (i)~Equation~\ref{exp:limit reach}; or (ii)~Equation~\ref{exp:almost reach}; 
or (iii)~Equation~\ref{exp:get 1}. 
But if an action $a_2$ must satisfy either Equation~\ref{exp:almost reach} or Equation~\ref{exp:get 1}, 
we must have that $\dist'_1$ ensures that $\overline{X}$ is reached with probability~0 (i.e.,
$\dest(s,\dist'_1,a_2) \subseteq X$). Hence, since $X \subseteq W$ we also must have that 
$\overline{W}$ is reached with probability 0. 
\end{enumerate}
By Accept property~d of $\LimitReach$ we have that, since $A^i_2$ is returned by $\LimitReach$,
no $a$-large witness distribution $\dist'_1$ can satisfy any action 
$a_2$ in $(C\setminus A^i_2)$ using Equation~\ref{exp:limit reach}, 
while satisfying all actions in $C$ using Equation~\ref{exp:limit reach}, 
or Equation~\ref{exp:almost reach}, or Equation~\ref{exp:get 1}. Also, for all $a_2$ in $(C\setminus A^i_2)$ and all $a_1\in\supp(\dist_1^\epsilon)$ we have that $M_{a_1,a_2}\neq U$ 
(by Accept property~b of $\LimitReach$). 
Furthermore, by definition of $C$  for all $a_2\in C$ we have that $M_{a,a_2}\neq U$.
Therefore we have established the following key properties for $A^i_2$:
\begin{itemize}
\item Any $a$-large witness distribution $\dist'_1$ must satisfy all actions $a_2$  in $(C\setminus A^i_2)$ 
using either Equation~\ref{exp:almost reach} or Equation~\ref{exp:get 1}. 
\item 
For all $a_2\in (C\setminus A^i_2)$ and $a_1\in \supp(\dist_1^\epsilon)\cup \set{a}=A^i_1$ we have that $M_{a_1,a_2}\neq U$ .
\end{itemize}

\item {\em The property of the set $A^i_1$.} 
By accept property~c of $\LimitReach$ and since we did not reject in Reject 1, 
the set $A^i_1$ is the largest set, such that for all $a_1\in A^i_1$ there exists no $a_2$ in $(C\setminus A^i_2)$ with $M_{a_1,a_2}=\overline{W}$. 
But this means that any distribution that satisfies for all actions in $(C\setminus A^i_2)$ either Equation~\ref{exp:almost reach} or Equation~\ref{exp:get 1}, 
must play only actions in $A_1^i$. But from our description of $A^i_2$ we obtain that all $a$-large witness distributions must ensure that 
all actions in $(C\setminus A^i_2)$ are satisfied using either Equation~\ref{exp:almost reach} or Equation~\ref{exp:get 1}. 
Therefore we have established the following key property for $A^i_1$:
All $a$-large witness distributions must play only actions in $A^i_1$ with positive probability. 
This proves the required property~1 of the loop invariant.

\item {\em The property of the set $B^i_2$.} 
From the first key property of $A^i_2$ we have that any $a$-large witness distribution 
must ensure that all actions in $(C\setminus A^i_2)$ satisfy either Equation~\ref{exp:almost reach} 
or Equation~\ref{exp:get 1}. From the second key property of $A^i_2$, 
for all $a_1\in A^i_1$ and all $a_2\in (C\setminus A^i_2)$, we have that $M_{a_1,a_2}\neq U$.
The key property of  $A^i_1$ implies that any $a$-large 
witness distribution must play only actions in $A^i_1$.

Hence, for an $a$-large witness distribution $\dist'_1$, 
for all $a_2$ in $(C\setminus A^i_2)$ we must have that either 
(i)~$M_{a,a_2}=Z^1$ (to satisfy Equation~\ref{exp:get 1}); 
or (ii)~there is an action $a_1$ in $A^i_1$ such that $M_{a_1,a_2}=Y$ (to satisfy Equation~\ref{exp:almost reach} --- 
it would also be satisfied if $M_{a_1,a_2}=U$ but we know that $M_{a_1,a_2}\neq U$ by Accept property~b of $\LimitReach$). But that is precisely the definition of 
$B^i_2$ (Step~4~(c) of \algopred).
Therefore, we have the following key property for $B^i_2$: Actions $a_2$ in $(C\setminus (A^i_2\cup B^i_2))$ cannot be satisfied by Equation~\ref{exp:limit reach} or Equation~\ref{exp:almost reach} or Equation~\ref{exp:get 1} by any $a$-large witness distribution.

\item {\em The property of the set $B^i_1$.} 
We know from the first key property of $A^i_2$ that all actions in $(C\setminus A^i_2)$ must satisfy 
Equation~\ref{exp:almost reach} or Equation~\ref{exp:get 1}. 
But to do so we must leave $X$ with probability $0$. But $B^i_1$ is the largest set of actions such that for all actions $a_1$ in $B^i_1$ and for all actions 
$a_2$ in $(C\setminus A^i_2)$, we have that $M_{a_1,a_2}\neq W$ (Step~4~(b) of \algopred). Hence we have that an $a$-large distribution that plays an action in 
$(\mov_1(s)\setminus B^i_1)$ with positive probability violates both Equation~\ref{exp:almost reach} and Equation~\ref{exp:get 1} for 
some $a_2$ in $(C\setminus A^i_2)$. 
Therefore, we have the following key property for $B^i_1$: 
All $a$-large witness distributions only plays actions in $B^i_1$.
This also proves the required property~2 of the loop invariant.

\end{enumerate}
This establishes the inductive proof of the loop invariant.

\begin{lemma}\label{lemm: distribution algorithm reject}
For a given $U\subseteq Y\subseteq Z\subseteq X\subseteq W$, if Algorithm \algopred\ rejects state $s$, then $s\not\in \EXP(W,U,X,Y,Z)$.
Also, algorithm \algopred\ accepts or rejects a choice of action $a$ as a candidate for the existence of $a$-large witness distributions  at most $\min(|\Gamma_1(s)|,|\Gamma_2(s)|)$ iterations of the inner loop.
\end{lemma}

\begin{proof}
In the algorithm there are five places where a choice of $a$ might get rejected. 
We have already argued the soundness of Reject~1.
We prove the soundness of the other rejects below.

\begin{enumerate}
\item \emph{(Reject~2).} If $\LimitReach(s,W,U, ((A_1^{i-1}\cap B_1^{i-1})\setminus \{a\}),C)$ is rejected, then for all actions $a_1$ in $((A_1^{i-1}\cap B_1^{i-1})\setminus \{a\})$, there exists an action $a_2$ in $C$ such that $M_{a_1,a_2}=\ov{W}$, by the reject property of $\LimitReach$. But then consider any distribution $\dist_1$ over $((A_1^{i-1}\cap B_1^{i-1})\setminus \{a\})$, some action $a_1$ is played with probability at least $\frac{1}{m}$.
Hence the action $a_2$ such that $M_{a_1,a_2}=\ov{W}$, cannot be satisfied using neither (i)~Equation~\ref{exp:limit reach}; nor (ii)~Equation~\ref{exp:almost reach}; nor (iii)~Equation~\ref{exp:get 1}. The latter two because $\ov{W}$ is entered with positive probability in one step and hence $X$ is left with positive probability in one step. The first is because we reach $\ov{W}$ with probability at least $x=\frac{\trans_{\min}}{m}$ and even if we reach $U$ with probability $1-x$, we still do not satisfy Equation~\ref{exp:limit reach}. Now consider some distribution $\dist_1'$ over $(A_1^{i-1}\cap B_1^{i-1})\cup \{a\}$. Either it plays $a$ with probability 1 or not. If it does, then it cannot be a witness distribution, since it otherwise would have been accepted in Accept~1. If it does not then the argument is similar to the previous argument (in the case of Equation~\ref{exp:limit reach}, the argument also uses that $M_{a,a_2}\neq U$ from the definition of $C$). Hence no witness distribution exists that only uses actions in $(A_1^{i-1}\cap B_1^{i-1})\cup \{a\}$. Thus Reject~2 is a sound reject, by the loop invariant.

\item \emph{(Reject~3).} 
If $a$ is not accepted by Accept~1, then $a$ could not be played with probability~1.
For Reject~3, the condition $((A_1^{i}\cap B_1^{i})\setminus \set{a})=\emptyset$ is satisfied. 
Thus no $a$-large witness distribution can play anything but $a$ by the loop invariant. 
Therefore no $a$-large witness distribution can exist in this case. Thus, Reject~3 is a sound reject.

\item \emph{(Reject~4).} 
%\noindent{\bf Putting the descriptions to use.}
Consider an $a$-large witness distribution $\dist_1^\epsilon$.
The key property of $B^i_2$ implies that any action $a_2\in (C\setminus (A^i_2\cup B^i_2))$ cannot be satisfied using either of the equations.  But since $B^i_2\subseteq (C\setminus A^i_2)$ we must have that $B^i_2=(C\setminus A^i_2)$ for any $a$-large witness distribution to exists.
Therefore we can reject the choice of $a$ if $(C\setminus A^i_2)\neq B^i_2$.
Hence Reject~4 is a sound reject.

\item \emph{(Reject~5).} 
From the key property of the set $B_1^i$, we have that if $a\not \in B^i_1$, 
then no $a$-large witness distribution can play $a$ with positive probability,
which implies that no $a$-large witness distribution can exist. 
Hence Reject~5 is also a sound  reject.

\end{enumerate}

\smallskip\noindent{\bf Termination.}
We have already established (in "monotonicity and termination for loop invariant") that 
$(A_1^{i}\cap B_1^{i})\cup \{a\}\subseteq(A_1^{i-1}\cap B_1^{i-1})\cup \{a\}$, for all $i>0$ 
and equality implies termination of the inner loop in iteration $i$. 
Notice that this implies that for any choice of $a$ the inner loop rejects $a$ or finds a distribution after at most $|\Gamma_1(s)|$ iterations. 
We will now show that $A^i_2\subseteq A^{i-1}_2$, for all $i>0$ and equality implies termination in iteration $i$. 
Notice that this implies that for any choice of $a$ the inner loop rejects $a$ or finds a distribution after at most $|\Gamma_2(s)|$ iterations. 
We have that $A^i_2\subseteq A^{i-1}_2$, because $\dist_1^{\epsilon,i}$ 
could also be returned in iteration $i-1$ and $\LimitReach$ maximizes the number of $a_1$'s for which $\dist_1^{\epsilon,i}(a_1)>0$
(Accept property~c). 
Assume that $A^i_2= A^{i-1}_2$. Then $(C\setminus A^{i}_2)=(C\setminus A^{i-1}_2)$ and thus $B^i_1=B^{i-1}_1$. 
We also have that $A^i_1\subseteq (A_1^{i-1}\cap B_1^{i-1})\cup \{a\}$, thus implying that $A^i_1\subseteq (A_1^{i-1}\cap B_1^{i})\cup \{a\}$. 
Therefore $A^i_1\subseteq B_1^i$, since if $B_1^i$ does not contain $a$, neither does $B_1^{i-1}$ and thus we would have rejected the choice of $a$ in 
iteration $i-1$, because of (Reject~5). 
The desired result follows.
\end{proof}

\begin{lemma}\label{lemm: distribution algorithm time}
Given $U\subseteq Y\subseteq Z\subseteq X\subseteq W$ and a state $s$, \algopred\ terminates in time $O(|\mov_1(s)|^2\cdot |\mov_2(s)|^2+\sum_{a_1\in \mov_1(s),a_2\in \mov_2(s)}|\supp(s,a_1,a_2)|)$. Alternatively, if $M$ is given as input, the running time is $O(|\mov_1(s)|^2\cdot |\mov_2(s)|^2)$.
\end{lemma}
\begin{proof}
The calculation of $M$ can be done in time $\sum_{a_1\in \mov_1(s),a_2\in \mov_2(s)}|\supp(s,a_1,a_2)|$. As mentioned in the definition of $M$, we could alternatively use $M$ as input to $\LimitReach$ since it encodes all information needed. There are $|\mov_1(s)|$ different choices for which action $a$ to play with high probability. Given $a$, there are at most $\min(|\mov_1(s)|,|\mov_2(s)|)$ iterations of the inner loop, see Lemma~\ref{lemm: distribution algorithm reject}. Each iteration of the inner loop can be done in $O(|\mov_1(s)|\cdot |\mov_2(s)|)$ time, and is dominated by the running time of $\LimitReach$, which runs in time $O(\mov_1(s)|\cdot |\mov_2(s)|)$ on $M$, see~\cite{dAHK98}.
Hence, if $M$ is given as input we get a running time of $O(|\mov_1(s)|\cdot \min(|\mov_1(s)|,|\mov_2(s)|)\cdot |\mov_1(s)|\cdot |\mov_2(s)|)$, which is less than $O(|\mov_1(s)|^2\cdot |\mov_2(s)|^2)$.
\end{proof}

Combining Lemma~\ref{lemm: distribution algorithm accept}, Lemma~\ref{lemm: distribution algorithm reject} and Lemma~\ref{lemm: distribution algorithm time} we get the following lemma.

\begin{lemma}
The algorithm \algopred, for a given state $s$ and sets $U\subseteq Y\subseteq Z\subseteq X\subseteq W$, correctly computes if $s\in \EXP(W,U,X,Y,Z)$ and 
runs in time $O(|\mov_1(s)|^2\cdot |\mov_2(s)|^2 + \sum_{a_1\in \mov_1(s),a_2\in \mov_2(s)}|\supp(s,a_1,a_2)|)$.
\end{lemma}

\subsection{Iterative algorithm for value~1 set computation}
In this section we will present the nested iterative algorithm 
for the value~1 set computation. 
The nested iterative algorithm is succinctly represented as the 
following nested fixpoint formula ($\mu$-calculus formula) that
uses the $\EXP$ one-step predecessor operator.
Let 
\[
W^*=\nu W. \mu U. \nu X. \mu Y. \nu Z. \EXP(W,U,X,Y,Z) \enspace .
\]
We will show that 
$W^*=\val_1(\LimInfAvg,\bigstra_1^F)$ (also see the appendix, Section~\ref{sec:Appendix}, for an algorithmic description of computation of the 
$\mu$-calculus formula).
First in the next subsection we show that 
$W^* \subseteq \val_1(\LimInfAvg,\bigstra_1^S) \subseteq \val_1(\LimInfAvg,\bigstra_1^F)$;
and in the following subsection will establish the other inclusion.

\subsubsection{First inclusion: $W^* \subseteq \val_1(\LimInfAvg,\bigstra_1^S)$\label{sec:incl1}}
Let $\Theta_i$ denote the random variable for the reward at the $i$-th 
step of the game.
We will show that for all states $s$ in $W^*$ for all $\epsilon>0$,
there exists a stationary (hence finite-memory) strategy $\sigma_1^\epsilon$ 
for player~1 such that for all positional strategies $\sigma_2$ for player~2 
we have that 
\[
\lim_{t\rightarrow \infty} \frac{\sum_{i=0}^t \E^{\sigma_1^\epsilon,\sigma_2}_s[\Theta_i]}{t}\geq 1-\epsilon \enspace .
\]
This will show that 
$W^* \subseteq \val_1(\LimInfAvg,\bigstra_1^S) \subseteq \val_1(\LimInfAvg,\bigstra_1^F)$.
Notice that the statement is trivially satisfied if $W^*=\emptyset$, and hence we will assume that 
this is not so.

\smallskip\noindent{\bf Computation of $W^*$.}
We first analyze the  computation of $W^*$. 
Since $W^*$ is a fixpoint, we can replace $W$ by $W^*$ and get rid
of the outer most $\nu$ operator, and the rest of the $\mu$-calculus
formula also computes $W^*$.
In other words, we have
\[
W^*=\mu U. \nu X. \mu Y. \nu Z. \EXP(W^*,U,X,Y,Z) \enspace ,
\]
Thus the computation of $W^*$ is achieved as follows: $U_0$ is the empty set; 
and $U_i=\nu X. \mu Y. \nu Z. \EXP(W^*,U_{i-1},X,Y,Z)$, for $i\geq 1$. 
Let $\ell$ be the least index such that $U_\ell=W^*$. For any $i\geq 0$, we also have that 
$Y_{i,0}$ is the empty set and that  
$Y_{i,j}=\nu Z. \EXP(W^*,U_{i-1},U_i,Y_{i,j-1},Z)$, for $j\geq 1$. 
For a state $s \in W^*$, let the rank of state $s$ (denoted $\rank(s)=(i,j)$) 
be the tuple of $(i,j)$ such that $i$ is the least index with $s \in U_i$ 
(i.e., $s\in U_i \setminus U_{i-1}$); 
and $j$ is the least index with $s \in Y_{i,j}$
(i.e., $s\in Y_{i,j} \setminus Y_{i,j-1}$).
For $1\leq i \leq \ell$, let $\rank(i)=j$ be the least index when the fix point
converges for $U_i$, i.e., the least $j$ such that $Y_{i,j}=Y_{i,j+1}$. 
By definition of $W^*$, for all states $s \in W^*$, if $\rank(s)=(i,j)$,
then we must have that for all $\epsilon>0$ 
there is a distribution $\dist_1^{\epsilon}$ over $\mov_1(s)$ 
such that for all actions $a_2\in \mov_2(s)$ for player~2 we have that 
\begin{align}
&\ \big(\epsilon \cdot \trans(s,\dist^{\epsilon}_1,a_2)(U_{i-1}) >  \trans(s,\dist^{\epsilon}_1,a_2)(\overline{W}^*)\big) \label{eqt: limit reach}\\[2ex]
\vee&\ \big(  \trans(s,\dist^{\epsilon}_1,a_2)(U_{i})=1 \wedge  
 \trans(s,\dist^{\epsilon}_1,a_2)(Y_{i,j-1}) > 0 \big)
\label{eqt: almost reach} \\[2ex]
\vee&\ \big( \trans(s,\dist^{\epsilon}_1,a_2)(U_{i})=1 \wedge 
\ExpRew(s,\dist^{\epsilon}_1,a_2) \geq 1-\epsilon   \wedge  \trans(s,\dist^{\epsilon}_1,a_2)(Y_{i,j}) \geq 1- \epsilon\big) \label{eqt: get 1}\enspace ;
\end{align}
where $\ov{W}^*=S \setminus W^*$ is the complement of $W^*$.
We refer to the above as Equation~\ref{eqt: limit reach}, Equation~\ref{eqt: almost reach}, and
Equation~\ref{eqt: get 1}, respectively.

\smallskip\noindent{\bf The construction of  stationary witness strategy $\sigma_1^\epsilon$.}
Fix $0<\epsilon<\frac{1}{2}$. The desired witness stationary strategy $\sigma_1^\epsilon$ will be constructed 
from a finite sequence of stationary strategies, 
\[
\sigma_1^{\epsilon,1,0},\sigma_1^{\epsilon,1,1},\dots,\sigma_1^{\epsilon,1,\rank(1)},
\sigma_1^{\epsilon,2,0},\dots,\sigma_1^{\epsilon,2,\rank(2)},
\dots,
\sigma_1^{\epsilon,\ell,0}, \dots, \sigma_1^{\epsilon,\ell,\rank(\ell)}.
\] 
The strategies will be constructed inductively. 
First we will construct it for states in $U_1$ and $(U_\ell\setminus U_{\ell-1})$, 
and then we will present the inductive construction for $(U_i \setminus U_{i-1})$, 
for $2\leq i\leq \ell-1$.
\begin{itemize}

\item \emph{(Base case).}
We will first describe the construction of the strategy 
$\sigma_1^{\epsilon,1,0}$ (resp. $\sigma_1^{\epsilon,\ell,0}$). 
\begin{enumerate}
\item The stationary strategy $\sigma_1^{\epsilon,1,0}$ 
(resp. $\sigma_1^{\epsilon,\ell,0}$) is arbitrary except for states in 
$(Y_{1,\rank(1)}\setminus Y_{1,\rank(1)-1})$ 
(resp. $(Y_{\ell,\rank(\ell)}\setminus Y_{\ell,\rank(\ell)-1})$).
 
\item 
For states $s$ in $(Y_{1,\rank(1)}\setminus Y_{1,\rank(1)-1})$ 
(resp. $(Y_{\ell,\rank(\ell)}\setminus Y_{\ell,\rank(\ell)-1})$) 
the strategy plays the distribution $\dist^{\eta}_1$ over $\mov_1(s)$, for 
$\eta=\frac{\epsilon}{2}$.

\item 
We next describe the construction of the strategy $\sigma_1^{\epsilon,1,j}$ (resp. $\sigma_1^{\epsilon,\ell,j}$), for $j\geq 1$, using induction in $j$.
\begin{enumerate} 

\item The strategy $\sigma_1^{\epsilon,1,j}$ (resp. $\sigma_1^{\epsilon,\ell,j}$) plays as 
$\sigma_1^{\epsilon,1,j-1}$ (resp. $\sigma_1^{\epsilon,\ell,j-1}$) except 
for states in $(Y_{1,\rank(1)-j}\setminus Y_{1,\rank(1)-(j+1)})$ 
(resp. $(Y_{\ell,\rank(\ell)-j}\setminus Y_{\ell,\rank(\ell)-(j+1)})$). 

\item For states $s$ in $(Y_{1,\rank(1)-j}\setminus Y_{1,\rank(1)-(j+1)})$ 
(resp. $(Y_{\ell,\rank(\ell)-j}\setminus Y_{\ell,\rank(\ell)-(j+1)})$) 
the strategy plays the distribution $\dist^{\eta}_1$ over $\mov_1(s)$, 
for $\eta=\left(\frac{\epsilon\cdot\delta_{\min}}{4}\right)^{(2m)^{j}}$. 
\end{enumerate}
\end{enumerate}

\item \emph{(Inductive case).}
We will next construct the strategy for the remaining states, in two steps, first for $\sigma_1^{\epsilon,i,0}$ and then for $\sigma_1^{\epsilon,i,j}$, for $2\leq i\leq \ell-1$ and $j\geq 1$. 
We will do so using induction backwards in $i$. That is the base case is $i=\ell$ and we then proceed downward. 
\begin{enumerate}
\item The strategy $\sigma_1^{\epsilon,i,0}$ plays as the strategy $\sigma_1^{\eta,i+1,\rank({i+1})}$, 
for $\eta=\left(\frac{\epsilon\cdot\delta_{\min}}{4}\right)^{(2m)^{\rank(i)}}$, except for states in $(Y_{i,\rank(i)}\setminus Y_{i,\rank(i)-1})$. 

\item For states  $s$ in $Y_{i,\rank(i)}\setminus Y_{i,\rank(i)-1}$ the strategy plays $\dist^{\eta}_1$ over $\mov_1(s)$, for $\eta=\frac{\epsilon}{2}$.

\item 
We now finally construct $\sigma_1^{\epsilon,i,j}$, for $2\leq i\leq \ell-1$, using induction in $j$. 
\begin{enumerate}
\item
The strategy $\sigma_1^{\epsilon,i,j}$ plays as $\sigma_1^{\epsilon,i,j-1}$ except for states in $(Y_{i,\rank(i)-j}\setminus Y_{i,\rank(i)-(j+1)})$.
\item  For states $s$ in $(Y_{i,\rank(i)-j}\setminus Y_{i,\rank(i)-(j+1)})$ 
the strategy plays $\dist^{\eta}_1$ over $\mov_1(s)$, 
for $\eta=\left(\frac{\epsilon\cdot\delta_{\min}}{4}\right)^{(2m)^{j}}$.
\end{enumerate}
\end{enumerate}

\item \emph{(The entire strategy).}
Let $\sigma_1^{\epsilon,i}=\sigma_1^{\epsilon,i,\rank(i)}$ for all $i$. 
Let $\sigma_1^{\epsilon}$ play as $\sigma_1^{\beta,1}$ in $U_1$ and 
$\sigma_1^{\beta,2}$, for $\beta=\frac{\epsilon}{2}$, in the remaining states.
\end{itemize}

\begin{lemma}\label{lemm:patience1}
The patience of $\sigma_1^{\epsilon,i}(s)$ for states $s$ of rank $(i,\rank(i)-j)$ is at most $\left(\frac{\epsilon\cdot\delta_{\min}}{4}\right)^{-(\frac{(2m)^{j+1}}{2}-1)}$.
\end{lemma}

\begin{proof}
By construction, the patience $\sigma_1^{\epsilon,i}(s)$ of states $s$ of rank $(i,\rank(i))$ is
$\left(\frac{\epsilon\cdot \delta_{\min}}{4}\right)^{-(m-1)}$ (by Lemma~\ref{lemm: distribution algorithm accept}). 
Also for $j\geq 1$, the patience  $\sigma_1^{\epsilon,i}(s)$ of states $s$ of rank $(i,\rank(i)-j)$ is at most  
\begin{align*}
\left(\frac{\left(\frac{\epsilon\cdot\delta_{\min}}{4}\right)^{(2m)^{j}}\cdot \delta_{\min}}{2}\right)^{-(m-1)}
&=\left(\frac{\epsilon\cdot\delta_{\min}}{4}\right)^{-(2m)^{j}\cdot (m-1)}\cdot \left(\frac{\delta_{\min}}{2}\right)^{-(m-1)} \\
&=\left(\frac{\epsilon\cdot\delta_{\min}}{4}\right)^{-(2m)^{j}\cdot (m-1)}\cdot \left(\frac{\delta_{\min}}{2}\right)^{-m} \cdot \left(\frac{\delta_{\min}}{2}\right)\\
&=\left(\frac{\epsilon\cdot\delta_{\min}}{4}\right)^{-(2m)^{j}\cdot m}\cdot \left(\frac{\epsilon\cdot\delta_{\min}}{4}\right)^{(2m)^{j}}\cdot \left(\frac{\delta_{\min}}{2}\right)^{-m} \cdot \left(\frac{\delta_{\min}}{2}\right)\\
&\leq \left(\frac{\epsilon\cdot\delta_{\min}}{4}\right)^{-(2m)^{j}\cdot m} \cdot \left(\frac{\epsilon\cdot \delta_{\min}}{4}\right)\\
&=\left(\frac{\epsilon\cdot\delta_{\min}}{4}\right)^{-(\frac{(2m)^{j+1}}{2}-1)} \enspace ,
\end{align*}
where the inequality is as follows: 
$
\left(\frac{\epsilon\cdot\delta_{\min}}{4}\right)^{(2m)^{j}}\cdot \left(\frac{\delta_{\min}}{2}\right)^{-m} = 
\left(\frac{\epsilon}{2}\right)^{(2m)^{j}} \cdot\left(\frac{\delta_{\min}}{2}\right)^{(2m)^{j}}\cdot \left(\frac{\delta_{\min}}{2}\right)^{-m} \leq \frac{\epsilon}{2}$ 
since $(2m)^j \geq m \geq 1$ and $\epsilon<1$.
The desired result follows.
\end{proof}

\begin{lemma}\label{lemm:patience2}
Let $0<\epsilon<\frac{1}{2}$ be given. The patience of the witness stationary 
strategy $\sigma_1^\epsilon$ is less than $\left(\frac{\epsilon\cdot\delta_{\min}}{4}\right)^{-(2m)^{n}}$.
\end{lemma}
\begin{proof} We first present the bound for $U_1$ (also $U_2$) and then for other states.

\smallskip\noindent{\bf The patience of $\sigma_1^{\epsilon,1}$ for states in $U_1$ (also similar for $U_2$).} For each state $s$ in $U_1$, the corresponding distribution $\sigma_1^{\epsilon,1}(s)$   has patience at most $\left(\frac{\epsilon\cdot\delta_{\min}}{4}\right)^{-(\frac{(2m)^{\rank(1)}}{2}-1)}$, since no states are in $Y_{1,0}$. Similarly for $s$ in $U_2$ and the corresponding distribution $\sigma_1^{\epsilon,1}(s)$.

\smallskip\noindent{\bf The  $\eta$ for which the strategy $\sigma_1^{\epsilon,2}$ follows $\sigma_1^{\eta,i}$: Inductive statement.} 
We will argue using induction that for each state $S\in (W^*\setminus U_{i-1})$, for $i\geq 3$, we have that 
the strategy $\sigma_1^{\epsilon,2}$ follows the strategy $\sigma_1^{\eta,i}$, for \[\eta\geq \left(\frac{\epsilon\cdot\delta_{\min}}{4}\right)^{\sum_{k=2}^{i-1}\prod_{k'=k}^{i-1} (2m)^{\rank(k')}}\enspace .\]

\smallskip\noindent{\bf Base case.} For each state $s\in (S\setminus U_{2})$, the strategy $\sigma_1^{\epsilon,2}$ follows the strategy $\sigma_1^{\eta,3}$, for $\eta\geq \left(\frac{\epsilon\cdot\delta_{\min}}{4}\right)^{(2m)^{\rank(2)}}$, by construction, which is the wanted expression.

\smallskip\noindent{\bf Induction case $i+1$.} For $i\geq 4$, for each state $s\in (S\setminus U_{i-1})$, the strategy 
$\sigma_1^{\epsilon,2}$ follows the strategy $\sigma_1^{\eta,i}$, for $\eta\geq \left(\frac{\epsilon\cdot\delta_{\min}}{4}\right)^{\sum_{k=2}^{i-1}\prod_{k'=k}^{i-1} (2m)^{\rank(k')}}$, by induction. In each state $s\in (S\setminus U_{i})$, the strategy $\sigma_1^{\eta,i}$ follows the strategy $\sigma_1^{\eta',i+1}$, for $\eta'\geq \left(\frac{\eta\cdot\delta_{\min}}{4}\right)^{(2m)^{\rank(i)}}$, by construction. Thus, the strategy $\sigma_1^{\epsilon,2}$ follows $\sigma_1^{\eta',i+1}$ for \begin{align*}
\eta'&\geq \left(\frac{\eta\cdot\delta_{\min}}{4}\right)^{(2m)^{\rank(i)}} \\
&\geq \left(\frac{\left(\frac{\epsilon\cdot\delta_{\min}}{4}\right)^{\sum_{k=2}^{i-1}\prod_{k'=k}^{i-1} (2m)^{\rank(k')}}\cdot\delta_{\min}}{4}\right)^{(2m)^{\rank(i)}} \\
& \geq \left(\left(\frac{\epsilon\cdot\delta_{\min}}{4}\right)^{1+\sum_{k=2}^{i-1}\prod_{k'=k}^{i-1} (2m)^{\rank(k')}}\right)^{(2m)^{\rank(i)}} \\
&= \left(\frac{\epsilon\cdot\delta_{\min}}{4}\right)^{\sum_{k=2}^{i}\prod_{k'=k}^{i} (2m)^{\rank(k')}} \enspace .
\end{align*}
The first inequality comes from our preceding explanation. The second inequality uses the inductive hypothesis. The third uses that $\frac{\delta_{\min}}{4}>\frac{\epsilon\cdot \delta_{\min}}{4}$. The last equality is the inductive hypothesis for $i+1$ and follows from \begin{align*}
(2m)^{\rank(i)}+(2m)^{\rank(i)}\cdot \sum_{k=2}^{i-1}\prod_{k'=k}^{i-1} (2m)^{\rank(k')} & = (2m)^{\rank(i)}+\sum_{k=2}^{i-1}\prod_{k'=k}^{i} (2m)^{\rank(k')} \\
 &= \sum_{k=2}^{i}\prod_{k'=k}^{i} (2m)^{\rank(k')} \enspace .
\end{align*}

\smallskip\noindent{\bf Patience of $\sigma_1^{\epsilon,2}(s)$ for states in $U_i$, for $i\geq 3$.} We see that for $i\geq 3$ and for each $s$ in $U_i$ we have that $\sigma_1^{\eta,i}(s)$ follows $\dist_1^{\eta'}$ for $\eta'\geq \left(\frac{\eta\cdot\delta_{\min}}{4}\right)^{(2m)^{\rank(i)-1}}$ (since $Y_{i,0}$ is empty), by construction. Hence, we get that $\sigma_1^{\epsilon,2}(s)=\dist_1^{\eta'}$  for $\eta'\geq \left(\frac{\epsilon\cdot\delta_{\min}}{4}\right)^{\frac{\sum_{k=2}^{i}\prod_{k'=k}^{i} (2m)^{\rank(k')}}{2m}}$, using a similar argument as the one used in the inductive case.
Since $\rank(i)\geq 1$ and $m\geq 1$, we see that each term in the sum $\sum_{k=2}^{i}\prod_{k'=k}^{i} (2m)^{\rank(k')}$ is at least twice as large as the following. Thus, we have that 
\begin{align*}
\sum_{k=2}^{i}\prod_{k'=k}^{i} (2m)^{\rank(k')}<2\cdot \prod_{k'=2}^{i} (2m)^{\rank(k')}=2\cdot (2m)^{\sum_{k'=2}^{i}\rank(k')}\leq 2\cdot (2m)^{n-1}\leq (2m)^{n} \enspace .
\end{align*}
The first inequality is because $U_1$ must contain at least 1 state. The second comes from $m\geq 1$. 
Hence, $\eta'\geq \left(\frac{\epsilon\cdot\delta_{\min}}{4}\right)^{(2m)^{n-1}}$. Using an argument similar to the one used to prove Lemma~\ref{lemm:patience1}, we get that the patience for $\dist_1^{\eta'}$ is then at most $\left(\frac{\epsilon\cdot\delta_{\min}}{4}\right)^{-(\frac{(2m)^{n}}{2}-1)}$. 

\smallskip\noindent{\bf Patience of $\sigma_1^{\epsilon}$.} We now need to consider the strategy $\sigma_1^{\epsilon}$. It follows $\sigma_1^{\beta,1}$ in $U_1$ and $\sigma_1^{\beta,2}$ elsewhere, for $\beta=\frac{\epsilon}{2}$, 
We see that 
\begin{align*}
\left(\frac{\beta\cdot\delta_{\min}}{4}\right)^{-(\frac{(2m)^{n}}{2}-1)}&=\left(\frac{\epsilon\cdot\delta_{\min}}{8}\right)^{-(\frac{(2m)^{n}}{2}-1)} \\
&<\left(\frac{\epsilon\cdot\delta_{\min}}{4}\right)^{-(2m)^{n}} 
\end{align*}
The inequality is because $4^2=16>8$ (and the last expression more than squares the preceding). 
This completes the proof.
\end{proof}

\smallskip\noindent{\bf Basic overview of the proof.} We first present the basic 
overview of the proof.
Let $\sigma_1$ be a stationary strategy that follows distribution $\dist^{\eta}_1$ over $\mov_1(s)$ in 
state $s \in W^*$ for some $\eta>0$ and let $\sigma_2$ be a positional counter-strategy for player~2. 
For state $s$ in $W^*$, $\sigma_1(s)$ and $\sigma_2(s)$ satisfies at least one of Equation~\ref{eqt: limit reach}, Equation~\ref{eqt: almost reach}, 
or Equation~\ref{eqt: get 1} in $s$. 
Let $\CC_1^{\sigma_1,\sigma_2} \subseteq W^*$ 
(resp. $\CC_2^{\sigma_1,\sigma_2} \subseteq W^*$ and $\CC_3^{\sigma_1,\sigma_2} \subseteq W^*$) 
be the set of states in $W^*$ that satisfies Equation~\ref{eqt: limit reach}
(resp. Equation~\ref{eqt: almost reach} and Equation~\ref{eqt: get 1}).
We will prove that $\sigma_1^{\epsilon}$ ensures value at least $1-\epsilon$ for each states $s$ in $W^*$.
We will split the proof into four parts, first we will show some properties 
for states in $U_1$, then for states in $U_\ell\setminus U_{\ell-1}$, 
and finally for states in $U_i\setminus U_{i-1}$ for $2\leq i\leq \ell-1$.
In the fourth part, we will then combine the three properties to establish 
the desired result.
The three properties are as follows 
\begin{itemize}
\item \emph{(Property~1).} For all states $s$ in $U_1$ we will show that $\sigma_1^{\epsilon,1}$ ensures
$\Safe(U_1)$ with probability~1 and mean-payoff at least $1-\epsilon$
(i.e., for all positional strategies $\sigma_2$ we have 
$\lim_{t\rightarrow \infty} \frac{\sum_{i=0}^t \E^{\sigma_1^{\epsilon,1},\sigma_2}_s[\Theta_i]}{t}\geq 1-\epsilon$).
%%In other words, both safety in $U_1$ and $\RR(s,\epsilon)$ is guaranteed.
\item \emph{(Property~2).}
For all states $s$ in $(U_{\ell} \setminus U_{\ell-1})$ we will show that 
$\sigma_1^{\epsilon,\ell}$ ensures that against all positional strategies $\sigma_2$ 
we have that 
\begin{enumerate}
\item given the event $\Safe(U_{\ell}\setminus U_{\ell-1})$, the mean-payoff is at least $1-\epsilon$;
\item $\Pr_s^{\sigma_1^{\epsilon,\ell},\sigma_2}(\Safe(U_{\ell}\setminus U_{\ell-1})\cup \Reach(U_{\ell-1} \cup \ov{W}^*))=1$; and
\item $\Pr_s^{\sigma_1^{\epsilon,\ell},\sigma_2}(\Safe(U_{\ell}\setminus U_{\ell-1})\cup \Reach(U_{\ell-1}))\geq 1-\epsilon$.
\end{enumerate}
 
\item \emph{(Property~3).} 
For all states $s$ in $(U_{\ell} \setminus U_{\ell-(i+1)})$, for $1 \leq i \leq \ell-2$, we will show that 
$\sigma_1^{\epsilon,i}$ ensures that against all positional strategies $\sigma_2$ 
we have that
\begin{enumerate}
\item given the event $\bigcup_{j\leq i}\coBuchi(U_{\ell-j}\setminus U_{\ell-(j+1)})$,
 the mean-payoff is at least $1-\epsilon$;
 \item $\Pr_s^{\sigma_1^{\epsilon,\ell-i},\sigma_2}(\bigcup_{j \leq i}\coBuchi(U_{\ell-j}\setminus U_{\ell-(j+1)})\cup \Reach(U_{\ell-(i+1)}\cup \ov{W}^*))=1$; and
 \item $\Pr_s^{\sigma_1^{\epsilon,\ell-i},\sigma_2}(\bigcup_{j \leq i}\coBuchi(U_{\ell-j}\setminus U_{\ell-(j+1)})\cup \Reach(U_{\ell-(i+1)}))\geq 1-\epsilon$.
\end{enumerate} 

\end{itemize}
In Lemma~\ref{lemm:U_1}, Lemma~\ref{lemm:U_ell}, and Lemma~\ref{lemm:U_i} 
we establish Properties~1,~2, and~3, respectively.
We first present the basic intuition of the proof of Lemma~\ref{lemm:U_1}.

\smallskip\noindent{\bf The basic intuition of Lemma~\ref{lemm:U_1}.}
The key idea of the proof is as follows.
Once we fix the strategies for both the players we have a Markov chain.
Let $\CC_2$ and $\CC_3$ denote the set of states in $U_1$ that satisfy
Equation~\ref{eqt: almost reach} and  Equation~\ref{eqt: get 1},  respectively.
Since $U_0$ is empty, no state in $U_1$ can satisfy Equation~\ref{eqt: limit reach}.
For states $s$ in $\CC_2$ of rank $(1,j)$, the fact that Equation~\ref{eqt: almost reach} is satisfied 
ensures that a state of rank $(1,j')$, for $j'<j$, is visited from $s$ with positive probability.
Let $\patc(j)$ denote the patience of the strategy $\sigma_1^{\epsilon,1}$ 
for states of rank $(1,\rank(1)-j)$. We now consider the following case analysis.
\begin{enumerate}
\item 
First we consider the set of states in $(Y_{1,\rank(1)}\setminus Y_{1,\rank(1)-1})$
and show that if we stay in the set $(Y_{1,\rank(1)}\setminus Y_{1,\rank(1)-1})$,
then the mean-payoff is at least $1-\epsilon$. 
The argument is as follows:
By Markov property~\ref{markov property recurrent class}, we must reach a recurrent class with probability 1.
A recurrent class contained in $(Y_{1,\rank(1)}\setminus Y_{1,\rank(1)-1})$ must 
consist of only states in $\CC_3$ (since from states in $\CC_2$ we reach 
lower rank states with positive probability), 
and since Equation~\ref{eqt: get 1} is satisfied for states in $\CC_3$ it 
follows that the mean-payoff value is at least $1-\epsilon$.
Hence, if we have a recurrent class of the Markov chain contained in $(U_1 \setminus Y_{1,\rank(1)-1})=
(Y_{1,\rank(1)}\setminus Y_{1,\rank(1)-1})$, 
then the mean-payoff of the recurrent class is at least $1-\epsilon$.
This completes the argument.
Also, if the set $(Y_{1,\rank(1)}\setminus Y_{1,\rank(1)-1})$ is left, then we can \emph{bound} 
the number of visits to states in $\CC_2$ (and in the worst case each such visit gives reward~0)  
in expectation encountered before leaving the set $(Y_{1,\rank(1)}\setminus Y_{1,\rank(1)-1})$.
This bound on the number of visits in expectation to $\CC_2$ 
(which we say has not been accounted for by visits to $\CC_3$) is 
%%at most 
$\kappa(0)=(\trans_{\min})^{-1} \cdot \patc(0)$. 
There is an illustration of this base case in Figure~\ref{fig:base_U_1}.

\item
Now we consider that we are at some intermediate part of the computation, i.e.,
in some state in  $(Y_{1,\rank(1)-j}\setminus Y_{1,\rank(1)-(j+1)})$, for $j\geq 1$.
Inductively we have an upper bound $\kappa(j)$ on the number of times that states
in $\CC_2$ were visited (in the worst case each such visit gives reward~0) in 
expectation that has not been accounted for by visits to states in $\CC_3$ 
till we reach the set  $(Y_{1,\rank(1)-j}\setminus Y_{1,\rank(1)-(j+1)})$ from any state in $Y_{1,\rank(1)-j+1}$.
The one-step probability distribution $\dist_1^\eta$ is chosen such that 
$\eta \cdot \kappa(j) \leq \epsilon$.
In other words, $\eta$ decreases rapidly as $i$ increases, and 
the small $\eta$ ensures that if the play stays in 
$(U_1 \setminus Y_{1,\rank(1)-(j+1)})$, then the mean-payoff is at least $1-\epsilon$,
i.e., if we have a recurrent class $L$ contained in $(U_1 \setminus Y_{1,\rank(1)-(j+1)})$ 
and $(L \cap Y_{1,\rank(1)-j})$ is non-empty, then all states in  $(L \cap Y_{1,\rank(1)-j})$ belong 
to $\CC_3$, and the mean-payoff of the recurrent class is at least $1-\epsilon$.
Moreover, we can also upper bound the number of visits to states in $\CC_2$ in expectation
that has not been accounted for by visits to states in $\CC_3$ before
reaching the set  $Y_{1,\rank(1)-(j+1)}$ if we leave $(U_1 \setminus Y_{1,\rank(1)-(j+1)})$ by
$\kappa(j+1)=(\kappa(j)+1)\cdot (\trans_{\min})^{-1} \cdot \patc(j)$, 
and then proceed inductively. 
There is an illustration of this inductive case in Figure~\ref{fig:induction_U_1}.
\end{enumerate}

\begin{figure}
\begin{center}

\begin{tikzpicture}[node distance=3cm]
\tikzstyle{every state}=[fill=white,draw=black,text=black,font=\small , inner sep=0pt,minimum size=4mm]
\tikzstyle{txt}=[state,draw=white]
\foreach \x in {-3,3,6} {\draw (\x,0) -- (\x,-6);  }
\draw (-3,0) -- (6,0); 
\draw (-3,-6) -- (6,-6);

\node[state] (w2) at (4.5,-1){$w_2$};
\node[txt] (w2to) at (0,-1){};
\path[->] (w2) edge node[fill=white,above]{$\Pr=\delta_{\min}\cdot \epsilon$} (w2to);
\path[->] (w2) edge[loop below] node[below, fill=white]{$\Pr=1-\delta_{\min}\cdot \epsilon$} (w2);
\path[->] (w2) edge (w2to);

\node[state] (w3) at (4.5,-4){$w_3$};
\path[->] (w3) edge[loop above] node[above,fill=white]{$\Pr=1-\epsilon$,$\cost=1$} (w3);
\path[->] (w3) edge[loop below] node[below, fill=white]{$\Pr=\epsilon$} (w3);

\draw[<-] (-3,0.25) -- node[very near end,above] {$Y_{1,\rank(1)}=U_1$}(6,0.25);
\draw[<-] (-3,0.5) -- node[above] {$Y_{1,\rank(1)-1}$}(3,0.5);
\end{tikzpicture}
\end{center}
\caption{Pictorial illustration of the intuitive explanation of the base case of Lemma~\ref{lemm:U_1}.}\label{fig:base_U_1}
\end{figure}
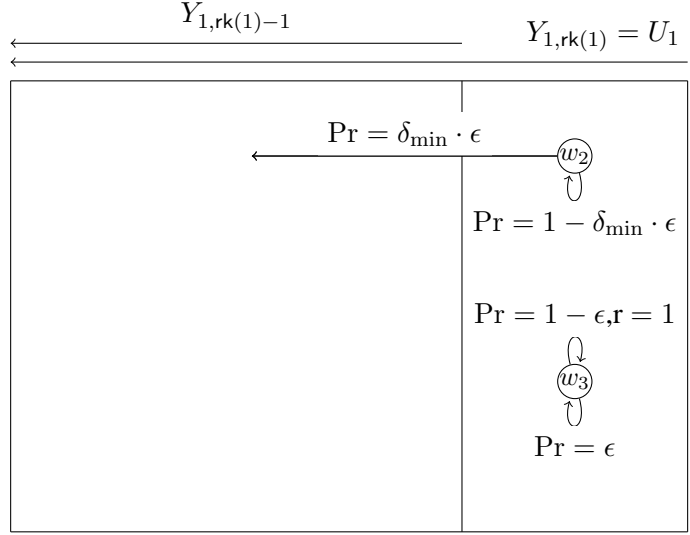

\begin{figure}
\begin{center}

\begin{tikzpicture}[node distance=3cm]
\tikzstyle{every state}=[fill=white,draw=black,text=black,font=\small , inner sep=0pt,minimum size=4mm]
\tikzstyle{txt}=[state,draw=white]
\foreach \x in {-3,3,6,12} {\draw (\x,0) -- (\x,-5);  }
\draw (-3,0) -- (12,0); 
\draw (-3,-5) -- (12,-5);

\node[state] (w2) at (4.5,-1){$w_2$};
\node[txt] (w2good) at (-1,-1){};
\node[txt] (w2bad) at (10,-1){};
\path[->] (w2) edge node[fill=white,above]{$\Pr=\eta$} (w2good);
\path[->] (w2) edge node[above, fill=white]{$\Pr=1-\eta$} (w2bad);
\path[->] (w2) edge (w2good);
\path[->] (w2) edge (w2bad);

\node[state] (w3) at (4.5,-4){$w_3$};
\node[txt] (w3bad) at (10,-4){};
\path[->] (w3) edge[loop above] node[above,fill=white]{$\Pr=1-\eta$} node[below left]{$\cost=1$}(w3);
\path[->] (w3) edge node[below, fill=white]{$\Pr=\eta$} (w3bad);
\path[->] (w3) edge  (w3bad);

\node[state] (back) at (10,-2.5){};
\node[txt] (backTo) at (5,-2.5){};
\path [->, decoration={zigzag,segment length=4,amplitude=.9,
  post=lineto,post length=0.25cm,pre length=0.25 cm}] (back) edge[decorate] node[above] {$\kappa(i)\times (\#\CC_2)$}(backTo);

\draw[<-] (-3,0.5) -- node[very near end,above] {$Y_{1,\rank(1)-i}$}(6,0.5);
\draw[<-] (-3,0.75) -- node[above] {$Y_{1,\rank(1)-(i+1)}$}(3,0.75);
\draw[<-] (-3,0.25) -- node[very near end,above] {$U_1$}(12,0.25);
\end{tikzpicture}
\end{center}
\caption{Pictorial illustration of the intuitive explanation of the inductive case of Lemma~\ref{lemm:U_1}.}\label{fig:induction_U_1}
\end{figure}
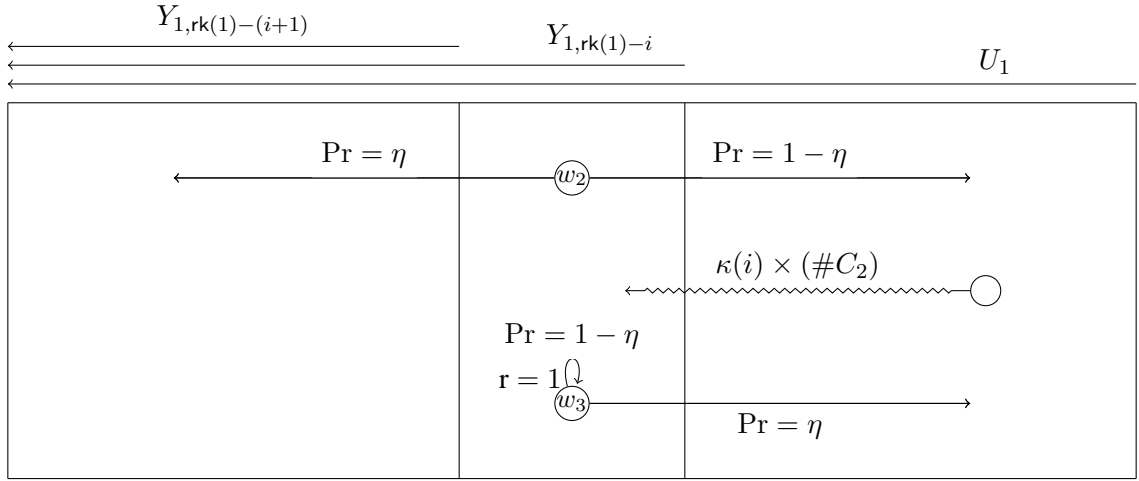

\begin{lemma}\label{lemm:U_1}{\em (Property~1).} 
Let $0<\epsilon<\frac{1}{2}$.
The strategy $\sigma_1^{\epsilon,1}$ ensures that for all $s \in U_1$ and all positional strategies $\sigma_2$ for player~2 
we have $\Pr_s^{\sigma_1^{\epsilon,1},\sigma_2}(\Safe(U_1))=1$ 
and $\lim_{t\rightarrow \infty} \frac{\sum_{i=0}^t \E^{\sigma_1^{\epsilon,1},\sigma_2}_s[\Theta_i]}{t}\geq 1-\epsilon$.
\end{lemma}

\begin{proof}
Given $\sigma_1^{\epsilon,1}$, let $\sigma_2$ be an arbitrary positional counter-strategy
for player~2.
Let $\CC_i^{\sigma_1^{\epsilon,1},\sigma_2}\cap U_1=\CC_i$,
i.e., given $\sigma_1^{\epsilon,1}$ and $\sigma_2$,
we have that $\CC_1, \CC_2, \CC_3$ are the set of states of $U_1$ that 
satisfy Equation~\ref{eqt: limit reach},  
Equation~\ref{eqt: almost reach},  
Equation~\ref{eqt: get 1},  respectively.
Notice that since $U_0$ is the empty set we have that $\CC_1$ is also empty.
Therefore we cannot leave $U_1$ if player~1 follows $\sigma_1^{\epsilon,1}$ (because both Equation~\ref{eqt: almost reach} and Equation~\ref{eqt: get 1} require that we stay in $U_1$).
This ensures that $\Safe(U_1)$ is satisfied with probability~1.
We now focus on the mean-payoff. 
%Since  $\sigma_1^{\epsilon,1}$ is stationary and the set $U_1$ is never left,
%we only need to consider positional counter-strategies $\sigma_2$ for player~2.
%We fix a positional strategy $\sigma_2$ for player~2.

\smallskip\noindent{\em Basic notations.}
Let us consider the Markov chain obtained given  $\sigma_1^{\epsilon,1}$ and 
$\sigma_2$.
For a state $s \in U_1$, let the rank of $s$ be $\rank(s)=(1,j)$, and then 
we denote $j$ by $\rank_2(s)$ (the second component of the rank).
Given a play $P$ in the Markov chain, and a number $t \in \N$,  
let $\widetilde{r}(P,t)$ be the expected number of times we get reward $0$ in the first 
$t$ steps of $P$. This implies that $\widetilde{r}(P,0)=0$.
%%We will also denote by $|P|$ the expected length of the play $P$.
%(if the play has ended before step $t$, then it is the expected number of reward $0$ in the play).
%let $|P|$ be the (posible infinite) expected length of $P$ and let $\widetilde{r}(P,t)$ be the expected mean payoff of the first $t$ steps of $P$, where $1\leq t\leq |P|$ (or strict if the walk has infinite expected length) .
For each state $s\in U_1$, %consider the Markov chain, which starts in $s$ and where the players follows their respective strategy. 
let $P_s^{j}$ be (a prefix of) a play in the Markov chain, which ends if a 
state in $Y_{1,j}$ is reached after the starting point $s$ (i.e., the play does not 
end at $s$ if $s \in Y_{1,j}$), and if $Y_{1,j}$ is not reached, then the walk does not end. 
We will also use the following notations: for $0 \leq j \leq \rank(1)-1$, 
let us denote by $\kappa(j+1)=\frac{\epsilon}{2}\cdot \left(\frac{\epsilon\cdot\delta_{\min}}{4}\right)^{-(2m)^{j+1}}$;  and let $\patc(j)= \left(\frac{\epsilon\cdot\delta_{\min}}{4}\right)^{-(\frac{(2m)^{j+1}}{2}-1)}$,
the patience of $\sigma_1^{\epsilon,1}$ for states in $U_1$ of rank $(1,\rank(1)-j)$ (by Lemma~\ref{lemm:patience1}).

\smallskip\noindent{\em Using recurrent class property.}
First, observe that since $Y_{1,0}$ is the empty set, 
the set $Y_{1,0}$ can never be reached, and hence $P_s^0$ represents the 
entire play from the start state $s$, for $s \in U_1$.
%Therefore it suffices to show that $\inf_{t\rightarrow \infty}\widetilde{r}(P_s^0,t)\geq 1-\epsilon$, 
%for each state $s\in U_1$. 
By Markov property~\ref{markov property recurrent class}
in the Markov chain, the recurrent classes are reached in a finite number 
of steps with probability~1, and given a recurrent class $L$ is reached,
every state in $L$ is reached with probability~1 
in a finite number of steps.
Given a recurrent class $L$ in $U_1$, and consider a state $s^*$ in $L$ that has the maximum 
rank among states in $L$ (i.e., $\rank_2(s^*)= \max_{s' \in L} \rank_2(s')$). 
Then all states visited after $s^*$ has rank at most the rank of $s^*$.
Hence every play $P_s^0$  with probability~1, after finitely many steps reaches a  state $s^*$ 
such that all states $s'$ visited after $s^*$ satisfy that $\rank_2(s')\geq \rank_2(s^*)$.
Since the mean-payoff is invariant under finite prefixes, we only need to obtain bounds for 
the mean-payoff of $P_{s^*}^{\rank(s^*)-1}$ (and this play has infinite length 
by definition as no state with smaller rank is reached in the Markov chain after $s^*$).

\smallskip\noindent{\em Inductive proof statement.}
We will show, inductively, that for all $0\leq j\leq \rank(1)$, all $t\geq 1$, 
and all states $s\in U_1$, if $\rank_2(s)=\rank(1)-j$, then 
\begin{align*}
\widetilde{r}(P_s^{\rank_2(s)-1},t)\leq t\cdot  \epsilon+ \frac{\kappa(j+1)}{2} =
t\cdot  \epsilon+\frac{\epsilon}{4}\cdot \left(\frac{\epsilon\cdot\delta_{\min}}{4}\right)^{- (2m)^{j+1}}
\end{align*}  
This will imply the desired result, since then the mean-payoff of 
$P_{s^*}^{\rank_2(s^*)-1}$ is at least $1-\epsilon$: 
the play $P_{s^*}^{\rank_2(s^*)-1}$ has infinite length and therefore the expected number of 
reward 1's must be $t-\widetilde{r}(P_{s^*}^{\rank_2(s^*)-1},t)$ in the first $t$ steps for all $t$, 
because all rewards are either~0 or~1, and hence the mean-payoff of 
$P_{s^*}^{\rank_2(s^*)-1}$ is $\inf_{t\rightarrow \infty}\frac{t-\widetilde{r}(P_{s^*}^{\rank_2(s^*)-1},t)}{t} \geq 1-\epsilon$.

\smallskip\noindent{\em Splitting the play.}
Consider a play $P_{s}^{\rank_2(s)-1}$ for $s \in U_1$.
We will split up the play $P_{s}^{\rank_2(s)-1}$ into a (possible infinite) 
sequence of \emph{rank preserving} plays $(P_{s_i}^{\rank_2(s_i)})_{i\geq 0}$, 
such that $s_0=s$, and for $i\geq 0$, the play $P_{s_i}^{\rank_2(s_i)}$ ends in state $s_{i+1}$ 
(which is formally a random variable and must be such that $\rank_2(s_i)=\rank_2(s_{i+1})$ by definition of $P_{s_i}^{\rank_2(s_i)}$ and 
since if a state of lower rank than $\rank_2(s)$ is reached, then the play  $P_{s}^{\rank_2(s)-1}$ 
ends). 
In other words, the next play begins where the previous play ends, and 
all the starting points of the play has the same rank.
Similarly, we will split up plays $P_{s}^{j}$, for $0\leq j< \rank_2(s)$, into a 
finite sequence of \emph{rank decreasing} plays $(P_{s_i}^{\rank_2(s_i)-1})_{i\geq 0}$, 
such that $s_0=s$, and for $i\geq 0$, the play $P_{s_i}^{\rank_2(s_i)-1}$ ends in state $s_{i+1}$ 
(which must be such that $\rank_2(s_i) > \rank_2(s_{i+1})>j$).
Note that since the play sequence is decreasing, the sequence of plays is finite
and the length of the sequence is at most $\rank_2(s)-j$.
Pictorial illustrations of rank preserving (both when the sequence is finite and infinite) 
and rank decreasing plays are given in Figure~\ref{fig:fin rank preserving}, Figure~\ref{fig:inf rank preserving}, 
and Figure~\ref{fig: rank decreasing}, respectively.

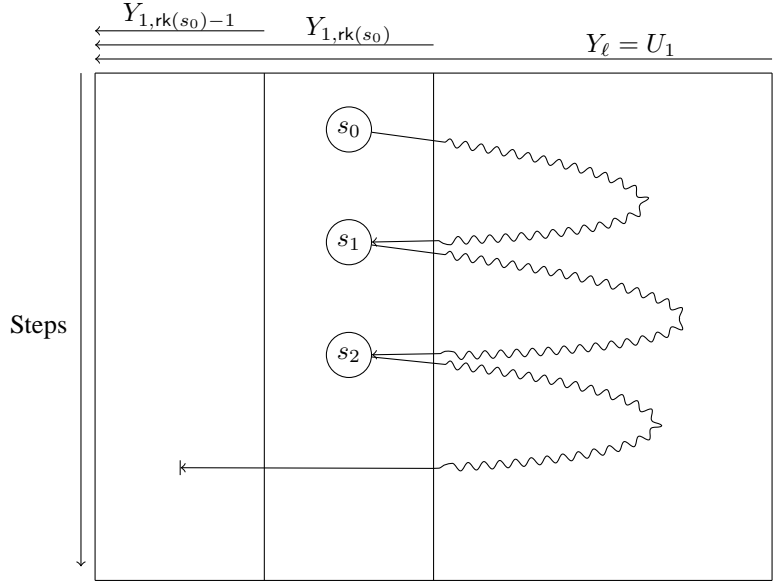
\begin{figure}
\begin{center}

\def\scalefactor{0.75}
\begin{tikzpicture}[scale=\scalefactor,node distance=3cm,decoration={snake,amplitude=2*\scalefactor,segment length=8*\scalefactor, post length=\scalefactor*1.1cm,pre length=\scalefactor*1.3cm}]
\tikzstyle{every state}=[fill=white,draw=black,text=black,font=\small , inner sep=0pt,minimum size=6mm]
\tikzstyle{txt}=[state,draw=white]

\foreach \x / \y / \z in {$Y_{1,\rank(s_0)-1}$/1.5/1,$Y_{1,\rank(s_0)}$/4.5/0.75,$Y_{\ell}=U_1$/9.5/0.5} {\node[txt]  at (\y,\z) {\x};  }
\foreach \y / \z in {3/0.75,6/0.50,12/0.25} {\draw[<-] (0,\z) -- (\y,\z); }

\foreach \x in {0,3,6,12} {\draw (\x,0) -- (\x,-9);  }
\draw (0,0) -- (12,0); 
\draw (0,-9) -- (12,-9);

\draw [->] (-0.25,0) -- (-0.25,-8.75);
\node at 	(-1,-4.5) [txt] (s1) {Steps};

\node at 	(4.5,-1) [state] (s1) {$s_0$}; 
\node at 	(4.5,-3) [state] (s2) {$s_1$}; 
\node at 	(4.5,-5) [state] (s3) {$s_2$};

\draw [decorate,->]  (s1) .. controls +(-7:6.8cm) and +(0.5:6.8cm) .. (s2);
\draw [decorate,->]  (s2) .. controls +(-7:7.7cm) and +(0.5:7.7cm) .. (s3);
\draw [decorate,->|,decoration={snake,amplitude=\scalefactor*2,segment length=\scalefactor*8,post length=\scalefactor*4.5cm,pre length=\scalefactor*1.3cm}]  (s3) .. controls +(-5:8cm) and +(-2:10cm) .. (1.5,-7);

\end{tikzpicture}
\end{center}
\caption{Pictorial illustration of a play $P_{s_0}^{\rank(s_0)-1}$ split into a finite sequence $\left(P_{s_i}^{\rank(s_i)}\right)_{i\geq 0}$ of rank preserving plays. Straight line segments indicate that all states are shown on them, while non-straight segements indicate that there might be states which are not shown. }\label{fig:fin rank preserving}
\end{figure}

\begin{figure}
\begin{center}

\def\scalefactor{0.75}
\begin{tikzpicture}[scale=\scalefactor,node distance=3cm,decoration={amplitude=\scalefactor*2,segment length=\scalefactor*8,snake,post length=\scalefactor*1.1cm,pre length=\scalefactor*1.3cm}]
\tikzstyle{every state}=[fill=white,draw=black,text=black,font=\small , inner sep=0pt,minimum size=6mm]
\tikzstyle{txt}=[state,draw=white]

\foreach \x / \y / \z in {$Y_{1,\rank(s_0)-1}$/1.5/1,$Y_{1,\rank(s_0)}$/4.5/0.75,$Y_{\ell}=U_1$/9.5/0.5} {\node[txt]  at (\y,\z) {\x};  }
\foreach \y / \z in {3/0.75,6/0.50,12/0.25} {\draw[<-] (0,\z) -- (\y,\z); }

\foreach \x in {0,3,6,12} {\draw (\x,0) -- (\x,-8);  }
\draw (0,0) -- (12,0); 
\draw (0,-8) -- (12,-8);

\draw [->] (-0.25,0) -- (-0.25,-7.75);
\node at 	(-1,-4) [txt] (s1) {Steps}; 

\node at 	(4.5,-1) [state] (s1) {$s_0$}; 
\node at 	(4.5,-3) [state] (s2) {$s_1$}; 
\node at 	(4.5,-5) [state] (s3) {$s_2$}; 
\node at 	(4.5,-7) [txt] (s4) {$\vdots$};

\draw [decorate,->]  (s1) .. controls +(-7:7.5cm) and +(0:7.5cm) .. (s2);
\draw [decorate,->]  (s2) .. controls +(-7:7cm) and +(0:7cm) .. (s3);
\draw [decorate,->]  (s3) .. controls +(-7:8cm) and +(0:8cm) .. (s4.east);

\end{tikzpicture}
\end{center}
\caption{Pictorial illustration of a play $P_{s_0}^{\rank(s)-1}$ split into an infinite sequence $\left(P_{s_i}^{\rank(s_i)}\right)_{i\geq 0}$ of rank preserving plays. Note that the last play could be infinite (which is not pictorially illustrated). Straight line segments indicate that all states are shown on them, while non-straight segements indicate that there might be states which are not shown. }\label{fig:inf rank preserving}
\end{figure}
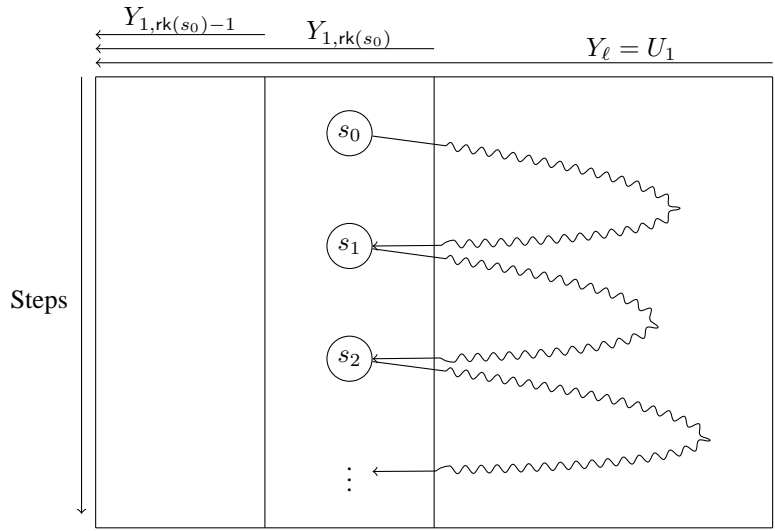

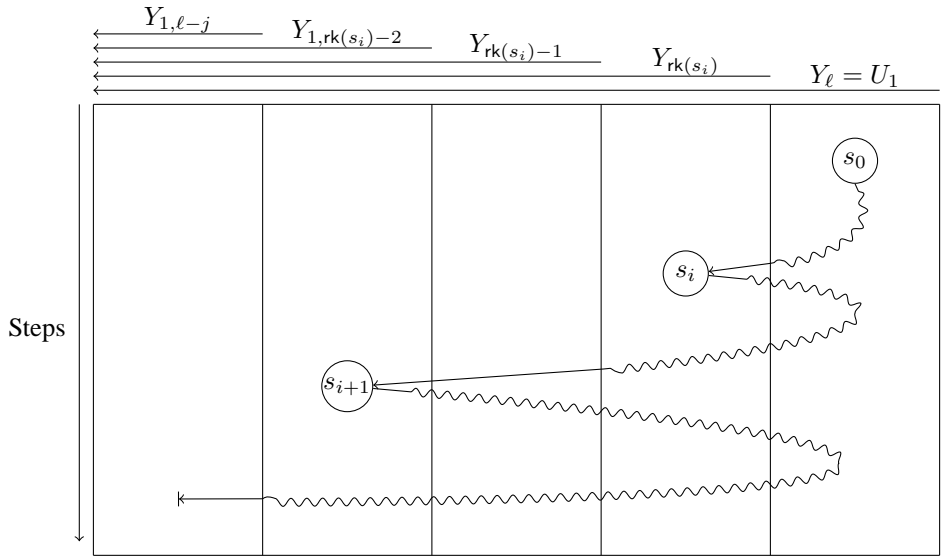
\begin{figure}
\begin{center}
\def\scalefactor{0.75}
\begin{tikzpicture}[scale=\scalefactor,node distance=3cm,decoration={amplitude=\scalefactor*2,segment length=\scalefactor*8,snake,post length=\scalefactor*1.3cm,pre length=\scalefactor*1.3cm}]
\tikzstyle{every state}=[fill=white,draw=black,text=black,font=\small , inner sep=0pt,minimum size=6mm]
\tikzstyle{txt}=[state,draw=white]
\foreach \x / \y / \z in {$Y_{1,\ell-j}$/1.5/1.5,$Y_{1,\rank(s_i)-2}$/4.5/1.25,$Y_{\rank(s_i)-1}$/7.5/1,$Y_{\rank(s_i)}$/10.5/0.75,$Y_{\ell}=U_1$/13.5/0.5} {\node[txt]  at (\y,\z) {\x};  }
\foreach \y / \z in {3/1.25,6/1,9/0.75,12/0.5,15/0.25} {\draw[<-] (0,\z) -- (\y,\z); }

\foreach \x in {0,3,6,9,12,15} {\draw (\x,0) -- (\x,-8);  }
\draw (0,0) -- (15,0); 
\draw (0,-8) -- (15,-8);
\draw [->] (-0.25,0) -- (-0.25,-7.75);
\node at 	(-1,-4) [txt] (s1) {Steps};

\node at 	(13.5,-1) [state] (s1) {$s_0$}; 
\node at 	(10.5,-3) [state] (s2) {$s_{i}$}; 
\node at 	(4.5,-5) [state] (s3) {$s_{i+1}$};

\draw [decorate,->,decoration={amplitude=\scalefactor*2,segment length=\scalefactor*8,snake,post length=\scalefactor*1.1cm,pre length=.1cm}]  (s1) .. controls +(-90:.5cm) and +(5:4cm) .. (s2);
\draw [decorate,->,decoration={amplitude=\scalefactor*2,segment length=\scalefactor*8,snake,post length=\scalefactor*4.1cm,pre length=.5cm}]  (s2) .. controls +(-5:6cm) and +(2.5:8cm) .. (s3);
\draw [decorate,->|,decoration={amplitude=\scalefactor*2,segment length=\scalefactor*8,snake,post length=\scalefactor*1.5cm,pre length=0.5cm}]  (s3) .. controls +(-5:12cm) and +(0:15cm) .. (1.5,-7);

\end{tikzpicture}
\end{center}
\caption{Pictorial illustration of a play $P_{s_0}^{\ell-j}$ split into a (always finite) sequence $\left(P_{s_i}^{\rank(s_i)-1}\right)_{i\geq 0}$ of rank decreasing plays. Note that the last play could be infinite (which is not pictorially illustrated). Straight line segments indicate that all states are shown on them, while non-straight segements indicate that there might be states which are not shown. }\label{fig: rank decreasing}
\end{figure}

\smallskip\noindent{\bf (Base case).} 
We first consider the base case, where $j=0$, i.e., we consider $s$ 
such that $\rank_2(s)=\rank(1)$. 
Consider the rank preserving split up of the play $P_{s}^{\rank_2(s)-1}$ into 
the sequence of plays $(P_{s_i}^{\rank_2(s_i)})_{i\geq 0}$,
mentioned above.
As already mentioned, safety in $U_1=Y_{1,\rank(1)}$ is guaranteed, and 
hence each play $P_{s_i}^{\rank_2(s_i)}$ has length~1.
We will consider $\widetilde{r}(P_{s'}^{\rank_2(s')},t)$, for all $s'$ such that $\rank(s')=\rank(s)$. 
We will now split the proof into the following two cases: 
(1)~$s'\in \CC_2$; and (2)~$s'\in \CC_3$; (as already argued at the start of the proof of this lemma, 
the set $\CC_1$ is empty).

\begin{enumerate}
\item 
In each state $s'$ in $(\CC_2\cap (Y_{1,\rank(1)}\setminus Y_{1,\rank(1)-1}))$ 
we reach a state $s''$ of rank $\rank_2(s'')=\rank_2(s)-1$ in the next step
with probability at least $\left(\frac{\epsilon\cdot \delta_{\min}}{4}\right)^{m-1} \cdot\delta_{\min}=\frac{4}{\epsilon}\cdot \left(\frac{\epsilon\cdot \delta_{\min}}{4}\right)^{m}$ (since $\left(\frac{\epsilon\cdot \delta_{\min}}{4}\right)^{-(m-1)}$ is an upper bound on the patience of states of rank $(1,\rank(1))$ in $\sigma_1^{\epsilon,1}$ by Lemma~\ref{lemm:patience1}), 
otherwise we reach a state of rank $\rank(s)$. 
Hence the expected number of visits to states in $\CC_2$ is at most 
$\frac{\epsilon}{4}\cdot \left(\frac{\epsilon\cdot \delta_{\min}}{4}\right)^{-m}$  before we reach $Y_{1,\rank(1)-1}$. 
In the worst case we get a reward of $0$ in each such step.

\item In each step we are in state $s'$ in $(\CC_3\cap (Y_{1,\rank(1)}\setminus Y_{1,\rank(1)-1}))$ 
we get reward~1 with probability at least $1-\epsilon$ (by Equation~\ref{eqt: get 1}). 
\end{enumerate}
%We will now consider $\widetilde{r}(P_s^{\rank_2(s)-1},t)$. %As argued (in the case where $s'\in \CC_2$), 
For the play $P_s^{\rank_2(s)-1}=(P_{s_i}^{\rank_2(s_i)})_{i\geq 0}$,
the expected number of indices $i$ such that $s_i \in \CC_2$ is at most  
$\frac{\epsilon}{4}\cdot \left(\frac{\epsilon\cdot \delta_{\min}}{4}\right)^{-m}$ (by the first item above).
The remaining (in the worst case, at least $t-\frac{\epsilon}{4}\cdot \left(\frac{\epsilon\cdot \delta_{\min}}{4}\right)^{-m}$ in expectation) 
indices $i'$ are such that $s_{i'}\in \CC_3$, for which the expected reward 
is at least $1-\epsilon$ (by the second item above). 
Thus we have
\[
\widetilde{r}(P_s^{\rank(s)-1},t)\leq 
t\cdot \epsilon+\frac{\epsilon}{4}\cdot \left(\frac{\epsilon\cdot \delta_{\min}}{4}\right)^{-m}\leq t\cdot \epsilon+\frac{\epsilon}{4}\cdot \left(\frac{\epsilon\cdot \delta_{\min}}{4}\right)^{-2m}
= t\cdot \epsilon+ \frac{\kappa(1)}{2}\enspace ,
\] 
as desired.

\smallskip\noindent{\bf (Inductive case).}
We now consider the inductive case for $j\geq 1$, i.e., we now consider $s$ 
such that $\rank_2(s)=\rank(1)-j$. 
Consider the rank preserving split of the play $P_{s}^{\rank_2(s)-1}$ as $(P_{s_i}^{\rank_2(s_i)})_{i\geq 0}$
as explained before the base case.
We will consider $\widetilde{r}(P_{s'}^{\rank_2(s')},t)$, for all $s'$ with $\rank(s')=\rank(s)$. 
As in the base case, we will split the proof into the two cases: (1)~$s'\in \CC_2$; and (2)~$s'\in \CC_3$; 
(and recall $\CC_1$ is empty).
Before we consider the case analysis, we first present the use of the inductive hypothesis.

\smallskip\noindent{\em Use of inductive hypothesis.}
The inductive hypothesis will be used in the same way for both cases in the case analysis. 
Let $t \in \N$ be given. 
For all states $s''\in U_1$ such that $\rank_2(s'')>\rank_2(s)=\rank(1)-j$, 
we will use the inductive hypothesis to upper bound $\widetilde{r}(P_{s''}^{\rank(1)-j},t)$. 
%%(note the rank of the end state does not depend on the rank of $s''$).  
Consider the rank decreasing split of $P_{s''}^{\rank(1)-j}$ as 
$(P_{s_i'}^{\rank_2(s_i')-1})_{i\geq 0}$.
There are most $j$ such plays in the sequence, one for each rank strictly higher than $\rank(1)-j$. 
We only argue about the worst case, and in the worst case, 
$s_i'$ is such that $\rank_2(s_i')=\rank(1)-i$. 
%Let $t_0=\min(|P_{s_0'}^{\rank_2(s_0')-1}|,t)$ and for each $0< i<j$, let $t_i=\min(|P_{s_i'}^{\rank_2(s_i')-1}|,t-\sum_{i=0}^{j-1} t_i)$ (basically 
Let $t_i$ be the random variable indicating the number of steps among the first $t$ steps such that $P_{s''}^{\rank(1)-j}$ is exactly $P_{s_i'}^{\rank_2(s_i')-1}$. 
We see that $\widetilde{r}(P_{s''}^{\rank(1)-j},t)=\sum_{i=0}^{j-1} \widetilde{r}(P_{s_i'}^{\rank_2(s_i')-1},t_i)$. 
By the inductive hypothesis we have that $\widetilde{r}(P_{s_i'}^{\rank(s_i')-1},t')\leq t'\cdot\epsilon+ \frac{\kappa(i+1)}{2}$ for each $t'\geq 1$. 
Thus, we get that 
\[
\widetilde{r}(P_{s''}^{\rank(1)-j},t)=\sum_{i=0}^{j-1} \widetilde{r}(P_{s_i'}^{\rank(s_i')-1},t_i) 
\leq \sum_{i=0}^{j-1} \left( t_i\cdot \epsilon+ \frac{\kappa(i+1)}{2}\right) 
\leq t\cdot \epsilon + \kappa(j)
\]
The first inequality is the inductive hypothesis, 
and we now argue that $\sum_{i=0}^{j-1} \frac{\kappa(i+1)}{2} \leq \kappa(j)$.
We have 
\[
\sum_{i=0}^{j-1} \frac{\kappa(i+1)}{2}=\frac{\epsilon}{4}\cdot \sum_{i=0}^{j-1}\left(\frac{\epsilon\cdot \delta_{\min}}{4}\right)^{-(2m)^{i+1}}
\leq \frac{\epsilon}{2}\cdot \left(\frac{\epsilon\cdot\delta_{\min}}{4}\right)^{-(2m)^{j}} =  \kappa(j) \enspace ,
\]
because each term of the sum is over $4$ times as large as the preceding (because $(2m)^{i+1}\geq 1+ (2m)^{i}$, 
for $m\geq 2$ and $i\geq 0$ and the factor of 4) and thus, the last term is over $2$ times larger than the sum of all the other terms 
(we just use that it is larger).
We now consider the case analysis.

\begin{itemize}
\item \emph{(States in $\CC_2$).} In this case we consider $\widetilde{r}(P_{s'}^{\rank_2(s')},t)$, for $s'\in \CC_2$, such that $\rank(s')=\rank(s)$.
We know that $\sigma_1^{\epsilon,1}$, has patience  $\patc(j)$ for states $s''\in U_1$ such that $\rank_2(s'')=\rank_2(s)=\rank(1)-j$ (from Lemma~\ref{lemm:patience1}). 
In expectation the play $P_s^{\rank_2(s)-1}$ is therefore in a state $s''$ in $\CC_2$ such that $\rank(s'')=\rank(s)$ 
at most $\patc(j) \cdot (\delta_{\min})^{-1}$ times before reaching a state with lower rank (i.e., before the play ends). 
If the play does not end, whenever we have been in $\CC_2$, we reach some state $s''$ in $U_1$ (as safety to $U_1$ is guaranteed).  
Also, in the worst case we get a reward of~0 in the every step we are in a state of rank $\rank_2(s)$ in $\CC_2$.
There are two sub-cases. Either $\rank_2(s'')=\rank_2(s)$ or $\rank_2(s'')>\rank_2(s)$ (because if the rank is lower the walk ends).
In the first sub-case the play $P_{s'}^{\rank_2(s')}$ has length~1. 
In the other case, we have already given an upper bound on $\widetilde{r}(P_{s''}^{\rank(1)-j},t')$, for all $t'\geq 1$, using the inductive
hypothesis. We therefore have that 
\[
\widetilde{r}(P_{s'}^{\rank(s')},t) 
\leq 1+\widetilde{r}(P_{s''}^{\rank(1)-j},t-1)
\leq 1+(t-1)\cdot \epsilon+\kappa(j) 
= t\cdot \epsilon + (1-\epsilon) + \kappa(j) %%\\[2ex]
 \leq t\cdot \epsilon+2\cdot \kappa(j)
\]
where we have just explained the first inequality. 
The second inequality is our use of the inductive hypothesis as previously explained. 
The last inequality uses that $\kappa(j)=\frac{\epsilon}{2}\cdot \left(\frac{\epsilon\cdot\delta_{\min}}{4}\right)^{-(2m)^{j}}>8>1$ 
(since $4^{(2m)^j}\geq 16$ and hence  $\left(\frac{\epsilon\cdot\delta_{\min}}{4}\right)^{-(2m)^{j}} \geq \frac{16}{\epsilon}$ for $i,m\geq 1$) 
and $1-\epsilon<1$.

\item \emph{(States in $\CC_3$).} In this case we consider $\widetilde{r}(P_{s'}^{\rank_2(s')},t)$, for $s'\in \CC_3$, such that $\rank(s')=\rank(s)$. 
By construction, the strategy $\sigma_1^{\epsilon,1}$ plays the distribution $\dist^{\eta}_1$ over $\mov_1(s')$, 
for $\eta=\left(\frac{\epsilon\cdot\delta_{\min}}{4}\right)^{(2m)^{j}}=\frac{\epsilon}{2}\cdot \frac{1}{\kappa(j)}$.
For the play $P_{s'}^{\rank_2(s')}$, the next state $s_1$ after the start state $s'$ is in $U_1$ with probability~1; 
the reward is~1 with probability at least $1-\eta$, 
and as well $s'\in Y_{1,\rank(1)-i}$ with probability at least $1-\eta$
(since Equation~\ref{eqt: get 1} is ensured). 
With the remaining probability of at most $\eta$, the play $P_{s'}^{\rank_2(s')}$ goes to 
a state $s''$ in $U_1$.
As before the worst case (for the proof) is that with the remaining 
probability of at most $\eta$ the state $s''$ is such that $\rank_2(s'')>\rank_2(s)$, 
for which we have a upper bound by inductive hypothesis on $\widetilde{r}(P_{s''}^{\rank(1)-i},t')$, 
for all $t'\geq 1$. 
Thus we have that  
\begin{align*}
\widetilde{r}(P_{s'}^{\rank_2(s')},t)&\leq \eta+\eta\cdot \widetilde{r}(P_{s''}^{\rank(1)-i},t-1) 
%%\\[2ex] &
\leq \eta+\eta\cdot \left((t-1)\cdot \epsilon+\kappa(j)\right) \\[2ex]
&= \eta+(t-1)\cdot \eta \cdot \epsilon+\frac{\epsilon}{2} 
\leq \eta + (t-1)\cdot \epsilon + \frac{\epsilon}{2}   
%%\\[2ex] &
\leq t\cdot \epsilon \enspace.
\end{align*}
The first inequality is by the preceding explanation. 
The second inequality uses the inductive hypothesis as previously described.  
In the first equality, we use that by definition we have $\eta \cdot \kappa(j) = \frac{\epsilon}{2}$. 
In the third inequality  we use that $\eta \cdot \epsilon \leq \epsilon$ since $\eta\leq 1$ and $t\geq 1$;
and the final inequality uses that  since $\eta\leq \frac{\epsilon}{4}$ we have $\eta+\frac{\epsilon}{2}<\epsilon$ and $\eta\cdot \epsilon<\epsilon$, for $\epsilon<1$; 
for $i,m\geq 1$ which ensures $\eta\leq \frac{\epsilon}{4}$. 
\end{itemize}

We now combine the above case analysis to establish the inductive proof.
We will now consider $\widetilde{r}(P_{s}^{\rank_2(s)-1},t)$ and  our rank preserving split $(P_{s_i}^{\rank_2(s_i)})_{i\geq 0}$ of $P_s^{\rank_2(s)-1}$. For all $i\geq 0$, let $t_i$ be the random variable indicating the number of steps $P_s^{\rank_2(s)-1}$ is exactly $P_{s_i}^{\rank_2(s_i)}$ among the first $t$ steps of $P_{s_i}^{\rank_2(s_i)}$. We see that $\widetilde{r}(P_{s}^{\rank_2(s)-1},t)=\sum_{i=0}^{k} \widetilde{r}(P_{s_i}^{\rank_2(s_i)},t_i)$ (the random variable $k$ indicates the highest index such that $t_k\geq 1$, implying that $t_i\geq 1$ for $0\leq i\leq k$). 
Hence, we have that
\begin{align*}
\widetilde{r}(P_{s}^{\rank_2(s)-1},t)&=\sum_{i=0}^{k} \widetilde{r}(P_{s_i}^{\rank_2(s_i)},t_i) \\
&=\sum_{s_i\in \CC_2,\ i\leq k} \widetilde{r}(P_{s_i}^{\rank_2(s_i)},t_i)+\sum_{s_i\in \CC_3,\ i\leq k} \widetilde{r}(P_{s_i}^{\rank_2(s_i)},t_i) \\
& \leq \sum_{s_i\in \CC_2,\ i\leq k} (t_i\cdot \epsilon+2\cdot \kappa(j))+\sum_{s_i\in \CC_3,\ i\leq k} (t_i\cdot \epsilon) \\
& =\sum_{i=0}^{k} (t_i\cdot \epsilon) +\sum_{s_i\in \CC_2,\ i\leq k} (2\cdot \kappa(j)) \\
&\leq t\cdot \epsilon+ \patc(j)\cdot (\delta_{\min})^{-1} \cdot 2\cdot \kappa(j) \\
&=t\cdot \epsilon+ (\delta_{\min})^{-1}\cdot \epsilon\cdot \left(\frac{\epsilon\cdot\delta_{\min}}{4}\right)^{-(\frac{(2m)^{j+1}}{2}-1+(2m)^{j})} \\
&\leq t\cdot \epsilon+ (\delta_{\min})^{-1}\cdot \epsilon\cdot \left(\frac{\epsilon\cdot\delta_{\min}}{4}\right)^{-((2m)^{j+1}-1)} \\
&\leq t\cdot \epsilon+ \frac{\epsilon}{4}\cdot \left(\frac{\epsilon\cdot\delta_{\min}}{4}\right)^{-(2m)^{j+1}}  \\
&= t\cdot \epsilon+\frac{\kappa(j+1)}{2} \enspace .
\end{align*}
The first equality  follows from our preceding explanation.
The first inequality uses our bound on $\widetilde{r}(P_{s_i}^{\rank_2(s_i)},t_i)$ from the respective items above, depending on 
whether $s_i\in \CC_2$ or $s_i \in \CC_3$.
The second inequality uses that there are at most $\patc(j)\cdot (\delta_{\min})^{-1}$ indices $i$ such that $s_i\in \CC_2$, from the first item above, 
and that $t=\sum_{i=0}^{k} t_i$.
The third inequality uses that $(2m)^{j}\leq \frac{(2m)^{j+1}}{2}$ for $m\geq 2$ and $j\geq 1$.
The last follows from $\frac{\epsilon\cdot\delta_{\min}}{4}<\frac{\delta_{\min}}{4}$ and gives the expression we required to establish our inductive claim 
for $j$.

This completes the inductive proof and gives us the desired result.
\end{proof}

\smallskip\noindent{\bf The combinatorial property established in Lemma~\ref{lemm:U_1}.}
The proof of Lemma~\ref{lemm:U_1} shows that the strategy $\sigma_1^{\epsilon,1}$ against
all positional counter-strategies of the opponent ensures that in the resulting Markov
chain all recurrent classes that intersect with $U_1$ are contained in $U_1$, 
all states in $U_1$ have successors only in $U_1$;
(i.e., the recurrent classes in $U_1$ are reached with probability~1 from all 
states in $U_1$); and in every recurrent class in $U_1$ the mean-payoff value is
at least $1-\epsilon$.

\begin{lemma}\label{lemm:U_ell}{\em (Property~2).} Let $0<\epsilon<\frac{1}{2}$. 
The strategy $\sigma_1^{\epsilon,\ell}$ ensures that against all positional strategies $\sigma_2$ 
for all states  $s\in(U_{\ell}\setminus U_{\ell-1})$ we have that 
\begin{enumerate}
\item given the event $\Safe(U_{\ell}\setminus U_{\ell-1})$, the mean-payoff is at least $1-\epsilon$;
\item $\Pr_s^{\sigma_1^{\epsilon,\ell},\sigma_2}(\Safe(U_{\ell}\setminus U_{\ell-1})\cup \Reach(U_{\ell-1} \cup \ov{W}^*))=1$; and
\item $\Pr_s^{\sigma_1^{\epsilon,\ell},\sigma_2}(\Safe(U_{\ell}\setminus U_{\ell-1})\cup \Reach(U_{\ell-1}))\geq 1-\epsilon$.
\end{enumerate}
\end{lemma}
\begin{proof}
Given $\sigma_1^{\epsilon,\ell}$, let $\sigma_2$ be an arbitrary positional
counter-strategy for player~2.
We see that $\sigma_1^{\epsilon,\ell}$ is stationary and follows the distribution  $\dist^{\eta}$ over $\mov_1(s)$ for some $0<\eta<\epsilon$ in state $s\in (W^*\setminus U_{\ell-1})$. %%Hence $\sigma_2$ can be assumed to be positional. 
Let $\CC_i^{\sigma_1^{\epsilon,\ell},\sigma_2}=\CC_i$, i.e., 
given $\sigma_1^{\epsilon,\ell}$ and $\sigma_2$,
we have that $\CC_1, \CC_2, \CC_3$ are the set of states of $(U_\ell\setminus U_{\ell-1})$ 
that satisfy Equation~\ref{eqt: limit reach},  
Equation~\ref{eqt: almost reach},  
Equation~\ref{eqt: get 1},  respectively.
Let $\RS$ be the set of states in  $(U_\ell\setminus U_{\ell-1})$, 
from which $(\CC_1\cap (U_{\ell}\setminus U_{\ell-1}))$ is not reachable in 
the Markov chain (i.e., in the graph of the Markov chain given $\sigma_1^{\epsilon,\ell}$ and 
$\sigma_2$, the set $\RS$ is the set of states in $(U_\ell\setminus U_{\ell-1})$ 
from which no state in $(\CC_1\cap (U_{\ell}\setminus U_{\ell-1}))$ is reachable).
Equivalently, $\RS$ is the set from which $(U_{\ell-1}\cup \ov{W}^*)$ cannot be reached (the definitions are equivalent, because, from each state $s$ in $(U_{\ell}\setminus U_{\ell-1})=(S\setminus (U_{\ell-1}\cup \ov{W}^*))$, the set $(U_{\ell-1}\cup \ov{W}^*)$ can be reached in one-step iff $s\in \CC_1$).
%Let $\RS=\RS(\CC_1\cap (U_{\ell}\setminus U_{\ell-1}))$.
Consider now the segment of the play from state $s$ in $(U_\ell\setminus U_{\ell-1})$ till the play leaves $(U_\ell\setminus U_{\ell-1})$. 
\begin{enumerate}
\item 
First we consider the case when $s \in \RS$.
This corresponds to the proof of correctness for states in $U_1$ (note that in the correctness proof of $U_1$ the set $\CC_1$
was empty; and if $\CC_1$ is not reached, then the proof is identical to Lemma~\ref{lemm:U_1}, by construction of the strategy).
Hence we have that $\Safe(U_{\ell}\setminus U_{\ell-1})$ is ensured with probability~1 
(because $(U_{\ell}\setminus U_{\ell-1})$ can only be left from states in 
$\CC_1\cap (U_{\ell}\setminus U_{\ell-1})$) and 
$\lim_{t\rightarrow \infty} \frac{\sum_{i=0}^t \E^{\sigma_1^{\epsilon,\ell},\sigma_2}_s[\Theta_i]}{t}\geq 1-\epsilon$ 
(as in Lemma~\ref{lemm:U_1}).
This establishes all the required conditions of the lemma.

\item 
By Markov property~\ref{markov property RS}, we have that 
$\Reach(U_{\ell-1} \cup \overline{W}^* \cup \RS)$ happens with probability~1 (since $\RS$ is the set from which $(U_{\ell-1} \cup \overline{W}^*)$ cannot be reached).
Note that since $(S \setminus (U_{\ell-1} \cup \ov{W}^*)) = (U_{\ell}\setminus U_{\ell-1})$,
it follows that $\Reach(U_{\ell-1} \cup \overline{W}^* \cup \RS)$ with probability~1 
implies $\Reach(U_{\ell-1} \cup \overline{W}^*)  \cup \Safe(U_{\ell}\setminus U_{\ell-1})$ 
is also ensured with probability~1, since $(U_{\ell}\setminus U_{\ell-1})$ cannot be left once $\RS$ is reached.
This also shows that every recurrent class contained in $(U_{\ell}\setminus U_{\ell-1})$ must 
be contained in $\RS$ (and by the first item has mean-payoff value at least $1-\epsilon$).
This shows that given the event $\Safe(U_{\ell}\setminus U_{\ell-1})$, the mean-payoff is at least $1-\epsilon$.
From every state in $(U_{\ell}\setminus U_{\ell-1})$, in the Markov chain, we have that $\trans(s)(U_{\ell-1})\cdot \epsilon\geq \trans(s)(\ov{W}^*)$ 
(from states which are not in $\CC_1$, both probabilities are 0 and $C_1$ by Equation~\ref{eqt: limit reach}). 
Hence, Markov property~\ref{markov property rs3} implies that event $\Reach(U_{\ell-1} \cup \RS)$ happens with probability $1-\epsilon$  (since $\RS$ is the set from which $(U_{\ell-1} \cup \overline{W}^*)$ cannot be reached), 
i.e., we have $\Pr_s^{\sigma_1^{\epsilon,\ell},\sigma_2}(\Safe(U_{\ell}\setminus U_{\ell-1})\cup \Reach(U_{\ell-1}))\geq 1-\epsilon$.
\end{enumerate}
The desired result follows.
\end{proof}

\begin{figure}
\begin{center}

\begin{tikzpicture}[scale=0.95,node distance=3cm]
\tikzstyle{every state}=[fill=white,draw=black,text=black,font=\small ,minimum size=4mm]
\tikzstyle{txt}=[state,draw=white]

    \node[txt] (Ulm) at (2,0.25) {$(U_{\ell-(i+1)}\cup \RS)$};
    \node[txt] (s3) at (6,0.25) {$(U_{\ell-i}\setminus (U_{\ell-(i+1)}\cup \RS))$};
    \node[txt] (s4) at (10,0.25) {$(U_{\ell}\setminus (U_{\ell-i}\cup \RS))$};
    \node[txt] (s1) at (15,0.25) {$\ov{W}^*$};

\foreach \x in {0,4,8,12,13,17} {\draw (\x,0) -- (\x,-4);  }
\draw (0,0) -- (12,0); 
\draw (0,-4) -- (12,-4);

\draw (13,0) -- (17,0); 
\draw (13,-4) -- (17,-4);

    \node[state] (s2) at (2,-2) {$s_1$};
    \node[state] (s3) at (6,-2) {$s_2$};
    \node[state] (s4) at (10,-2) {$s_3$};
    \node[state] (s1) at (15,-2) {$s_4$};
    
\path (s1)  edge [>=stealth',shorten >=1pt,loop above,->] (s1);
\path (s2)  edge [>=stealth',shorten >=1pt,loop above,->] (s2);
\path (s3)  edge [>=stealth',shorten >=1pt,out=-40,in=180+40,->] node[above] {$\frac{\epsilon}{2}\cdot x$} (s1);
\path (s3)  edge [>=stealth',shorten >=1pt,->] node[above,near start] {$x$} (s2);
\path (s3)  edge [>=stealth',shorten >=1pt,bend left,->] node [above,fill=white,yshift=0.1cm]{$1-(1+\frac{\epsilon}{2})\cdot x$} (s4);
\path (s4)  edge [>=stealth',shorten >=1pt,bend left,->]  node [below,near start] {$1-\eta$} (s3);
\path (s4)  edge [>=stealth',shorten >=1pt,->] node  [above]{$\eta$} (s1);

\end{tikzpicture}
\end{center}
\caption{Pictorial illustration of the Markov chain $G_4^{x,\epsilon,\eta}$.}\label{fig:simple chain}
\end{figure}
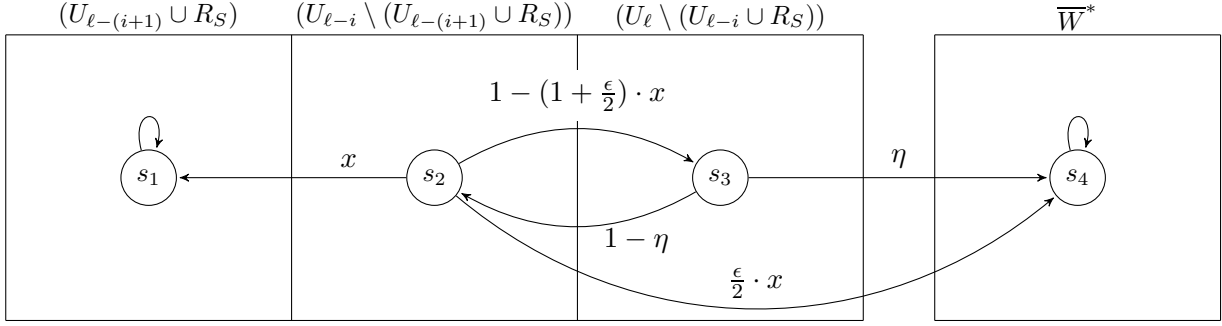

\begin{remark}
Lemma~\ref{lemm:U_ell} proves the desired result only for states in $(U_{\ell}\setminus U_{\ell-1})$
and can be considered as the base case of Lemma~\ref{lemm:g5} which proves
a similar result for states in $(U_{\ell-i}\setminus U_{\ell-(i+1)})$, for $1\leq i\leq \ell-2$. The case for states $(U_1\setminus U_0)=U_1$ is
handled by Lemma~\ref{lemm:U_1}.
Note that $\Safe(U_{\ell}\setminus U_{\ell-1}) \subseteq \coBuchi(U_{\ell}\setminus U_{\ell-1})$
and since mean-payoff objectives are independent of finite prefixes, it also follows 
from Lemma~\ref{lemm:U_ell} that given the event $\coBuchi(U_{\ell}\setminus U_{\ell-1})$, 
we have that the mean-payoff is at least $1-\epsilon$. 
\end{remark}

Before presenting the proof for Property~3 we first present a lemma that we will use to prove
the property.

\begin{lemma}\label{lemm:g4}
Given $0\leq  x\leq \frac{1}{2}$ and $0\leq \epsilon,\eta\leq 1$, consider the four-state Markov chain 
$G_4^{x,\epsilon,\eta}$ shown in Figure~\ref{fig:simple chain}. 
The probability to eventually reach $s_1$ from $s_2$ and $s_3$ is 
$\frac{x}{\eta+(1+\frac{\epsilon}{2})\cdot x\cdot (1-\eta)}$ and $\frac{x\cdot (1-\eta)}{\eta+(1+\frac{\epsilon}{2})\cdot x\cdot (1-\eta)}$, respectively.
\end{lemma}
\begin{proof}
Let $y_2$ and $y_3$ denote the probability to reach $s_1$ from $s_2$ and
$s_3$, respectively.
Then we have 
\[
y_2= x+ (1-(1+\frac{\epsilon}{2})\cdot x) \cdot y_3; \qquad \qquad y_3= (1-\eta)\cdot y_2 \enspace .
\]
Hence we have 
\[
y_2= x+ (1-(1+\frac{\epsilon}{2})\cdot x) \cdot (1-\eta)\cdot y_2 \enspace .
\]
Solving for $y_2$, and then inserting into $y_3= (1-\eta)\cdot y_2$, we obtain the desired result.
\end{proof}

\begin{lemma}\label{lemm:U_i}{\em (Property~3).}\label{lemm:g5} 
Let $0<\epsilon<\frac{1}{2}$ and $1\leq i \leq \ell-2$. 
The strategy $\sigma_1^{\epsilon,\ell-i}$ ensures that against all positional strategies $\sigma_2$ 
for all states  $s\in (U_{\ell}\setminus U_{\ell-(i+1)})$ we have that
\begin{enumerate}
\item given the event $\bigcup_{j\leq i}\coBuchi(U_{\ell-j}\setminus U_{\ell-(j+1)})$,
 the mean-payoff is at least $1-\epsilon$;
 \item $\Pr_s^{\sigma_1^{\epsilon,\ell-i},\sigma_2}(\bigcup_{j \leq i}\coBuchi(U_{\ell-j}\setminus U_{\ell-(j+1)})\cup \Reach(U_{\ell-(i+1)}\cup \ov{W}^*))=1$; and
 \item $\Pr_s^{\sigma_1^{\epsilon,\ell-i},\sigma_2}(\bigcup_{j \leq i}\coBuchi(U_{\ell-j}\setminus U_{\ell-(j+1)})\cup \Reach(U_{\ell-(i+1)}))\geq 1-\epsilon$.
\end{enumerate} 
\end{lemma}

\begin{proof}
Given $\sigma_1^{\epsilon,\ell-i}$, let $\sigma_2$ be an arbitrary positional
counter-strategy for player~2.
Let $\CC_i^{\sigma_1^{\epsilon,\ell-i},\sigma_2}=\CC_i$, i.e., 
given $\sigma_1^{\epsilon,\ell-i}$ and $\sigma_2$,
we have that $\CC_1, \CC_2, \CC_3$ are the set of states of $(U_\ell\setminus U_{\ell-(i+1)})$ 
that satisfy Equation~\ref{eqt: limit reach},  
Equation~\ref{eqt: almost reach},  
Equation~\ref{eqt: get 1},  respectively.
This proof is similar to the proof of Lemma~\ref{lemm:U_ell}.
The proof will be by induction in $i$, where $i=0$ is the base case. Hence, the base case is settled by Lemma~\ref{lemm:U_ell}.
We see that $\sigma_1^{\epsilon,\ell-i}$ is stationary and follows the distribution $\dist_1^{\eta}$ over $\mov_1(s)$ 
for some $\eta>0$ in state $s\in (W^*\setminus U_{\ell-(i+1)})$. 
%%Hence $\sigma_2$ can be assumed to be positional. Also let $W_i^{\sigma_1^{\epsilon,\ell-i},\sigma_2}=W_i$. 
We consider the Markov chain obtained by fixing the two strategies.
In the worst case, states in $\ov{W}^*$ are absorbing with reward~0;
and since the target is to reach $U_{\ell-(i+1)}$ we consider that  
the plays end if they leave $T=(W^*\setminus U_{\ell-(i+1)})$, i.e., we 
are interested in the segment of the play in $(W^* \setminus U_{\ell-(i+1)})$. 
The play can only end from a state in 
$\CC_1\cap T$ because $T=\bigcup_{j\leq i}(U_{\ell-j}\setminus U_{\ell-(j+1)})$ 
and if a state $s$ in $(U_{\ell-j}\setminus U_{\ell-(j+1)})$ satisfies either 
Equation~\ref{eqt: almost reach} (in $\CC_2$) or 
Equation~\ref{eqt: get 1} (in $\CC_3$), 
then the set $(U_{\ell-j}\setminus U_{\ell-(j+1)})$ is not left from $s$ in one-step. 
Now consider a play $P$ in the Markov chain. 
Let  $\RS$ be the subset of $T$, from which $\CC_1\cap T$ is not reachable
in the Markov chain. 
%Let $\RS=\RS(\CC_1\cap T)$.
There are two cases
\begin{enumerate}
\item {\bf($P$ starts in $s\in \RS$)}. 
Let $(\ell-i',j')=\rank(s)$. 
Note that $i'\leq i$, by definition of $\RS$. 
Precisely, like in the proof of Lemma~\ref{lemm:U_ell}, 
we have that $\Safe(U_{\ell-i'}\setminus U_{\ell-(i'+1)})$ is ensured with probability~1, because the set $(U_{\ell-i'}\setminus U_{\ell-(i'+1)})$ cannot be left from states in $\CC_2$ or $\CC_3$.
Hence, if $i'< i$, then we are done, by induction, since $\sigma_1^{\epsilon,\ell-i}$ follows $\sigma_1^{\eta,\ell-i+1}$ in such states, by construction of $\sigma_1^{\epsilon,\ell-i}$, for  $\eta=\left(\frac{\epsilon\cdot\delta_{\min}}{4}\right)^{(2m)^{\rank(\ell-i)}}$ and we have that $\eta<\epsilon$, for $m\geq 2$ and $\rank(\ell-i)\geq 1$.
If $i'=i$, then, precisely like in the proof of Lemma~\ref{lemm:U_ell}, the set $(U_{\ell-i}\setminus U_{\ell-(i+1)})$ cannot be left in $\CC_2$ or $\CC_3$ and hence, using an argument like Lemma~\ref{lemm:U_1}, we have 
that $\lim_{t\rightarrow \infty} \frac{\sum_{i=0}^t \E^{\sigma_1^{\epsilon,\ell-i},\sigma_2}_s[\Theta_i]}{t}\geq 1-\epsilon$, because of the similarities between the construction of the strategy $\sigma_1^{\epsilon,i}$ and $\sigma_1^{\epsilon,1}$ for states in $(U_{\ell-i}\setminus U_{\ell-(i+1)})$ and states in $U_{1}$, respectively.
Observe that this case is the same as the corresponding case in Lemma~\ref{lemm:U_ell} and ensures
all the required items of the lemma.

\item {\bf($P$ starts outside $\RS$: Item (1) of the lemma statement)}.
First observe that we can only ensure $\Safe(U_{\ell-j}\setminus U_{\ell-(j+1)})$, 
for some $j\leq i$, from states in $\RS$, since from all other states 
$\CC_1$ is reachable and for every $j$, 
states in $(\CC_1\cap (U_{\ell-j}\setminus U_{\ell-(j+1)}))$, can reach $U_{\ell-(j+1)}$ in one-step with 
positive probability, by Equation~\ref{eqt: limit reach}. 
Hence, if $\bigcup_{j\leq i}\coBuchi(U_{\ell-j}\setminus U_{\ell-(j+1)})$ is 
ensured, then given the event 
$\bigcup_{j\leq i}\coBuchi(U_{\ell-j}\setminus U_{\ell-(j+1)})$ a recurrent
class that is reached must be contained in $\RS$.
Hence given the event $\bigcup_{j\leq i}\coBuchi(U_{\ell-j}\setminus U_{\ell-(j+1)})$, 
the set $\RS$ is reached in a finite number of steps with probability~1.
Since mean-payoffs are independent of finite-prefixes, 
the finite prefix to reach $\RS$ does not change the mean-payoff.
Moreover, since if we start in $\RS$ the mean-payoff is at least $1-\epsilon$, 
it follows that given the event $\bigcup_{j\leq i}\coBuchi(U_{\ell-j}\setminus U_{\ell-(j+1)})$ 
we have that the mean-payoff is at least $1-\epsilon$.

\item {\bf($P$ starts outside $\RS$: Item (2) of the lemma statement)}. 
For $0 \leq i' \leq i$, let $\cale_{i'}$ denote the following 
event,
\[
\cale_{i'}=\bigcup_{j \leq i'}\coBuchi(U_{\ell-j}\setminus U_{\ell-(j+1)})\cup \Reach(U_{\ell-(i'+1)}\cup \ov{W}^*).
\] 
Let $\SP(s,\ell-i')=\Pr\nolimits_s^{\sigma_1^{\epsilon,\ell-i'},\sigma_2}(\cale_{i'})$,
for all $0\leq i'\leq i$, denote the \emph{success probability} of the 
event $\cale_{i'}$.
%%given $s$, $\sigma_1^{\epsilon,\ell-i'}$, and $\sigma_2$ 
We need to argue that $\SP(s,\ell-i)=1$, for all states in $(U_{\ell}\setminus U_{\ell-(i+1)})$.
By induction we have that $\SP(s,\ell-(i-1))=1$, 
from states in $(U_{\ell}\setminus U_{\ell-i})$. 
Since $\sigma_1^{\epsilon,\ell-i}$ has the same support as $\sigma_1^{\epsilon,\ell-(i-1)}$ for all states in 
$(U_{\ell}\setminus U_{\ell-i})$, it follows that for each state $s$ in $(U_{\ell}\setminus U_{\ell-i})$ we have  
$\SP(s,\ell-i)=1$. 
If the event $\bigcup_{j \leq \ell-(i+1)}\coBuchi(U_{\ell-j}\setminus U_{\ell-(j+1)})\cup \Reach(\ov{W}^*)$ happens, then we are done. 
Thus, in the worst case we have that $\Pr_s^{\sigma_1^{\epsilon,\ell-i},\sigma_2}(\Reach(U_{\ell-i}))=1$ from state $s$ in $(U_{\ell}\setminus U_{\ell-i})$ 
(clearly, from such states $U_{\ell-i}$ is reachable in the Markov chain 
since they are reached with probability~1). 
We only need to argue about the worst case.
%%We therefore only need to consider $((U_{\ell-i}\setminus (U_{\ell-(i+1)}\cup \RS))$. 
Let $\RS'$ be the subset of $(U_{\ell-i}\setminus U_{\ell-(i+1)})$, 
from which $(\CC_1\cap (U_{\ell-i}\setminus U_{\ell-(i+1)}))$ cannot be reached in the Markov chain. 
Hence, for each state $s$  in $(U_{\ell-i}\setminus U_{\ell-(i+1)})$, the state $s$ must either be in $\RS'$ (in which case $\RS'$ is reachable) or the set $(\CC_1\cap (U_{\ell-i}\setminus U_{\ell-(i+1)}))$ must be reachable from $s$. From the set $(\CC_1\cap (U_{\ell-i}\setminus U_{\ell-(i+1)}))$, the set $U_{\ell-(i+1)}$ is reached in one-step with positive probability.
We therefore get that from any state in $T=((U_{\ell}\setminus U_{\ell-i})\cup (U_{\ell-i}\setminus U_{\ell-(i+1)}))$, the set $(U_{\ell-(i+1)}\cup \RS')$ is reachable, by transitivity of reachabillity.
Hence, by Markov property~\ref{markov property reachability} we have that $\Pr_s^{\sigma_1^{\epsilon,\ell-i},\sigma_2}\Reach((S\setminus T)\cup U_{\ell-(i+1)}\cup \RS')=1$, from any state $s\in T$.
Note that from states in $\RS'$ no state in $\CC_1 \cap (U_{\ell-i}\setminus U_{\ell-(i+1)})$ is reachable, 
and the set $(U_{\ell-i}\setminus U_{\ell-(i+1)})$ can be left only from states in $\CC_1 \cap (U_{\ell-i}\setminus U_{\ell-(i+1)})$.
Hence reachability to $\RS'$ ensures $\coBuchi((U_{\ell-i}\setminus U_{\ell-(i+1)}))$.
Thus we have that 
\begin{align*}
\Reach((S\setminus T)\cup U_{\ell-(i+1)}\cup \RS') &= \Reach(U_{\ell-(i+1)}\cup \ov{W}^*\cup U_{\ell-(i+1)}\cup \RS') \\
& = \Reach(U_{\ell-(i+1)}\cup \ov{W}^*\cup \RS') \\ 
& \subseteq \Reach(U_{\ell-(i+1)}\cup \ov{W}^*) \cup \coBuchi(U_{\ell-i}\setminus U_{\ell-(i+1)}) \subseteq \cale_i \enspace .
\end{align*}
The first equality uses that $(S\setminus T)=(U_{\ell-(i+1)}\cup \ov{W}^*)$. 
The first inclusion uses that $\Reach(\RS')$ ensures $\coBuchi(U_{\ell-i}\setminus U_{\ell-(i+1)})$.
Hence, from each state $s\in T$ we have that $\SP(s,\ell-i)=1$ as desired.

\item {\bf($P$ starts outside $(\RS \cap T)$: Item (3) of the lemma statement.)}.
We will now show that the probability of the event $(\bigcup_{j \leq i}\coBuchi(U_{\ell-j}\setminus U_{\ell-(j+1)})\cup \Reach(U_{\ell-i}))$ is at least $1-\epsilon$.
We will do so by modeling the worst case using the Markov chain $G_4^{x,\epsilon,\eta}$ 
of Lemma~\ref{lemm:g4}.  There is an illustration of the Markov chain $G_4^{x,\epsilon,\eta}$ in Figure~\ref{fig:simple chain}.
We have one state representing each of the following sets 
\begin{enumerate}
\renewcommand{\labelenumii}{(\arabic{enumii})}
\item $(U_{\ell-(i+1)}\cup \RS)$
\item $(U_{\ell-i}\setminus (U_{\ell-(i+1)}\cup \RS))$
\item $(U_{\ell}\setminus (U_{\ell-i}\cup \RS))$
\item $\ov{W}^*$
\end{enumerate}
We will refer to the states as $s_1$, $s_2$, $s_3$ and $s_4$, respectively. 
We will now argue about the transition probabilities, and first consider
the absorbing states.

\smallskip\noindent{\bf The state $s_1$.} 
We are interested in the probability that $(U_{\ell-(i+1)}\cup \RS)$ is eventually reached. This probability does not depend on what happens after $(U_{\ell-(i+1)}\cup \RS)$ is reached.
Hence, we consider $s_1$ as absorbing, like in $G_4^{x,\epsilon,\eta}$.

\smallskip\noindent{\bf The state $s_4$.} In the worst case $\ov{W}^*$ cannot be left, once reached. Thus $s_4$ is an absorbing state, like in $G_4^{x,\epsilon,\eta}$.

\smallskip\noindent{\bf The state $s_2$.}  For each state $s\in (U_{\ell-i}\setminus (U_{\ell-(i+1)}\cup \RS))\subseteq (U_{\ell-i}\setminus U_{\ell-(i+1)})$, we must eventually reach a state in either $(\CC_1\cap (U_{\ell-i}\setminus U_{\ell-(i+1)}))=((\CC_1\cap T)\cap (U_{\ell-i}\setminus U_{\ell-(i+1)}))$ or $(\RS\cap (U_{\ell-i}\setminus U_{\ell-(i+1)}))$, with probability~1, by Markov property~\ref{markov property rs2new} (recall that we cannot reach states outside $(U_{\ell-i}\setminus U_{\ell-(i+1)})$, except from states in $(\CC_1\cap (U_{\ell-i}\setminus U_{\ell-(i+1)}))$ by Equation~\ref{eqt: limit reach}, Equation~\ref{eqt: almost reach} and Equation~\ref{eqt: get 1}. Also, $(\RS\cap (U_{\ell-i}\setminus U_{\ell-(i+1)}))$ is the subset of $(U_{\ell-i}\setminus U_{\ell-(i+1)})$ from which $(\CC_1\cap T)$ cannot be reached).
If we reach $\RS$,  an argument similar to the first item in the proof of this lemma shows that we satisfy the desired statement. 
Thus, in the worst case we always reach $(\CC_1\cap (U_{\ell-i}\setminus U_{\ell-(i+1)}))$. 
For each state $s$ in $(\CC_1\cap (U_{\ell-i}\setminus U_{\ell-(i+1)}))$, let 
$x_s = \trans(s,\sigma_1^{\epsilon,\ell-i},\sigma_2)(U_{\ell-(i+1)})$ be the one-step
transition probability to $U_{\ell-(i+1)}$. 
By Equation~\ref{eqt: limit reach}, and the construction of the strategy, 
we have that $\frac{\epsilon}{2}\cdot x_s>\trans(s,\sigma_1^{\epsilon,\ell-i},\sigma_2)(\ov{W}^*)$. 
Clearly, in the worst case we have that $\frac{\epsilon}{2}\cdot x_s=\trans(s,\sigma_1^{\epsilon,\ell-i},\sigma_2)(\ov{W}^*)$
(recall that $\ov{W}^*$ is absorbing). 
Also, the fact $x_s>\trans(s,\sigma_1^{\epsilon,\ell-i},\sigma_2)(\ov{W}^*)$ implies that $x_s>0$ and 
therefore we have that $x_s\geq \frac{\delta_{\min}}{\patc(\ell-i)}$, where  $\patc(\ell-i)=\left(\frac{\epsilon\cdot\delta_{\min}}{4}\right)^{-(\frac{(2m)^{\rank(\ell-i)}}{2}-1)}$, is an upper bound on the patience of the distribution $\sigma_1^{\epsilon,\ell-i}(s)$, by Lemma~\ref{lemm:patience1}. 
Thus with probability $x_s$ we go to $U_{\ell-(i+1)}$, with 
probability $\frac{\epsilon}{2}\cdot x_s$ we go to $\ov{W}^*$,
and with the remaining probability of $(1-(1+\frac{\epsilon}{2})\cdot x_s)$ 
we go to a state in $T$, which in the worst case is a state in $(U_{\ell}\setminus (U_{\ell-i} \cup \RS))$. 
This is so, because, in the worst case, to reach  $(U_{\ell-(i+1)} \cup \RS)$ from $(U_{\ell}\setminus (U_{\ell-i}\cup\RS))$ we must go through 
a state in $(U_{\ell-i}\setminus (U_{\ell-(i+1)}\cup \RS))$, and hence the probability to reach $U_{\ell-(i+1)}$ 
is minimized when $x_s$ is as small as possible, for all $s$. That is, $x_s=\frac{\delta_{\min}}{\patc(\ell-i)}$, 
for all $s\in (\CC_1\cap (U_{\ell-i}\setminus U_{\ell}))$. 
Let $x=\frac{\delta_{\min}}{\patc(\ell-i)}$. Thus, the transition probabilities are as follows:
(i)~from $s_2$ to $s_4$ is $\frac{\epsilon}{2}\cdot x$;
(ii)~from $s_2$ to $s_1$ is  $x$; and 
(iii)~from $s_2$ to $s_3$ is $1-(1+\frac{\epsilon}{2})\cdot x$. 
Thus, $s_2$ is like in $G_4^{x,\epsilon,\eta}$.

\smallskip\noindent{\bf The state $s_3$.} For each state $s\in (U_{\ell}\setminus (U_{\ell-i}\cup \RS))\subseteq (U_{\ell}\setminus U_{\ell-i})$, by induction and since $\sigma_1^{\epsilon,\ell-i}$ follows $\sigma_1^{\eta,\ell-i}$, 
we satisfy that $\Pr_s^{\sigma_1^{\epsilon,\ell-i},\sigma_2}(\bigcup_{j \leq i-1}\coBuchi(U_{\ell-j}\setminus U_{\ell-(j+1)})\cup \Reach(U_{\ell-i}))\geq 1-\eta$, 
where $\eta$ is $\left(\frac{\epsilon\cdot\delta_{\min}}{4}\right)^{(2m)^{\rank(\ell-i)}}$. By item (2) of the lemma statement, we enter $\ov{W}^*$ with the remaining probability (which is absorbing).
Hence, the worst case must be where 
$\Pr_s^{\sigma_1^{\epsilon,\ell-i},\sigma_2}(\bigcup_{j \leq i-1}\coBuchi(U_{\ell-j}\setminus U_{\ell-(j+1)})\cup \Reach(U_{\ell-i}))= 1-\eta$ (and thus $\Pr_s^{\sigma_1^{\epsilon,\ell-i},\sigma_2}(\Reach(\ov{W}^*))=\eta$).
As previously argued, in the first item and second item of this lemma, 
%%we have that $\RS=\bigcup_{j \leq i}\Safe(U_{\ell-j}\setminus U_{\ell-(j+1)}))$. Thus, 
the event $\bigcup_{j \leq i-1}\coBuchi(U_{\ell-j}\setminus U_{\ell-(j+1)})$ ensures reachability to $\RS$ (i.e., ensures $\Reach(\RS)$).
In the worst case for the proof the probability to reach $(\RS\cup U_{\ell-i-1})$ is minimized, and thus in the worst case 
we have $\Pr_s^{\sigma_1^{\epsilon,\ell-i},\sigma_2}(\Reach((U_{\ell-i}\setminus (U_{\ell-(i+1)}\cup \RS)))=1-\eta$ and 
$\Pr_s^{\sigma_1^{\epsilon,\ell-i},\sigma_2}(\Reach(\ov{W}^*))=\eta$. 
%Again, since we are trying to minimize the proability of eventually reaching $(\RS\cup U_{\ell-i-1})$ and $(\RS\cup U_{\ell-i-1})\subseteq U_{\ell-i}$, in the worst case we must have that we reach $(U_{\ell-i}\setminus (U_{\ell-(i+1)}\cup \RS))$ with probability $1-\eta$.
Thus, from $s_3$ the transition probability to $s_2$ and $s_4$ are $1-\eta$ and $\eta$, respectively. 
Thus, $s_3$ is like in $G_4^{x,\epsilon,\eta}$.

\smallskip\noindent{\bf The probability to eventually reach $s_1$ from $s_2$ or $s_3$.}  
We have that $x\leq \frac{1}{2}$ (since $\patc(\ell-i)\leq \frac{1}{2}$, for $m\geq 2$ and $\rank(\ell-i)\geq 1$). Also, $0<\eta,\epsilon<1$ (in the case of $\eta$, because $m\geq 2$ and $\rank(\ell-i)\geq 1$). Hence we can apply Lemma~\ref{lemm:g4} and get that the probability to eventually reach $s_1$ from $s_2$ and $s_3$ is 
$\frac{x}{\eta+(1+\frac{\epsilon}{2})\cdot x\cdot (1-\eta)}$ and $\frac{x\cdot (1-\eta)}{\eta+(1+\frac{\epsilon}{2})\cdot x\cdot (1-\eta)}$, respectively. Cleary, the probability from $s_3$ is the smallest. We will show that it is greater than $1-\epsilon$. We have that \begin{align*}
\frac{x\cdot (1-\eta)}{\eta+(1+\frac{\epsilon}{2})\cdot x\cdot (1-\eta)} &=
\frac{1}{\frac{\eta}{x\cdot (1-\eta)}+1+\frac{\epsilon}{2}} \geq \frac{1}{1+\epsilon} \geq 1-\epsilon \enspace .
\end{align*}
We will argue about the first inequality last. The second inequality follows from $1>1-\epsilon^2=(1+\epsilon)\cdot (1-\epsilon)\Rightarrow \frac{1}{1+\epsilon}>1-\epsilon$.
To show the first inequality we will argue that $\frac{\eta}{x\cdot (1-\eta)}\leq \frac{\epsilon}{2}$ or, equivalently, that $\frac{2\cdot \eta}{x\cdot (1-\eta)\cdot \epsilon}\leq 1$, since $\epsilon>0$. We have that
\begin{align*}
\frac{2\cdot \eta}{x\cdot (1-\eta)\cdot \epsilon}&<\frac{4\cdot \eta}{x\cdot \epsilon} =\frac{4\cdot \eta\cdot \patc(\ell-i)}{\delta_{\min}\cdot \epsilon} =\eta \cdot \left(\frac{\epsilon\cdot\delta_{\min}}{4}\right)^{-\frac{(2m)^{\rank(\ell-i)}}{2}} = \eta^{\frac{1}{2}} < 1 \enspace .
\end{align*}
The inequalities comes from $\eta<\frac{1}{2}$ (which is the case because $m\geq 2$ and $\rank(\ell-i)\geq 1$). The first equality is because $x=\frac{\delta_{\min}}{\patc(\ell-i)}$, by definition. The second equality is because
$\patc(\ell-i)=\left(\frac{\epsilon\cdot\delta_{\min}}{4}\right)^{-(\frac{(2m)^{\rank(\ell-i)}}{2}-1)}$, by definition. The third equality uses that $\eta=\left(\frac{\epsilon\cdot\delta_{\min}}{4}\right)^{(2m)^{\rank(\ell-i)}}$, by definition.

\smallskip\noindent{\bf Ensuring item (3) of the lemma statement.}  We see that the probability to reach $(U_{\ell-(i+1)}\cup \RS)$ from $T$ is more than $1-\epsilon$ (by recalling the definition of $s_1$, $s_2$ and $s_3$) 
and thus item (3) of the lemma statement is ensured, because from states in $\RS$ the event $\bigcup_{j \leq i}\Safe(U_{\ell-j}\setminus U_{\ell-(j+1)})$
is ensured (as argued in the beginning of the lemma) and hence reaching $\RS$ ensures $\bigcup_{j \leq i}\coBuchi(U_{\ell-j}\setminus U_{\ell-(j+1)})$.

\end{enumerate}
The desired result follows.
\end{proof}

\begin{lemma}\label{lemm:final_val1}
Let $0<\epsilon<\frac{1}{2}$.
The stationary strategy $\sigma_1^{\epsilon}$ ensures that for all states $s \in W^*$ and 
all strategies $\sigma_2$ we have 
$\Exp_s^{\sigma_1^\epsilon,\sigma_2}[\LimSupAvg] \geq 
\Exp_s^{\sigma_1^\epsilon,\sigma_2}[\LimInfAvg] \geq 1-\epsilon$.
\end{lemma}
\begin{proof}
By construction $\sigma_1^{\epsilon}$ plays as $\sigma_1^{\beta,1}$ in $U_1$ and 
$\sigma_1^{\beta,2}$, for $\beta=\frac{\epsilon}{2}$, in the remaining states.
Therefore $\sigma_1^{\epsilon}$ ensures that the mean-payoff of any play  that starts in $U_1$ is at least $1-\beta$, by Lemma~\ref{lemm:U_1}. 
Since $\sigma_1^{\epsilon}$ is stationary, once $\sigma_1^\epsilon$ is fixed we obtain an MDP
for player~2, and in MDPs positional strategies always suffice to minimize mean-payoff objectives~\cite{FV97}.
Hence, Lemma~\ref{lemm:U_i} shows that if the play starts in $s\in (U_{\ell}\setminus U_{1})$, then with probability 
$1-\beta$ the play either stays in $(U_{j}\setminus U_{j-1})$ for some $j\geq 2$ and ensures mean-payoff of at least $1-\beta$ or reaches $U_{1}$, 
from which we will get mean-payoff $1-\beta$.
By simple multiplication (using that rewards are at least 0) we therefore see that we get mean-payoff at least 
\[
\left(1-\beta\right)^2 = 1+\beta^2-2\beta\geq 1-\epsilon.
\]
The desired result follows.
\end{proof}

Lemma~\ref{lemm:final_val1} implies the following inclusion.
\begin{lemma}\label{lemm:incl1}
We have $W^* \subseteq \val_1(\LimInfAvg(\cost),\bigstra_1^S) \subseteq \val_1(\LimSupAvg(\cost),\bigstra_1^S)$.
\end{lemma}

\subsubsection{Second inclusion: $\ov{W}^* \subseteq S \setminus \val_1(\LimInfAvg,\bigstra_1^F)$\label{sec:incl2}}
We will now show that for all states $s\in \ov{W}^*$ that there exists a constant $c>0$
such that no finite-memory strategy $\sigma_1$ for player~1 can ensure value more than  $1-\frac{c^n}{n}$.
Again the statement is trivially true if $\ov{W}^*$ is empty, and hence we assume that 
this is not the case.

\smallskip\noindent{\bf Computation of $\overline{W}^*$.} %%and counter-strategy $\sigma_2$.}
We first analyze the computation of $\overline{W}^*$. To analyze the computation of $\ov{W}^*$ 
we consider the iterative computation $W^*$
\begin{itemize}
\item Let $W_0$ be $S$ and $W_i$ be $\mu U. \nu X. \mu Y. \nu Z. \EXP(W_{i-1},U,X,Y,Z)$.
\item Let $X_{i,0}$ be $S$ and $X_{i,j}$ be  $\nu X. \mu Y. \nu Z. \EXP(W_{i-1},W_i,X_{i,j-1},Y,Z)$.
\item Also let $Z_{i,j,0}$ be $S$ and $Z_{i,j,k}$ be  $\EXP(W_{i-1},W_i,X_{i,j-1},X_{i,j},Z_{i,j,k-1})$.
\end{itemize}   
Let $\ov{\ell}\geq 0$ be the smallest number such that $W_{\ov{\ell}}=W_{\ov{\ell}+1}=W^*$. Let $\ov{\rank}(i)$, be the smallest number $j$ such that $X_{i,j}=X_{i,j+1}$. Also, let $\ov{\rank}(i,j)$, be the smallest number $k$ such that $Z_{i,j,k}=Z_{i,j,k+1}$.
We have that for any state $s$  in $\overline{W}^*$, there must be some smallest number $i$ such that $s$ is not in $W_i$ 
(since $W_0$ is $S$, we have that $i>0$).
Also, there must be some smallest $j$ such that $s$ is not in $X_{i,j}$ and similar for $k$ and $Z_{i,j,k}$. 
We define the rank of a state $s\in \ov{W}^*$ as $\ov{\rank}(s)=(i,j,k)$, where $i$ (resp. $j$, and $k$)
is the smallest number such that $s$ not in $W_i$ (resp. $X_{i,j}$ and $Z_{i,j,k}$).  
By definition of $\overline{W}^*$,  there exists a constant $c>0$, such that for a state $s$, 
with $\ov{\rank}(s)=(i,j,k)$, for all distributions $\dist_1$ over $\mov_1(s)$ 
there must exist an counter-action $a_2^{s,\dist_1}\in \mov_2(s)$ for player~2  
such that all the following conditions hold (i.e., the negation of the conditions of 
$\EXP$ hold):
\begin{align*}
&(c \cdot \trans(s,\dist_1,a_2^s)(W_i) \leq  \trans(s,\dist_1,a_2^s)(\overline{W}_{i-1})) \\[2ex]
\wedge&(\trans(s,\dist_1,a_2^s)(X_{i,j-1})<1 \vee \trans(s,\dist_1,a_2^s)(X_{i,j}) = 0) \\[2ex]
\wedge&(\trans(s,\dist_1,a_2^s)(X_{i,j-1})<1 \vee
\ExpRew(s,\dist_1,a_2^s) < 1-c  \vee
 \trans(s,\dist_1,a_2^s)(Z_{i,j,k-1}) < 1- c) \enspace .
\end{align*}
If the above conditions hold, then one of the following three 
conditions hold as well.
We first explain the following cases: (i)~if $\trans(s,\dist_1,a_2^{s,\dist_1})(W_{i})>0$, 
then $c \cdot \trans(s,\dist_1,a_2^{s,\dist_1})(W_{i}) \leq  \trans(s,\dist_1,a_2^{s,\dist_1})(\overline{W}_{i-1})$
must hold to ensure the first condition above (this corresponds to Case~(3) below);
(ii)~if $\trans(s,\dist_1,a_2^{s,\dist_1})(W_{i})=0$, then the first condition above is satisfied;
then we have two sub-cases: (a)~if $\trans(s,\dist_1,a_2^{s,\dist_1})(X_{i,j-1})<1$, then both 
the second and third condition is satisfied (this corresponds to Case~(2) below);
(b)~otherwise we must have $\trans(s,\dist_1,a_2^{s,\dist_1})(X_{i,j})=0$ to satisfy 
the second condition above and 
$(\ExpRew(s,\dist_1,a_2^{s,\dist_1}) < 1-c \quad \vee \quad  \trans(s,\dist_1,a_2^{s,\dist_1})(Z_{1,j,i-1})<1-c)$ 
to satisfy the third condition above (this corresponds to Case~(1) below).
Thus we have that either
\begin{itemize}
\item {\bf Case~$(1)$.} There is a  $a_2^{s,\dist_1}$ such that
\begin{align*}
&\quad\trans(s,\dist_1,a_2^{s,\dist_1})(W_{i})=0 \\
\wedge&\quad\trans(s,\dist_1,a_2^{s,\dist_1})(X_{i,j})=0\\
\wedge&\quad\big(\ExpRew(s,\dist_1,a_2^{s,\dist_1}) < 1-c \quad \vee \quad \trans(s,\dist_1,a_2^{s,\dist_1})(Z_{1,j,i-1})<1-c\big)
\end{align*}
or;
\item {\bf Case~$(2)$.} There is a $a_2^{s,\dist_1}$ such that \begin{align*}
&\big(\trans(s,\dist_1,a_2^{s,\dist_1})(W_{i})=0 \big) \quad \wedge \quad \big(\trans(s,\dist_1,a_2^{s,\dist_1})(X_{i,j-1})<1 \big)
\end{align*}
or;
\item  {\bf Case~$(3)$.} There is a  $a_2^{s,\dist_1}$ such that 
\begin{align*}
& \big(c \cdot \trans(s,\dist_1,a_2^{s,\dist_1})(W_{i}) \leq  \trans(s,\dist_1,a_2^{s,\dist_1})(\overline{W}_{i-1}) \big) \quad \wedge \quad \big(\trans(s,\dist_1,a_2^{s,\dist_1})(W_{i})>0\big) \enspace .
\end{align*}
\end{itemize}
We will use the above three cases explicitly in our proof.

\smallskip\noindent{\bf The counter-strategy $\sigma_2$ given ${\sigma_1}$.}
Fix an arbitrary finite-memory strategy $\sigma_1$ for player~1.
Let the finite set of memories used by $\sigma_1$ be $\Mem$. 
A counter-strategy $\sigma_2$ given ${\sigma_1}$ is defined as follows:
given the current state $s$ of the game, and current memory state
$m \in \Mem$, let $\dist_1$ be the distribution played by 
$\sigma_1$.
The strategy $\sigma_2$ for player~2 plays an action $a_2^{s,\dist_1}$ 
(if there are more than one option for $a_2^{s,\dist_1}$, pick one arbitrarily) 
with probability one. If $\sigma_1$ uses memory set $\Mem$, then $\sigma_2$
also uses the memory set $\Mem$ and has the same memory update function.

\smallskip\noindent{\bf Upper bound on value ensured by $\sigma_1$.}
We will show that given $\sigma_1$ and the counter-strategy $\sigma_2$ 
the mean-payoff value is at most $1-\frac{c^n}{n}$ for all 
starting states in $\ov{W}^*$.
Also note that the upper bound on the value is independent of the size of
the memory, 
and this shows that in the complement of $W^*$ the values achievable 
by finite-memory strategies is strictly bounded below~1.

\smallskip\noindent{\bf The game $G\times \Mem$.}
Consider the game $G$ and a product with any deterministic automaton $A$ with state
space $Q$. Every state in $\ov{W}^* \times Q$ in the synchronous product 
game belongs to the set $\ov{W}^*$ computed in the product game and the ranks
also coincide (by the properties of $\mu$-calculus formulae).
Consider the synchronous product game $G\times \Mem$ of $G$ and the memories of $\sigma_1$ and $\sigma_2$,  
where states corresponds to pairs in $(S,\Mem)$ and where $\trans((t,m),a,b)((t',m'))=\trans(s,a,b)(t)$ where 
$\sigma_1^u(t,a,b,m)=m'$ and hence also $\sigma_2^u(t,a,b,m)=m'$. In this game the strategy corresponding to $\sigma_1$ can be interpreted as a 
stationary strategy $\sigma_1'$. Also the strategy corresponding to $\sigma_2$ can be interpreted as a 
positional strategy $\sigma_2'$ in $G \times \Mem$.
Hence given the strategies $\sigma_1$ and $\sigma_2$ we can obtain a Markov chain on 
$G \times \Mem$, considering the stationary strategies $\sigma_1'$ and $\sigma_2'$ on the
product game.
Also for all states $t\in \overline{W}^*$ in $G$, all the corresponding states $(t,m)$ in 
$G\times \Mem$ belong to $\ov{W}^*$ computed in the product game and has the same 
rank as $t$ in $G$.
%reachable from the starting 
%state the corresponding state $t'$ is in $\overline{W}^*$ in $G\times \Mem$.
%Let $P^s_{\Mem}$ be a play in $\Mem$ that starts in some state $s$, when player~1 follows $\sigma_1'$ and player~2 follows $\sigma_2'$. Since $\sigma_1'$ is stationary and $\sigma_2'$ is positional, for any set of states $T$ in $G\times \Mem$, with probability~1, after a bounded number of steps, either $P^s_{\Mem}$ reaches a state in $T$ or a state $t'$ from which no state in $T$ can be reached. Hence, all states that are reached, are reached the first time after a bounded number of steps with probability~1.

\smallskip\noindent{\bf Upper bound on value ensured by $\sigma_1$.}
We show that given $\sigma_1$ and the counter-strategy $\sigma_2$ 
the mean-payoff value is at most $1-\frac{c^n}{n}$ for all 
starting states in $\ov{W}^*$.
The proof is split in the following cases, and the basic intuitive
arguments are as follows:
\begin{compactenum}
\item Consider a play that starts in $\ov{X}_{1,1}$.
We show that the play always stays in $\ov{X}_{1,1}$ and Case~(1) is 
satisfied always. 
Thus we show that from every state there is a path of length at 
most $n$  where reward~0 occurs at least once.

\item For a play that starts in $\ov{W}_1 \setminus \ov{X}_{1,1}$, we always
satisfy either Case~(1) or Case~(2). 
First we establish that the event of Case~(2) being satisfied infinitely 
often has probability~0.
Hence from some point on Case~(1) is always satisfied, and then the argument 
is similar to the previous case.

\item Finally we consider a play that starts in $\ov{W}^*\setminus \ov{W}_1$.
Whenever Case~(3) is satisfied, and if the current state is $\ov{W}_j$, for $j> 1$,
then $\ov{W}_{j-1}$ is reached with positive probability in one-step. 
We establish that either (i)~we are similar to the previous case 
or (ii)~reach $W$ or $\ov{W}_1$ and the probability to reach $\ov{W}_1$
is at least $c^n$.
\end{compactenum}
Intuitively, in the first two cases above, we reach a recurrent class
that consists of states satisfying Case~(1) only, and in such recurrent 
classes the mean-payoff value is at most $1-c^n$.
In the last case, either we reach a recurrent class of the above type,
or whenever we satisfy Case~(3) with positive probability $c>0$ 
we make progress to a recurrent class of the above type.
The above case analysis establish the proof. 
We now present the formal proof.

\begin{lemma}\label{lemm:final_val2}
Fix an arbitrary finite-memory strategy $\sigma_1$ and consider
the counter-strategy $\sigma_2$ given $\sigma_1$.  
For all states in $\ov{W}^*$ we have that $\Exp_s^{\sigma_1,\sigma_2}[\LimSupAvg] \leq 1-\frac{c^n}{n}$.
\end{lemma}
\begin{proof}
In game $G\times \Mem$, let $\ov{\CC}_i$ be the set of states where Case~$(i)$ 
is satisfied\footnote{Note that $\ov{\CC}_i\neq (S\setminus \CC_i)$, for $i\in \set{1,2,3}$, in general, 
where $\CC_i$ is the set defined in Subsection~\ref{sec:incl1}, but this notation is used because 
$\ov{\CC}_1,\ov{\CC}_2,\ov{\CC}_3$ serve similar roles for properties of $\ov{W}^*$ as $\CC_1,\CC_2,\CC_3$ 
did for properties of $W^*$}.
That is $\ov{\CC}_1,\ov{\CC}_2,$ and $\ov{\CC}_3$ satisfy 
Case~(1), Case~(2), and Case~(3), respectively. 
We consider the Markov chain given $\sigma_1$ and $\sigma_2$, and consider a 
play $P^s$ starting from state $s$.
We will consider three cases to establish the result.

\begin{enumerate}
%Consider the probability space $\omega$ induced by $\sigma_1$ and $\sigma_2$, when the play starts in $s\in \ov{W}^*$. Let $P^s$ be any play in $\omega$.  

\item \noindent{\bf Plays starting in $s\in \overline{X}_{1,1}$.} 
Recall that $\ov{X}_{1,1}$ is the complement of $X_{1,1}$.
Consider state $s$ in $\overline{Z}_{1,1,k}$, for some $k\geq 1$ 
(that is: states in $\overline{X}_{1,1}$). 
Since $W_0=X_{1,0}=S$, we have that the play corresponding to $P^s$ in $G\times \Mem$ 
is always in $\ov{\CC}_1$ (note that only in Case~(1) do we have probability~0 to
go to $\ov{W}_0$ and $\ov{X}_{1,0}$). 
Hence the play $P^s$ always stays in $\overline{X}_{1,1}$.
Hence, from states in $\overline{Z}_{1,1,k}$, if player~1 plays 
according to $\sigma_1$ and player~2 plays $\sigma_2$, with probability $c$ 
we either $(i)$~reach a state in $\overline{Z}_{1,1,k-1}$, or 
$(ii)$~get a reward of $0$. 
%In both cases we have probability~0 to reach either $W_1$ or $X_{1,1}$. 
Since $Z_{1,1,0}=S$ we must get a reward of $0$ with at least probability $c$ when in $\overline{Z}_{1,1,1}$. 
Hence, for all states in $\overline{X}_{1,1}$, given player~1 follows $\sigma_1$ and player~2 follows $\sigma_2$, 
there is a path of play of length at most $\ov{\rank}(1,1)>\ov{\rank}(1)$ 
where each step happens with probability at least $c$ and the reward~0 happens at 
least once. 
Thus, for any state $s$ in $\overline{X}_{1,1}$, the play $P^s$ stays in $\overline{X}_{1,1}$ 
and gives a expected average reward of at most $1-\frac{c^j}{j}$, with probability~1, where $j=\ov{\rank}(1)$.
In other words, we have established the following property: in the Markov chain
all recurrent classes that intersect with $(\ov{X}_{1,1}\times \Mem)$ 
are contained in  $(\ov{X}_{1,1}\times \Mem)$ and have mean-payoff at most 
$1-\frac{c^n}{n}$.  

\item \noindent{\bf Plays starting in $s\in (\overline{W}_1\setminus \overline{X}_{1,1})$.} 
Consider now state $s$ in $(\overline{W}_1\setminus \overline{X}_{1,1})$. 
Since $W_0=S$, we have that the play $P_{\Mem}^s$, corresponding to $P^s$ 
in $G\times \Mem$, is always in $(\ov{\CC}_1\cup \ov{\CC}_2)$ 
(note that in Case~(3) we have positive probability to goto $W_0$). This is the only property of $(\overline{W}_1\setminus \overline{X}_{1,1})$ we will use.
Notice that this ensures that $P^s$ always stays in $\overline{W}_1$. 
Let $\RS$ be the set of states  from which no state in $\ov{\CC}_2$ can be reached. 
There are now two cases, either $P_{\Mem}^s$ reaches a state in $\RS$ or it does not.
\begin{itemize}
\item {\bf The play $P_{\Mem}^s$ reaches a state in $\RS$.} 
Let $j=\ov{\rank}(1)$.
Then the mean-payoff is at most $1-\frac{c^j}{j}$ after reaching $\RS$, 
by a argument similar to the one for states in $\overline{X}_{1,1}$. 
Therefore, in this case, the mean-payoff of $P^s$ is at most 
$1-\frac{c^j}{j}$, 
%because $\RS$ is reached in finite time with probability~1 (by reachability in a Markov chain), 
since the mean-payoff is independent of the finite-prefix.
%%as we have previously seen.
\item {\bf The play $P_{\Mem}^s$ does not reach a state in $\RS$.}
In this case, we must visit states in $\ov{\CC}_2$ infinitely often with 
probability~1, by Markov property~\ref{markov property RSinf}.
Whenever we are in a state $s'$ in $\ov{\CC}_2\cap 
((\overline{X}_{1,j} \times \Mem)\setminus (\overline{X}_{1,j-1} \times \Mem))$, 
we have probability at least $p\cdot \delta_{\min}$ to reach $(\overline{X}_{1,j-1}\times\Mem)$ in one-step 
where $\frac{1}{p}$ is the maximum patience of any distribution played by $\sigma_1$.
Whenever we are in a state $s'$ in $\ov{\CC}_1\cap ((\overline{X}_{1,j}\times \Mem)\setminus 
(\overline{X}_{1,j-1}\times \Mem))$, we have probability~0 to leave $((\overline{X}_{1,j}\times\Mem)\setminus (\overline{X}_{1,j-1}\times\Mem))$ 
in one-step.
Therefore we must reach $(\ov{X}_{1,1}\times \Mem)$ in a finite number of steps with probability~1 and 
from $(\ov{X}_{1,1}\times \Mem)$ we get a mean-payoff of at most $1-\frac{c^j}{j}$, where $j=\ov{\rank}(1)$, 
as we have already established  in the first item\footnote{In fact, alternatively we can prove this case using contradiction, 
since $(\ov{X}_{1,1}\times \Mem)\subseteq \ov{\CC}_1$ and therefore $(\ov{X}_{1,1}\times \Mem)\subseteq \RS$, since $(\ov{X}_{1,1}\times \Mem)$ cannot be left
in the Markov chain}.
%Also, by Markov property~\ref{markov property finite step reachability}, since we reach $(\ov{X}_{1,1}\times \Mem)$, we do so in a finite number of steps with probability~1. 
%Therefore, by ignoring the finite prefix (since mean-payoff is independent of finite
%prefixes) before reaching $(\ov{X}_{1,1}\times \Mem)$, we see that the mean-payoff is 
%at most $1-\frac{c^n}{n}$.
\end{itemize} 

Therefore, in both cases we get a mean-payoff of at most $1-\frac{c^j}{j}$ with probability~1, where $j=\ov{\rank}(1)$,
i.e., all recurrent classes have mean-payoff of at most $1-\frac{c^j}{j}$.

\item \noindent{\bf Plays starting in $s\in (\overline{W}^*\setminus \overline{W}_1)$.} 
Consider now state $s$ in $(\overline{W}^*\setminus \overline{W}_1)$. 
Consider the play $P^s$ in $G$ and the corresponding play $P^s_{\Mem}$ in $G\times \Mem$. 
For $i \geq 1$, let $L_i = W_{i} \cup \overline{W}_{i-1}$ and note that $\ov{L}_i= \overline{W}_{i} \setminus \overline{W}_{i-1}$.
Let $\ov{R}_i$ be the set of states in $\ov{L}_i$ from which no state in $\ov{\CC}_3 \cap \ov{L}_i$ 
is reachable; (note that $\ov{R}_i \subseteq \ov{L}_i \cap (\ov{\CC}_1 \cup \ov{\CC}_2)$).
Note that from $\ov{L}_i$, the set $\ov{L}_i$ can be left only from states in $\ov{\CC}_3 \cap \ov{L}_i$.
We now consider two sub-cases.
\begin{itemize}

\item We first consider the case where we reach $\ov{R}_i$.
Let $j=\ov{\rank}(i)$.
In this case, the mean-payoff is at most $1-\frac{c^j}{j}$ by an argument similar to the argument 
for $s$ in $\overline{W}_1\setminus \overline{X}_{1,1}$.
The argument for $s$ in $\overline{W}_1\setminus \overline{X}_{1,1}$ only uses that 
states in  $\ov{\CC}_1\cup \ov{\CC}_2$ are visited. 
Once $\ov{R}_i$ is reached we are guaranteed that only states in $\ov{R}_i$ are 
visited, and hence the recurrent classes in $\ov{R}_i$ has mean-payoff 
of at most  $1-\frac{c^n}{n}$.

\item If $\ov{R}_i$ is not reached, then since from every state $\ov{\CC}_3 \cap \ov{L}_i$ 
we have positive transition probability to $L_i$, it follows that $L_i$ is reached
with probability~1, by Markov property~\ref{markov property rs2new2}.
But if we reach either $W_{i}$ or $\overline{W}_{i-1}$, we have a probability of at least $c$ 
that it will be $\overline{W}_{i-1}$ (since it can only be done whenever $P^s_{\Mem}$ is in 
$\ov{\CC}_3 \cap L_i$, which ensures so). 

\end{itemize} 
Each time we repeat the second case, all states in $\ov{L}_i$, will never be visited again, in the worst case. Since each set $\ov{L}_i$ must contain atleast one state, we see that, if we repeat the second case $k$ times and thereafter enter $\ov{R}_{i'}$ (and are thus in the first case), then $n-k\geq \ov{\rank}(i')$. We have a probability of $c^k$ to follow such a play and we then get value at most $1-\frac{c^{n-k}}{n-k}$. Even if we got mean-payoff 1 with the remaining probability of $1-c^k$, we still have a expected mean-payoff of at most $1-\frac{c^n}{n-k}$. Thus, we see that in the worst case $k=0$ with probability 1, in which case we get mean-payoff at most $1-\frac{c^n}{n}$.
\end{enumerate}
%%Since $\sigma_1$ is an arbitrary finite-memory strategy, the desired result follows for all finite-memory strategies.
The desired result follows.
\end{proof}

Lemma~\ref{lemm:final_val2} implies the following inclusion.
\begin{lemma}\label{lemm:incl2}
We have $\val_1(\LimSupAvg(\cost),\bigstra_1^F)\subseteq W^*$.
\end{lemma}

\section{Improved Rank-Based Algorithm}
In this section we present an improved rank-based algorithm, which is 
based on the same principle as the small-progress measure algorithm~\cite{Jur00}
(for parity games).
While the naive computation of the $\mu$-calculus formula for the value~1 
set requires $O(n^4)$ iterations, the improved algorithm will require $O(n^2)$ 
iterations.

\smallskip\noindent{\bf Basic idea.}
The basic idea of the algorithm is to consider the ranking function $\rank$ from Section~\ref{sec:incl1} 
and use that to obtain an algorithm. 
Notice that $\rank(s)$ for $s\in W^*$ is always a pair $(i,j)$ such that $2\leq i+j\leq n+1$ and where $1\leq i,j \leq n$. 
We see that for any number $k$ there are $k-1$ pairs $(i,j)$ such that $i+j=k$ and such that $1\leq i,j\leq k-1$. 
Hence, there are $\sum_{k=1}^n k=\frac{n(n+1)}{2}$ such pairs $(i,j)$ such that $2\leq i+j\leq n+1$ and where $1\leq i,j\leq n$. 
Furthermore we also have a special rank $\top$ for not being in $W^*$. 
The ranks are lexicographically ordered as follows
\[
(1,1) < (1,2) < \dots < (1,n) < (2,1) < \dots < (n,1) < \top \enspace .
\]
We will thus say that $(i,j) < \top$ for all $i,j$ and 
$(i,j)<(i',j')$ if $i<i'$ or $i=i'$ and $j< j'$; 
(and for $(i,j)\leq(i',j')$ we change $j \leq j'$). 
To distinguish with the ranking function in Section~\ref{sec:incl1},
we denote the ranking function of the improved algorithm as 
$\rank'(s)$.

\smallskip\noindent{\bf Definition of matrix.}
Consider a given assignment of ranks to states.
Let $s$ be some state of rank $\rank'(s)\neq \top$ and therefore of rank $(i,j)$ for some $i$ and $j$; and 
also consider a state $s'$ of rank $(i',j')$. 
We define some sets, $U_s,Y_s,Z_s,X_s,W_s$ as follows: 
\begin{enumerate}
\item The state $s'$ is in $U_s$, if $i>i'$.
\item The state $s'$ is in $Y_s$, if $i>i'$ or $i'=i$ and $j>j'$.
\item The state $s'$ is in $Z_s$, if $i> i'$ or $i'=i$ and $j\geq j'$.
\item The state $s'$ is in $X_s$, if $i\geq i'$.
\item The state $s'$ is in $W_s$ independent of $s$.
\end{enumerate}
Also if a state $s''$ has rank $\top$, then it is in the set $\ov{W}_s$. 
This set also does not depend on $s$.
Let $M^s_{a_1,a_2}\in \set{\ov{W}_s,U_s,W_s,Y_s,X_s,Z_s^1,Z_s^0}$, for $a_1 \in \mov_1(s)$ and $a_2 \in \mov_2(s)$, 
be the matrix similar to the matrix $M$ from Section~\ref{sec:one-step}, 
except that instead of set $\ov{W}$ use $\ov{W}_s$ and similar for $U$, $Y$, $Z$, $X$ and $W$. 

\smallskip\noindent{\bf The \algo\ algorithm.}
We will refer to our algorithm as \algo\ and the description is as follows:
\begin{enumerate} 
\item For each state $s$ set $\rank'(s)\leftarrow (1,1)$
\item Let $i\leftarrow 0$ and $S^0\leftarrow S$.
\item (Iteration) While $S^i$ is not the empty set:
\begin{enumerate}
\item Let $Q^{i} = S^i\cup \set{s \mid \exists a_1 \in \mov_1(s), \exists a_2 \in \mov_2(s).\ \dest(s,a_1,a_2) \cap S^i \neq \emptyset}$
be the set of states in $S^{i}$ and their predecessors.
\item For each state $s \in Q^{i}$ such that $\rank'(s)\neq \top$, run \algopred\ on $M^{s}$ 
(if $M^s$ has not changed since the last time \algopred\ was run on $M^s$, then use the result from the last time instead of rerunning \algopred). 
Let $S^{i+1}$ be the set of states which \algopred\ rejected. 
\item Increment the rank (according to the lexicographic ordering) of all states in $S^{i+1}$.
\item Let $i\leftarrow i+1$.
\end{enumerate}
\item Return the set of states which does not have rank $\top$.
\end{enumerate}

\subsection{Running time of algorithm \algo}
We now analyze the running time of the algorithm.
We first analyze the work done for updating matrices $M^s$ and 
then analyze the work done for \algopred\ computation.
%a state, given a particular 
%rank, and then analyze the overall running time.
\begin{itemize}
\item \emph{Work to update matrix.} 
For a state $s$ of rank $(i,j)$, notice that we do not need to recalculate the entire $M^s$ whenever some successor $s'$ of $s$ changes rank, 
but only the entries $(a_1,a_2)$ such that $s'\in \dest(s,a_1,a_2)$. 
Also notice that we do not need to change $M^s$ at all whenever $s'$ changes rank to ranks other than in  $\set{(i,1),(i,j),(i,j+1),(i+1,1),\top}$. 
Hence, as long as $s$ has some rank $(i,j)$, we can do all updates of $M^s$ in time $O(\sum_{a\in \mov_1(s),b\in \mov_2(s)}|\supp(s,a,b)|)$.
We also recalculate $M^s$ whenever $s$ changes rank, and since each state has at most $O(n^2)$ different ranks 
therefore we use $O(n^2 \cdot \sum_{s\in S}\sum_{a\in \mov_1(s),b\in \mov_2(s)}|\supp(s,a,b)|)$ 
time to do all updates of $M^s$ for all states $s$. 
\item \emph{Work of \algopred.} 
%We use the same amount of time to calculate if a successor has changed rank to some rank in $\set{(i,1),(i,j),(i,j+1),(i+1,1),\top}$, since there are $|\supp(s,a,b)|$ possible successors of $s$ when player~1 plays $a$ and player~2 plays $b$. 
%%Each state has $O(n^2)$ different ranks. By summing over all we get the result.
Note that each entry of $M^s$ can take at most~7 different values, and as long as $s$ has a fixed rank 
each update makes some entry worse than before.
Hence as long as $s$ has some fixed rank $(i,j)$ we can do no more than $6\cdot |\mov_1(s)|\cdot |\mov_2(s)|$ updates of $M^s$.
% because each update makes some entry worse and each entry can only be made worse at most 5 times. 
Hence we run \algopred\ at most $\frac{n(n+1)}{2}\cdot 6\cdot |\mov_1(s)|\cdot |\mov_2(s)|$ times for a fixed $s$. 
\end{itemize}
Therefore, we get a total running time of $O(n^2\cdot \sum_{s\in S}(|\mov_1(s)|^3\cdot |\mov_2(s)|^3+\sum_{a_1\in \mov_1(s),a_2\in \mov_2(s)}|\supp(s,a_1,a_2)|))$, 
using Lemma~\ref{lemm: distribution algorithm time}.

\subsection{Proof of correctness of algorithm \algo}
The correctness proof is similar to the results of~\cite{Jur00}.
The proof of~\cite{Jur00} shows the equivalence of $\mu$-calculus
formula and a rank-based algorithm (called small-progress measure 
algorithm) for parity games; and the crucial argument of the correctness
was based on the fact that the predecessor operator is 
monotonic. 
Our correctness proof is similar and uses that $\EXP$ is monotonic.
We just present the proof of one inclusion and the other inclusion 
is similar.
For simplicity we will say that the rank of $s$ is $\rank(s)=\top$ 
if $s\in \ov{W}^*$. Let $\widetilde{W}^*$ be the output of the algorithm. 
We show that $\widetilde{W}^*=W^*$. %%To do so we will show that $\rank'(s)=\rank(s)$, thus the statement follows. 

\smallskip\noindent{\bf $\widetilde{W}^* \subseteq W^*: \rank'(s)\leq \rank(s)$.} 
We only need to show the statement for $\rank(s)\neq \top$ 
since otherwise the statement follows by definition. Hence,
assume towards contradiction that $\rank'(s)> \rank(s)$ and let  $\rank(s)=(i,j)$.
Also, we can WLOG assume that $s$ gets assigned a rank higher than $\rank(s)$ in the first iteration for which any state $s'$ gets assigned rank higher than $\rank(s')$ by the algorithm. Therefore in that iteration all states $s'$ are such that the rank assigned by the algorithm is at most $\rank(s')$ and $s$ has rank $\rank(s)$ assigned. Therefore $W^*\subseteq W_s$, $U_{i-1}\subseteq U_s$, $U_{i}\subseteq X_s$, $Y_{i,j-1}\subseteq Y_s$, $Y_{i,j}\subseteq Z_s$. But $s$ is in $\EXP(W^*,U_{i-1},U_{i},Y_{i,j-1},Y_{i,j})$ by definition since $s$ is such that $\rank(s)=(i,j)$. By monotonicity of $\EXP$ we have that 
$s$ is also in $\EXP(W_s,U_s,X_s,Y_s,Z_s)$, contradicting that $s$ changes rank.

\begin{lemma}\label{lemm:algo}
The algorithm \algo\ correctly computes the set $\val_1(\LimInfAvg(\cost),\bigstra_1^F)$ of states 
in time $O(n^2\cdot \sum_{s\in S}(|\mov_1(s)|^3\cdot |\mov_2(s)|^3+\sum_{a_1\in \mov_1(s),a_2\in \mov_2(s)}|\supp(s,a_1,a_2)|))$.
\end{lemma}

\section{Main result and Concluding Remarks}
We now summarize the main result, and conclude with an open question.

\begin{theorem}
The following assertions hold for concurrent mean-payoff games.
\begin{enumerate}

\item \emph{(Value~1 set characterization).}
Let $W^*= \nu W. \mu U. \nu X. \mu Y. \nu Z. \EXP(W,U,X,Y,Z)$, then  
we have 
\begin{align*}
W^* &=\val_1(\LimSupAvg(\cost),\bigstra_1^S) = 
\val_1(\LimSupAvg(\cost),\bigstra_1^F) \\[1.5ex]
&=\val_1(\LimInfAvg(\cost),\bigstra_1^S) =
\val_1(\LimInfAvg(\cost),\bigstra_1^F)
\end{align*}

\item \emph{(Running time).} The value~1 sets 
$\val_1(\LimSupAvg(\cost),\bigstra_1^S)=\val_1(\LimSupAvg(\cost),\bigstra_1^F)$
can be computed in time $O(n^2\cdot\sum_{s\in S}(|\mov_1(s)|^3\cdot |\mov_2(s)|^3+\sum_{a_1\in \mov_1(s),a_2\in \mov_2(s)}|\supp(s,a_1,a_2)|))$.

\item \emph{(Optimal patience).}
For all $\epsilon>0$, there exist stationary $\epsilon$-optimal strategies 
in the set $\val_1(\LimSupAvg(\cost),\bigstra_1^S)$ with patience at most 
$\left(\frac{\epsilon\cdot\delta_{\min}}{4}\right)^{-(2m)^{n}}$.
\end{enumerate}
\end{theorem}
\begin{proof}
The first item follows from Lemma~\ref{lemm:incl2} together with 
Lemma~\ref{lemm:incl1}.
The second item comes from Lemma~\ref{lemm:algo}.
The third item follows from Lemma~\ref{lemm:patience2}.
\end{proof}

Notice that the patience closely matches the patience obtained for the concurrent reachability game Purgatory, by Hansen, Ibsen-Jensen and Miltersen~\cite[Theorem~10]{HIM11} (the bound for $m=2$ is also in~\cite{HKM09}). Concurrent reachability games is a subclass of concurrent mean-payoff games and always have $\epsilon$-optimal stationary strategies, for all $\epsilon>0$, and all states in Purgatory have value~1. Thus the example provides a closely matching lower bound for patience.

\smallskip\noindent{\bf Robustness.} 
Our results show that the value~1 set computation can be achieved
by an iterative algorithm with the $\EXP$ operator. 
Our algorithm for the $\EXP$ operator computation is based on the 
matrix construction $M$, and observe that the entries in the matrix
depends only on the support set, but not the precise probabilities.
It follows that given two concurrent games where the support sets of the 
transition functions match, but the precise transition probabilities may 
differ, the value~1 set remains unchanged.

\smallskip\noindent{\bf Concluding remarks.} 
In this work we considered concurrent mean-payoff games 
and presented a polynomial-time algorithm to compute the 
value~1 set for finite-memory strategies for player~1. 
An interesting open question is whether the value~1 set 
with infinite-memory strategies can also be computed in 
polynomial time.

\smallskip\noindent{\em Acknowledgement.} 
The research was partly supported by FWF Grant No P 23499-N23,  
FWF NFN Grant No S11407-N23 (RiSE), ERC Start grant (279307: Graph Games), 
and Microsoft faculty fellows award.

\bibliographystyle{plain}
\bibliography{diss}

\begin{thebibliography}{10}

\bibitem{BK76}
T.~Bewley and E.~Kohlberg.
\newblock The asymptotic behavior of stochastic games.
\newblock {\em Math. Op. Res.}, (1), 1976.

\bibitem{BF68}
D.~Blackwell and T.S. Ferguson.
\newblock The big match.
\newblock {\em AMS}, 39:159--163, 1968.

\bibitem{BBEKW10}
T.~Br{\'a}zdil, V.~Brozek, Kousha Etessami, A.~Kucera, and D.~Wojtczak.
\newblock One-counter markov decision processes.
\newblock In {\em SODA}, pages 863--874, 2010.

\bibitem{Cha07a}
K.~Chatterjee.
\newblock Concurrent games with tail objectives.
\newblock {\em Theor. Comput. Sci.}, 388(1-3):181--198, 2007.

\bibitem{CdAH11}
K.~Chatterjee, L.~de~Alfaro, and T.A. Henzinger.
\newblock Qualitative concurrent parity games.
\newblock {\em ACM ToCL}, 2011.

\bibitem{CMH08}
K.~Chatterjee, R.~Majumdar, and T.~A. Henzinger.
\newblock Stochastic limit-average games are in exptime.
\newblock {\em Int. J. Game Theory}, 37(2):219--234, 2008.

\bibitem{CT12}
K.~Chatterjee and M.~Tracol.
\newblock Decidable problems for probabilistic automata on infinite words.
\newblock In {\em LICS}, pages 185--194, 2012.

\bibitem{Cha14}
Krishnendu Chatterjee.
\newblock Qualitative concurrent parity games: Bounded rationality.
\newblock In {\em {CONCUR} 2014 - Concurrency Theory - 25th International
  Conference, {CONCUR} 2014, Rome, Italy, September 2-5, 2014. Proceedings},
  pages 544--559, 2014.

\bibitem{CdAH08}
Krishnendu Chatterjee, Luca de~Alfaro, and Thomas~A. Henzinger.
\newblock Qualitative concurrent parity games, 2008.

\bibitem{CGHIKPS08}
Krishnendu Chatterjee, Arkadeb Ghosal, Thomas~A. Henzinger, Daniel~T. Iercan,
  Christoph~M. Kirsch, Claudio Pinello, and Alberto~L.
  Sangiovanni{-}Vincentelli.
\newblock Logical reliability of interacting real-time tasks.
\newblock In {\em Design, Automation and Test in Europe, {DATE} 2008, Munich,
  Germany, March 10-14, 2008}, pages 909--914, 2008.

\bibitem{CI14a}
Krishnendu Chatterjee and Rasmus Ibsen-Jensen.
\newblock Qualitative analysis of concurrent mean-payoff games,
  arxiv:1409.5306, 2014.

\bibitem{Con92}
A.~Condon.
\newblock The complexity of stochastic games.
\newblock {\em I\&C}, 96(2):203--224, 1992.

\bibitem{dAHK98}
L.~de~Alfaro, T.A. Henzinger, and O.~Kupferman.
\newblock Concurrent reachability games.
\newblock In {\em FOCS'98}, pages 564--575. IEEE, 1998.

\bibitem{EM79}
A.~Ehrenfeucht and J.~Mycielski.
\newblock Positional strategies for mean payoff games.
\newblock {\em Int. Journal of Game Theory}, 8(2):109--113, 1979.

\bibitem{EY06}
K.~Etessami and M.~Yannakakis.
\newblock Recursive concurrent stochastic games.
\newblock In {\em ICALP'06 (2)}, LNCS 4052, Springer, pages 324--335, 2006.

\bibitem{Eve57}
H.~Everett.
\newblock Recursive games.
\newblock In {\em CTG}, volume~39 of {\em AMS}, pages 47--78, 1957.

\bibitem{FGO12}
N.~Fijalkow, H.~Gimbert, and Y.~Oualhadj.
\newblock Deciding the value 1 problem for probabilistic leaktight automata.
\newblock In {\em LICS}, pages 295--304, 2012.

\bibitem{FV97}
J.~Filar and K.~Vrieze.
\newblock {\em Competitive {Markov} Decision Processes}.
\newblock Springer-Verlag, 1997.

\bibitem{Gil57}
D.~Gillette.
\newblock Stochastic games with zero stop probabilitites.
\newblock In {\em CTG}, pages 179--188. Princeton University Press, 1957.

\bibitem{HIM11}
K.~A. Hansen, R.~Ibsen-Jensen, and P.~B. Miltersen.
\newblock The complexity of solving reachability games using value and strategy
  iteration.
\newblock In {\em CSR}, pages 77--90, 2011.

\bibitem{HKLMT11}
K.~A. Hansen, M.~Kouck{\'y}, N.~Lauritzen, P.~B. Miltersen, and E.~P.
  Tsigaridas.
\newblock Exact algorithms for solving stochastic games: extended abstract.
\newblock In {\em STOC}, pages 205--214, 2011.

\bibitem{HKM09}
K.~A. Hansen, M.~Kouck{\'y}, and P.~B. Miltersen.
\newblock Winning concurrent reachability games requires doubly-exponential
  patience.
\newblock In {\em LICS}, pages 332--341, 2009.

\bibitem{RasmusThesis}
R.~Ibsen-Jensen.
\newblock {\em Strategy complexity of two-player, zero-sum games}.
\newblock PhD thesis, Aarhus University, 2013.

\bibitem{Jur00}
M.~Jurdzinski.
\newblock Small progress measures for solving parity games.
\newblock In {\em STACS'00}, pages 290--301. LNCS 1770, Springer, 2000.

\bibitem{MN81}
J.F. Mertens and A.~Neyman.
\newblock Stochastic games.
\newblock {\em Int. J. Game Theory}, 10:53--66, 1981.

\bibitem{Sha53}
L.S. Shapley.
\newblock Stochastic games.
\newblock {\em PNAS}, 39:1095--1100, 1953.

\bibitem{VardiP85}
M.Y. Vardi.
\newblock Automatic verification of probabilistic concurrent finite-state
  systems.
\newblock In {\em FOCS'85}, pages 327--338. IEEE Computer Society Press, 1985.

\bibitem{ZP96}
U.~Zwick and M.~Paterson.
\newblock The complexity of mean payoff games on graphs.
\newblock {\em Theoretical Computer Science}, 158:343--359, 1996.

\end{thebibliography}
\pagebreak
\section{Appendix --- Expanded mu-calculus formula\label{sec:Appendix}}

\smallskip\noindent{\bf Description of algorithm.}
Note that we established that if  
\[
W^*=\nu W. \mu U. \nu X. \mu Y. \nu Z. \EXP(W,U,X,Y,Z);
\]
then $W^*= \set{s \in S \mid \val(\LimInfAvg(\cost),\bigstra_1^F)(s) =1}$.
The $\mu$-calculus formula is a very succinct description of an 
algorithm.
The expanded iterative algorithm is presented as Algorithm~\ref{algo:mu-calc}.

\begin{algorithm}%[H]
\small
%\SetAlgoNoLine
%%\DontPrintSemicolon
\caption{Naive $\mu$-calculus Algorithm}\label{algo:mu-calc}
\KwIn{A concurrent mean-payoff game $G$ over the set of states $S$}
\KwOut{The set of states $W^*$} %%=\set{s \in S \mid \val(\LimInfAvg(\cost),\bigstra_1^F)(s) =1}$}
\BlankLine

$W\leftarrow S$\\
\Repeat{$W=W'$}{
 $W'\leftarrow W$\\
 $U\leftarrow \emptyset$\\
\Repeat{$U=U'$}{
 $U'\leftarrow U$\\
 $X\leftarrow W$\\
\Repeat{$X=X'$}{
 $X'\leftarrow X$\\
 $Y\leftarrow U$\\
\Repeat{$Y=Y'$}{
 $Y'\leftarrow Y$\\
 $Z\leftarrow X$\\
\Repeat{$Z=Z'$}{
 $Z'\leftarrow Z$\\
 $Z\!\leftarrow\!\text{\algopred}(W,U,X,Y,Z)$\\
}
 $Y\leftarrow Z$\\
}
 $X\leftarrow Y$\\
}
 $U\leftarrow X$\\
}
 $W\leftarrow U$\\
}
\Return $W$
\end{algorithm}

\pagebreak
\clearpage
\section{Technical appendix --- Computation of $\LimitReach$} 
We now present the details of the computation of 
%For a given CMPG $G$, we will in this section consider the computation of 
$\LimitReach(s,W,U,A_1,A_2)$. 
We will establish the Reject property and Accept properties~a---d of 
$\LimitReach$.
We first recall the properties:

\smallskip\noindent\emph{(Accept properties of $\LimitReach$).}
Accepts and returns the set $A_3 \subseteq A_2$ and 
a parametrized distribution $\dist_1^\epsilon$, for $0<\epsilon<\frac{1}{2}$, with support $\supp(\dist_1^\epsilon)\subseteq A_1$, such that the following properties
hold:
\begin{itemize}
\item (Accept property~a). For all $a_2 \in A_3$, the distribution 
$\dist_1^\epsilon$ satisfies Equation~\ref{exp:limit reach} for $a_2$.

\item (Accept property~b). For all $a_2 \in (A_2 \setminus A_3)$, 
we have $\dest(s,\dist_1^\epsilon,a_2) \cap \ov{W}=\emptyset$ and 
$\dest(s,\dist_1^\epsilon,a_2) \cap U=\emptyset$.

\item (Accept property~c). For all 
$a_1 \in (A_1\setminus \supp(\dist_1^\epsilon))$, there exists an action $a_2$ 
in $(A_2 \setminus A_3)$ such that $\dest(s,a_1,a_2)\cap\ov{W}\neq \emptyset$.

\item (Accept property~d). The set $A_3$ is largest in the sense that for all
$a_2 \in (A_2 \setminus A_3)$ and for all parametrized distributions $\dist_1^{\epsilon}$ 
over $A_1$, the Equation~\ref{exp:limit reach} cannot be satisfied,
while satisfying actions in $A_2$ using Equation~\ref{exp:limit reach}, 
or Equation~\ref{exp:almost reach}, or Equation~\ref{exp:get 1}, for any $X,Y,Z$ 
such that $U \subseteq Y \subseteq Z \subseteq X \subseteq W$. 

\end{itemize}

The computation of $\LimitReach(s,W,U,A_1,A_2)$ will be done similar to the 
computation of the similar named $\LimitReach(s,W,U)$ in~\cite{dAHK98,CdAH08}, and 
we will follow notations from~\cite{CdAH08}. 
We will use the two methods $\stay$ and $\cover$, defined as follows:
\begin{align*}
\stay(s,W,A_1,A_2,A) & = \set{a_1\in A_1\mid \forall a_2\in (A_2 \setminus A). 
\bigl[ (\dest(s,a_1,a_2) \cap \ov{W}) = \emptyset \bigr]}  \\[2ex]
\cover(s,U,A_1,A_2,A)& = \set{a_2\in A_2\mid \exists a_1\in (A_1\cap A). \bigl[ (\dest(s,a_1,a_2) \cap U)\neq \emptyset \bigr]}
\end{align*}

The algorithm $\LimitReach(s,W,U,A_1,A_2)$ is then as follows:
\begin{enumerate}
\item Let $A^*\leftarrow \mu A. \bigl[ \stay(s,W,A_1,A_2,A) \cup 
\cover(s,U,A_1,A_2,A) \bigr]$ 
and for all $a_1\in (A^*\cap A_1)$ let $\ell(a_1)$ be the level of $a_1$ 
in the formula.
\item If $(A^*\cap A_1)$ is empty, return reject. 
Otherwise, return accept and $(A^*\cap A_2,\dist_1^\epsilon)$, where 
$\dist_1^\epsilon$ is the parametrized distribution, with support 
$(A^*\cap A_1)$, and the ranking function of $a_1\in (A^*\cap A_1)$ 
is $\frac{\ell(a_1)-1}{2}$.
\end{enumerate}

The algorithm for $\LimitReach(s,W,U)$ of~\cite{dAHK98,CdAH08} can be obtained
as a special case of our description above
as follows: 
\begin{enumerate}
\item Let $(A_3,\dist_1^\epsilon)\leftarrow \LimitReach(s,W,U,\mov_1(s),\mov_2(s))$. 
If either (i)~$\LimitReach(s,W,U,\mov_1(s),\mov_2(s))$ rejects; or 
(ii)~$A_3\neq \mov_2(s)$, then return reject, otherwise return accept 
and $\dist_1^\epsilon$.
\end{enumerate}

We will now show that $\LimitReach(s,W,U,A_1,A_2)$ satisfies the desired
properties.

\begin{lemma}
The algorithm $\LimitReach(s,W,U,A_1,A_2)$ satisfies the Reject property of $\LimitReach$ and Accept properties a---d. 
Also, the patience of $\dist_1^\epsilon$ is at most $\left(\frac{\epsilon\cdot \delta_{\min}}{2}\right)^{|A_1|-1}$.
\end{lemma}
\begin{proof}
We establish the desired properties.

\smallskip\noindent{\bf The reject property of $\LimitReach$.} 
We see that $\LimitReach(s,W,U,A_1,A_2)$ only rejects if $(A^*\cap A_1)$ is empty. 
By definition of $\stay(s,W,A_1,A_2,A)$ we have $(A^*\cap A_1)$ is empty iff for all 
$a_1\in A_1$ there exists $a_2\in (A_2\setminus A^*)$ such that $(\dest(s,a_1,a_2)\cap \ov{W})\neq \emptyset$. 
We also see the reverse, since we see that also $(A_2\cap A^*)$ is empty if $(A^*\cap A_1)$ is empty 
by definition of $\cover(s,U,A_1,A_2,A)$. This implies that the empty set is a fixpoint of 
$\mu A. \bigl[ \stay(s,W,A_1,A_2,A) \cup \cover(s,U,A_1,A_2,A) \bigr]$ and 
thus must be $A^*$.
Since $A^*$ is empty, it follows that for all $a_1\in A_1$ there exists $a_2\in (A_2\setminus A^*)=A_2$ 
such that $(\dest(s,a_1,a_2)\cap \ov{W})\neq \emptyset$. 
Hence, if $\LimitReach(s,W,U,A_1,A_2)$ rejects, then the reject property of $\LimitReach$ is satisfied.

\smallskip\noindent{\bf Properties of the set $A^*$.} 
We have that if $\LimitReach(s,W,U,A_1,A_2)$ returns $(A_3,\dist_1^\epsilon)$, then  $A^*=(\supp(\dist_1^\epsilon)\cup A_3)$ and $A^*$ is a fixpoint of 
$\mu A. \bigl[ \stay(s,W,A_1,A_2,A) \cup \cover(s,U,A_1,A_2,A) \bigr]$.

\smallskip\noindent{\bf Accept property a.} 
We note that if we restrict the set of actions of player~1 to $A^* \cap A_1$ and actions of 
player~2 to $A_3$, then $\LimitReach(s,W,U)$ would return accept and the same parametrized
distribution, and then the proof of~\cite[Lemma~4]{CdAH08} ensures Accept property~a and
the desired patience.

\smallskip\noindent{\bf Accept property b.}  
We see that for an action $a_1$ to be in $(A^*\cap A_1)=\supp(\dist_1^\epsilon)$, 
by definition of $\stay(s,W,A_1,A_2,A^*)$, for all $a_2$ in $(A^* \cap A_2)=A_3$ we have that  
$(\dest(s,a_1,a_2) \cap \ov{W}) = \emptyset$ (or equivalently that $(\dest(s,\dist_1^\epsilon,a_2) \cap \ov{W}) = \emptyset$). 
This establishes the first half of Accept property~b. Also, we see that if an an action $a_2$ is in $(A_2\setminus A^*)=(A_2\setminus A_3)$, then by definition of $\cover(s,U,A_1,A_2,A^*)$ 
for all $a_1$ in $(A^*\cap A_1)=\supp(\dist_1^\epsilon)$ we have that $(\dest(s,a_1,a_2)\cap U)=\emptyset$ (or equivalently that $(\dest(s,\dist_1^\epsilon,a_2) \cap U) = \emptyset$). 
This establishes the second half of Accept property~b.

\smallskip\noindent{\bf Accept property c.}  For $A^*$ to be a fixpoint we must have, by definition of $\stay(s,W,A_1,A_2,A^*)$, that for each action $a_1\in (A_1 \setminus A^*)=(A_1\setminus \supp(\dist_1^\epsilon))$ 
that the condition to be in $\stay(s,W,A_1,A_2,A^*)$ must be violated and thus, there exists $a_2\in (A_2\setminus A^*)=(A_2\setminus A_3)$ 
such that $(\dest(s,a_1,a_2) \cap \ov{W}) \neq \emptyset$. 
This establishes Accept property~c.

\smallskip\noindent{\bf Accept property d.}
Along with $U$ and $W$ consider any $X,Y,Z$ such that $U \subseteq Y \subseteq Z \subseteq X \subseteq W$.
Consider a real number $0<\epsilon<\frac{\delta_{\min}}{|A_1|}$ and a distribution 
$\dist_1$ over $A_1$. 
We will show that if Equation~\ref{exp:limit reach} is satisfied by $\dist_1$ for some action 
$a_2 \in (A_2 \setminus A_3)$, then there is some action $a_2'\in A_2$ which is not satisfied by either (i)~Equation~\ref{exp:limit reach}; or (ii)~Equation~\ref{exp:almost reach}; or (iii)~Equation~\ref{exp:get 1}. 
The proof will be by contradiction and assume towards contradiction that such an action $a_2$ exists. 
Let $A_4\subseteq A_2$ be the set of actions which does satisfy Equation~\ref{exp:limit reach} by $\dist_1$ and let the remaining actions be satisfied by either Equation~\ref{exp:almost reach} or Equation~\ref{exp:get 1}. Notice that $A_4\not\subseteq A_3$, since $a_2\in A_4$ and $a_2\not\in A_3$. 

We consider two cases depending on whether or not $\supp(\dist_1)\subseteq\supp(\dist_1^\epsilon)$ 
to establish the result. 
\begin{itemize}
\item We first consider the case, where $\supp(\dist_1)\subseteq\supp(\dist_1^\epsilon)$. 
Then Equation~\ref{exp:limit reach} is violated for all $a_2'\in (A_2\setminus A_3)$, 
since $U$ cannot be reached by Accept property~b. In particular, it must be violated for $a_2$. 
That is a contradiction.

\item We next consider the case, where $\supp(\dist_1)\not \subseteq \supp(\dist_1^\epsilon)$.
Let $a_1\in (\supp(\dist_1)\setminus \supp(\dist_1^\epsilon))$ be an action, 
such that $a_1\in \arg\max_{a_1'\in (\supp(\dist_1)\setminus \supp(\dist_1^\epsilon))}\dist_1(a_1')$. 
By Accept property c, there exists an action $a_2' \in (A_2 \setminus A_3)$ such that $\dest(s,a_1,a_2')\cap \ov{W}\neq \emptyset$, 
since $a_1\in(\supp(\dist_1)\setminus \supp(\dist_1^\epsilon))\subseteq (A_1\setminus \supp(\dist_1^{\epsilon}))$.
We again split into two cases. Either $a_2'$ is in $A_4$ or not.
\begin{itemize}
\item We first consider the case then $a_2'\in A_4$. We will show that we go to $\ov{W}$ with too high probability, compared to the probability with which we go to $U$. We see that $\delta(s,\dist_1,a_2')(\ov{W})\geq \delta_{\min}\cdot \dist_1(a_1)$, by definition of $a_2'$. Each action $a_1'$ in $\supp(\dist_1^{\epsilon})$ ensures that $\dest(s,a_1',a_2')\cap U=\emptyset$ by Accept property~b, since $a_2'\not\in A_3$. It follows that $\delta(s,\dist_1,a_2')(U)\leq \dist_1(a_1) \cdot (|A_1|-1)$. This is because each action $a_1'$ such that $\dist(a_1')>\dist(a_1)$ are in $\supp(\dist_1^{\epsilon})$ by definition of $a_1$ and there are at most $|A_1|-1$ actions in $(\supp(\dist_1)\setminus\supp(\dist_1^{\epsilon}))$
(since $\dist_1$ and $\dist_1^\epsilon$ are distributions over $A_1$ and $|\supp(\dist_1^{\epsilon})|\geq 1$). But then $\delta(s,\dist_1,a_2')(U)\cdot \epsilon<\delta(s,\dist_1,a_2')(\ov{W})$ and 
thus Equation~\ref{exp:limit reach} is violated by $\dist_1$ and $a_2'$. This contradicts either  that $a_2'\in A_4$ or the definition of $A_4$.

\item We next consider the case then $a_2'\in (A_2\setminus A_4)$. Recall that $\dest(s,\dist_1,a_2') \cap \ov{W}\neq \emptyset$. 
Hence, Equation~\ref{exp:almost reach} and Equation~\ref{exp:get 1} are violated, since $\dest(s,\dist_1,a_2')\cap \ov{X} \neq \emptyset$ 
(because $X\subseteq W$ and if $\ov{W}$ is reached with positive probability, then $\ov{X}$ is reached with positive probability). 
Moreover, Equation~\ref{exp:limit reach} cannot be satisfied either, since $a_2'\not\in A_4$. Thus we have a contradiction.
\end{itemize}
\end{itemize}
Thus, in all cases we reach contradiction and, hence Accept property~d is satisfied.

The desired result follows.
\end{proof}

\end{document}